\documentclass[12pt]{article}

\usepackage[utf8]{inputenc}
\usepackage{csquotes}
\usepackage[american]{babel}
\usepackage[margin=0.85in]{geometry}
\usepackage{enumitem}
\usepackage{hyperref}
\usepackage{amsmath}
\usepackage{cleveref}
\usepackage{amssymb}
\usepackage{amsthm}
\usepackage{mathtools}
\usepackage{bbm}
\usepackage{accents}
\usepackage{booktabs}
\usepackage{graphicx}
\usepackage{comment}
\usepackage[title,page]{appendix}
\usepackage{xcolor}
\usepackage{authblk}
\usepackage{multirow}
\usepackage{multicol}
\usepackage{setspace}
\usepackage{float} 
\usepackage{proof-at-the-end}

\usepackage{centernot}
\usepackage{subcaption}
\usepackage{threeparttable}
\usepackage{rotating}
\usepackage{pdflscape}
\usepackage{enumitem}
\usepackage{titlesec} 
\usepackage{tikz}

\hypersetup{colorlinks, linkcolor = {blue}, citecolor = {blue}, urlcolor ={blue!80!black}}

\usetikzlibrary{arrows.meta, positioning, shapes.geometric}

\usepackage[authordate,backend=bibtex,natbib]{biblatex-chicago}
\newtheorem{theorem}{Theorem}

\newtheorem{lemma}{Lemma}

\newtheorem{proposition}{Proposition}
\newtheorem{assumption}{Assumption}

\theoremstyle{definition}
\newtheorem{definition}{Definition}
\newtheorem{remark}{Remark}
\newtheorem{example}{Example}

\newenvironment{continueexample}[1]
  {\renewcommand{\theexample}{\ref{#1} (continued)}\begin{example}}
  {\end{example}}

\newenvironment{myassump}[1]
{\renewcommand{\theassumption}{#1}\begin{assumption}}
  {\end{assumption}}

\DeclareMathOperator*{\argmin}{arg\,min}

\makeatletter
\@addtoreset{case}{proposition}
\makeatother

\makeatletter
\def\blfootnote{\gdef\@thefnmark{}\@footnotetext}
\makeatother

\makeatletter
\def\footnoterule{\kern-3\p@
  \hrule \@width 2in \kern 2.6\p@} 
\makeatother

\newcommand{\dtilde}[1]{\tilde{\raisebox{0pt}[0.85\height]{$\tilde{#1}$}}} 

\newcommand{\independent}{\perp\mkern-9.5mu\perp}  

\newcommand{\conv}{\xrightarrow{}}
\newcommand{\convp}{\xrightarrow{{\scriptscriptstyle \mathrm{p}}}}

\newcommand{\eqd}{\overset{d}{=}}

\graphicspath{{graphs/}}

\makeatletter
\def\input@path{{Tables/}}
\makeatother

\begin{document}
\author[1]{Filip Obradović\thanks{UCLA, Department of Economics. Email: \href{mailto:obradovicfilip@ucla.edu}{obradovicfilip@ucla.edu}
\\
I am deeply grateful to Charles Manski, Ivan Canay, and Federico Bugni for their guidance and support. I am also thankful to Nemanja Antić, Susan Athey, Eric Auerbach, Denis Chetverikov, Piotr Dworczak, Danil Fedchenko, Joel Horowitz, Matías Martinez, Jana Obradović, Andres Santos, Alexander Torgovitsky, and Gabriel Ziegler for valuable suggestions. I thank ICPSR for providing the data and the participants at the NBER Summer Institute, and seminars and conferences at Bocconi, Booth, Duke, Emory, Mannheim, Notre Dame, Michigan, Northwestern, Penn State, Princeton, UPenn, UCLA, USC, UVic, and Yale for comments. Financial support from the Robert Eisner Memorial Fellowship is gratefully acknowledged. 
}

}

\vspace{-1in}

\title{Identification of Long-Term Treatment Effects via Temporal Links, Observational, and Experimental Data}

\vspace{-1in}

\date{\today}

\maketitle

\begin{spacing}{1.2}
\begin{abstract}

Recent literature proposes combining short-term experimental and long-term observational data to provide alternatives to conventional observational studies for the identification of long-term average treatment effects (LTEs). This paper re-examines the identification problem and uncovers that assumptions restricting \textit{temporal link functions} -- relationships between short-term and mean long-term potential outcomes -- are central in this context. The experimental data serve to \textit{amplify} the identifying power of such assumptions; absent them, the combined data are no more informative than the observational data alone. Plausible inference thus hinges on justifiable restrictions in this class. Motivated by this, I introduce two \textit{treatment response} assumptions that may be defensible based on economic theory or intuition. To utilize them and facilitate future developments, I develop a novel unifying identification framework that computationally produces sharp bounds on the LTE for a general class of temporal link function restrictions and accommodates imperfect experimental compliance -- thereby also extending existing approaches. I illustrate the method by estimating the long-term effects of Head Start participation. The findings indicate that the effects on educational attainment, employment, and criminal involvement are lasting but smaller in magnitude than those established by sibling comparisons.
\end{abstract}
\end{spacing}

\section{Introduction}

Identifying long-term average treatment effects (LTEs) is an important goal in economics and other fields of science. For example, researchers may be interested in the effects of childhood interventions on earnings in adulthood; the impact of early-life conditional cash transfers on employment prospects; or the long-run adverse or protective effects of vaccination. LTEs are also often of interest in private sector research \citep{gupta2019top}.

Nevertheless, identifying LTEs is often challenging in practice. Long-term experimentation is frequently infeasible due to cost or institutional constraints.\footnote{Institutions supporting RCTs in development economics frequently require phase-in designs with staggered rollout of treatment to the whole sample. This limits follow-up for the control group.} Short-term experiments may be more accessible, but alone, they may not reveal outcomes of interest. As a result, researchers often rely on observational data \citep{currie2011human, hoynes2018safety}. However, observational studies critically rely on identifying assumptions that may be challenging to justify. 

This motivates a burgeoning strand of recent literature that seeks credible alternatives to conventional observational studies by combining (i) a long-run observational dataset with nonrandomized treatment assignment and (ii) a short-run experimental dataset in which long-run outcomes are unobserved \citep{ athey2024surrogate,athey2025combining}. However, follow-up work indicates that commonly-used modeling assumptions in this literature may also be challenging to justify in various economic settings, highlighting the need for alternative restrictions \citep{ghassami2022combining, van2023estimating, imbens2024long, park2024bracketing}.

This paper re-examines the identification problem by first characterizing which assumptions can exploit the specific data structure. For identification of the LTE, the experimental data serve only to potentially \textit{amplify} the identifying power of a well-defined class of modeling assumptions, restrictions on \textit{temporal link functions} -- means of long-term potential outcomes conditional on short-term potential outcomes. Absent such restrictions, the combined data are necessarily equally informative about the LTE as the observational data alone. When restrictions on temporal link functions are imposed, however, combining data can provide additional identifying power. Therefore, plausible inference that leverages the experimental data hinges on imposing justifiable restrictions of this type.

Building on this characterization, I introduce novel modeling assumptions within this class that may be justified by economic theory or intuition. Both impose shape restrictions on temporal link functions without constraining the treatment selection in any of the datasets, and therefore constitute treatment response assumptions. The first stipulates that temporal link functions are monotonic. Intuitively, this means that, absent selection into treatment, average long-term outcomes would be monotonic in short-term ones. Monotonicity of conditional means has been widely used in other settings since it may often be defensible on economic or intuitive grounds (see e.g. \citet{manski2000monotone, manski2009identification}, \citet{mogstad2018using}, \citet{torgovitsky2019nonparametric}). The second assumption postulates that the temporal link functions are invariant to treatment. Such invariance is implied by established mediation models, as in \citet{heckman2013understanding, garcia2020quantifying}, or statistical surrogacy (\citet{prentice1989surrogate}, \citet{athey2024surrogate}).

Finally, the paper develops a unifying identification framework that accommodates: (i) a general class of restrictions on temporal link functions and (ii) imperfect compliance in the experiment. The framework delivers sharp bounds on the LTE computationally, by solving optimization problems in which the maintained restrictions enter as constraints. The general sharp characterization is obtained by extending arguments of \citet{beresteanu2012partial} and \citet{chesher2017generalized} to address a new technical challenge -- jointly bounding distribution functions and conditional means of latent variables. 

The optimization-based formulation of the identified set enables direct implementation of the proposed restrictions. Moreover, it facilitates the development of new restrictions on temporal link functions by eliminating the need to derive bounds algebraically or to prove sharpness on a case-by-case basis. The framework also nests existing point-identification results and extends them to allow for imperfect compliance in the experiment, which is of substantial practical relevance. Building on \citet{shi2015simple} and this characterization, I propose a consistent criterion-based estimator for the bounds.

I illustrate the method by estimating the long-term effects of Head Start participation, the largest federally funded early childhood education program in the United States. To do so, I combine data from the Head Start Impact Study, a short-term experiment, and the Child and Young Adult Supplement to the National Longitudinal Survey of Youth 1979 cohort, a longitudinal survey. I find evidence of beneficial program impacts on educational and labor-market outcomes, as well as criminal involvement in adulthood. Head Start is estimated to increase the probability of high school graduation by $1.9$ to $3.2$ percentage points (pp), and decrease the probability of grade repetition by $1.1$ to $5.3$ pp. The program is also estimated to lower the probability of idleness (neither working nor in school) by $1.5$ to $4.6$ pp and criminal involvement by $1.2$ to $4.0$ pp. The results suggest that Head Start has lasting effects, though smaller in magnitude than reported by sibling comparison studies (\citet{deming2009early}). More broadly, they illustrate that the proposed assumptions can yield informative estimated bounds in applied work.

This paper is related to several strands of literature. It contributes to the recent body of work that combines long-term observational and short-term experimental data to identify long-term treatment effects (see also \citet{garcia2020quantifying}, \citet{dynarski2021closing}, \citet{hu2022identification}, \citet{chen2023semiparametric}, \citet{park2024informativeness}, \citet{Aizer2024}), and to the broader literature on data combination (\citet{cross2002regressions}, \citet{molinari2006generalization}, \citet{ridder2007econometrics}, \citet{fan2014identifying}, \citet{d2024partially}). In doing so, it draws on and extends results from the partial identification literature based on random set theory (\citet{galichon2011set}, \citet{beresteanu2012partial}, \citet{molchanov2014applications}, \citet{chesher2017generalized, chesher2020generalized}). The paper also contributes to work on general identification frameworks that operationalize broad classes of assumptions via optimization (\citet{mogstad2018using}, \citet{torgovitsky2019nonparametric}, \citet{russell2021sharp}, \citet{kamat2024identifying}). Finally, it is related to research that combines experimental data with economic theory to identify parameters beyond what the experimental data alone reveal (\citet{todd2006assessing}, \citet{attanasio2012education}, \citet{todd2023best}).

\Cref{sect:setup} introduces the setting. \Cref{sect:roles} details the role of restrictions on temporal link functions and experimental data. \Cref{sect:modeling_assumptions} introduces new restrictions on temporal link functions. \Cref{sect:identification} develops the identification framework. \Cref{sect:app} provides the empirical illustration. \Cref{sect:conclusion} concludes. Appendix \ref{sect:extensions} contains additional discussions; Appendix \ref{sect:proofs} proves the main results; the Supplemental Appendix discusses estimation and proves auxiliary results.

\section{Setting and Basic Assumptions}\label{sect:setup}

I formalize the problem using the standard potential outcomes model. Let $Y(d)\in\mathcal{Y}\subseteq\mathbb{R}$ and $S(d)\in{\mathcal{S}}\subseteq\mathbb{R}^{d_s}$ denote the long-term and short-term potential outcomes under some binary treatment $d\in\{0,1\}$, respectively. Denote the realized treatment by $D\in\{0,1\}$. The observed outcomes are:
\begin{align}
    \begin{split}
        Y &= DY(1)+(1-D)Y(0)\\
        S &= DS(1) +(1-D)S(0).
    \end{split}
\end{align}

Let $X\in\mathcal{X}\subseteq\mathbb{R}^{d_x}$ be a vector of observed covariates. Define the conditional \textit{long-term} average treatment effect (CLTE) $\tau(x)$:
\begin{align}
    \begin{split}
        \tau(x) &= E[Y(1)-Y(0)|X=x].
    \end{split}
\end{align}

The parameter of interest can be the CLTE itself or its weighted averages, such as the average \textit{long-term} treatment effect (LTE) $E[\tau(X)]$. I focus on the former for generality, noting that it is sufficient for identification of the latter when the weights are identified or given. 

\begin{example}{\textit{(Head Start Participation)}}\label{ex:head_start_basic}
In the empirical illustration, $D$ denotes Head Start participation, $S(d)$ is a vector of potential cognitive test scores in childhood, and $Y(d)$ are potential outcomes in adulthood, such as high school degree status or earnings, for treatment $d$.
\end{example}

\subsection{Observed Data}

As in prior work, I maintain that the researcher observes: 1) a short-term experimental dataset; and 2) a long-term observational dataset.\footnote{This setting is increasingly common. See also \citet{garcia2020quantifying}, \citet{ghassami2022combining}, \citet{hu2022identification},  \citet{van2023estimating}, \citet{chen2023semiparametric}, \citet{park2024bracketing, park2024informativeness}, \citet{Aizer2024}, \citet{imbens2024long},  \citet{athey2025combining}.}
The population is partitioned into two subpopulations that are randomly sampled to generate the two datasets. Let $G\in\{O,E\}$ denote the subpopulation indicator, where $G=O$ produces the observational and $G=E$ the experimental data.

Let $Z\in\mathcal{Z}$ denote an exogenous (i.e., randomly assigned) instrument in the experimental dataset that induces individuals into treatment. The identification analysis accommodates bounded $\mathcal{Z}$ with an arbitrary number of support points. Typically, $Z\in\{0,1\}$, representing random assignment to treatment or control groups, which may differ from the realized treatment $D$. More generally, $Z$ may have more than two support points, or even satisfy $Z\in[0,1]$ as in \citet{heckman1999local}. Since the binary case is predominant in practice, I adopt its terminology for expositional convenience. I refer to experiments with 
$P(D=Z|G=E)=1$ as having perfect compliance, and to the remaining cases as exhibiting imperfect compliance.

The short-term experimental dataset reveals $(S,D,X,Z)$, but not the long-term outcome $Y$. The long-term observational dataset reveals $(Y,S,D,X)$, but contains no exogenous instrument $Z$. Note that the data structure does not support the direct use of well-established instrumental variable methods for identification of \textit{long-term} treatment effects, since the outcome of interest $Y$ is never observed together with an instrument $Z$ (e.g. as in \citet{imbens1994identification}, \citet{heckman1999local} and \citet{mogstad2018using}).

\begin{continueexample}{ex:head_start_basic}
The observational dataset is the Child and Young Adult Supplement to the National Longitudinal Survey of Youth 79 Cohort (NLSY79) which reveals $(Y,S,D)$. The experimental dataset is the Head Start Impact Study (HSIS) which reveals $(S,D,Z)$, where $Z=1$ if the individual is assigned to participation in Head Start and $Z=0$ if assigned to non-participation. \citet{puma2010head} explain that some individuals may have $D\neq Z$.
\end{continueexample}

\begin{remark}
Introducing $Z$ in the experiment allows for (but does not require) imperfect compliance, which is of great practical relevance. In contrast, existing methods predominantly assume that $D$ is randomly assigned, which need not hold under imperfect compliance. The identification framework in \Cref{sect:identification} nests these methods and extends them to allow this possibility.
\end{remark}

I maintain the following assumptions throughout the paper.
\begin{myassump}{RA}{(Random Assignment)}
\label{ass:rand_assign}
    $Z\independent (Y(d),S(d))|X,G=E$ for $d\in\{0,1\}$.
\end{myassump}
 
\begin{myassump}{EV}\label{ass:ex_validity}{(Experimental External Validity)} 
$G\independent (Y(d),S(d))|X$ for $d\in\{0,1\}$.    
\end{myassump}

Assumption \ref{ass:rand_assign} is standard in the program evaluation literature. It is satisfied if $Z$ in the experimental data is randomly assigned, conditional on $X$. Assumption \ref{ass:ex_validity} is a key assumption in the data combination literature, linking the two datasets (\citet{ghassami2022combining}, \citet{chen2023semiparametric}, \citet{ park2024bracketing, park2024informativeness}, \citet{athey2025combining}). It states that, conditional on $X$, the subpopulations generating the datasets do not differ in their counterfactual distributions for each $d$. It may be plausible if two datasets are representative of the same population, conditional on $X$. For example, in the empirical illustration, HSIS and NLSY79 are designed to be representative of the U.S. population and pertain to the same treatment. More broadly, a growing body of empirical work relies on related links across datasets in this context; see, for example, \citet{garcia2020quantifying}, \citet{dynarski2021closing}, \citet{hu2022identification}, \citet{Aizer2024}.

It is worth emphasizing what is \textit{not} assumed. Maintained assumptions allow for $D\not\independent (Y(d),S(d))|X,G=g$ for any $g\in\{O,E\}$ and $d\in\{0,1\}$. This is expected in the observational dataset, and in the experimental data when compliance is imperfect. The assumptions also do not require $P(D=1|G=g)\in(0,1)$ for any $g\in\{O,E\}$. Instead, they allow $P(D=1|G=g)\in[0,1]$, which is relevant when a treatment is available only in one dataset, typically in the experiment. This is the case with some ``model'' early childhood intervention programs or novel vaccines.

Under Assumption \ref{ass:ex_validity}, CLTE is invariant to $G$, $E[Y(1)-Y(0)|X=x,G] = E[Y(1)-Y(0)|X=x] = \tau(x)$. Henceforth, I keep conditioning on $X$ implicit. The following analysis should be understood as conditional-on-$X$; I write the parameter of interest $\tau(x)$ as:
\begin{align}
    \tau = E[Y(1)-Y(0)]
\end{align}
and I continue referring to it as the LTE, with the understanding that it represents the CLTE.

\begin{remark}
Quasi-experimental datasets with 
$Z$ satisfying Assumptions \ref{ass:rand_assign} and \ref{ass:ex_validity} may also serve as the experimental dataset. I continue to refer to such data as experimental, following prevailing terminology in the literature.
\end{remark}

\noindent\textbf{Notation:} $\mathcal{H}(\theta)$ denotes the identified set for a parameter $\theta$. I write laws conditional on an event $\mathcal{E}$, $P(\cdot|\mathcal{E},G=g)$, as $P_g(\cdot|\mathcal{E})$ for $g\in\{O,E\}$. Whenever $P_E(\cdot|\mathcal{E}) = P_O(\cdot|\mathcal{E})$, I omit the subscript $g$. This is inherited by their features, e.g. $E_g[\cdot|\mathcal{E}] := E[\cdot|\mathcal{E},G=g]$. If necessary, I specify the random element using subscripts (e.g. $P_{S(d)}$ is the law of $S(d)$). 

\section{Role of Experimental Data}\label{sect:roles}

This section uncovers the role played by the experimental data in identifying $\tau$. Rather than removing the need for additional assumptions, experimental data serve to potentially \textit{amplify} the identifying power of a well-defined class of modeling restrictions. \Cref{sect:modeling_assumptions} leverages this insight to propose assumptions within this class that may be justifiable based on economic theory or intuition. \Cref{sect:identification} develops a general identification framework that enables tractable implementation of various assumptions in the class. Define for $s\in\mathcal{S}$ and $d\in\{0,1\}$:
\begin{align}
    m_d(s) := E[Y(d)|S(d)=s], \quad
    \gamma_d := P_{S(d)}.
\end{align}

\noindent where $P_{S(d)}$ is the marginal law of $S(d)$. I refer to $m_0(s)$ and $m_1(s)$ as \textit{temporal link functions}, since they ``link'' the short-term and long-term potential outcomes. Collect the two functions $m:=(m_0,m_1)\in\mathcal{M}$ and the marginal laws $\gamma := (\gamma_0,\gamma_1) = (P_{S(0)},P_{S(1)})$.\footnote{Let $(\Omega, \mathcal{F},P)$ denote the probability space and $\mathcal{B}$ the Borel $\sigma-$algebra. $\mathcal{M}$ is the set of Borel-measurable functions $\mu:\mathcal{S}\times\mathcal{S}\rightarrow \mathcal{Y}\times\mathcal{Y}$ such that $\mu\circ \varsigma$ is $P$-integrable for some $\mathcal{F}/\mathcal{B}(\mathcal{S}\times\mathcal{S})$-measurable function $\varsigma:\Omega\rightarrow\mathcal{S}\times\mathcal{S}$.} These objects are directly related to the parameter of interest:
\begin{align}\label{eq:identification_identity}
\begin{split}
     \tau &= E[Y(1)-Y(0)]= \int_\mathcal{S} m_1(s) d\gamma_1(s) - \int_\mathcal{S} m_0(s) d\gamma_0(s).
\end{split}
\end{align}

Consider the class of modeling assumptions defined by the following generic restriction.

\begin{myassump}{MA}{(Modeling Assumption)}\label{ass:modeling_generic}
    $m\in\mathcal{M}^A\subseteq\mathcal{M}$ for a known or identified set $\mathcal{M}^A$.
\end{myassump}

Let $\mathcal{H}(\tau)$ be the identified set for $\tau$, that is, the set of all values of $\tau$ compatible with both datasets and Assumptions \ref{ass:rand_assign}, \ref{ass:ex_validity}, and \ref{ass:modeling_generic}. Note that Assumption \ref{ass:modeling_generic} subsumes the trivial case $\mathcal{M}^A = \mathcal{M}$, corresponding to the absence of additional modeling assumptions. Denote by $\mathcal{H}^O(\tau)$ the identified set for $\tau$ under the same assumptions when \textit{only observational data} are used, and by $\subsetneq$ a \textit{strict} subset. The following proposition elucidates the relationship between $\mathcal{H}(\tau)$, $\mathcal{H}^O(\tau)$, and the modeling assumptions represented by $\mathcal{M}^A$.

\begin{theoremEnd}[proof at the end, no link to proof]{proposition}\label{prop:role_experiment}

Let Assumptions \ref{ass:rand_assign}, \ref{ass:ex_validity} hold, and suppose $\mathcal{H}(\tau)\neq\emptyset$. If $\mathcal{H}(\tau) \subsetneq \mathcal{H}^O(\tau)$, then $\mathcal{M}^A\subsetneq\mathcal{M}$. Equivalently, if $\mathcal{M}^A=\mathcal{M}$, then $\mathcal{H}(\tau)= \mathcal{H}^O(\tau)$.
\end{theoremEnd} 

\begin{proofEnd}

 If only Assumptions \ref{ass:rand_assign},  \ref{ass:ex_validity} hold, by \Cref{lem:marginals_y} $i)$, $\mathcal{H}(P_{Y(0)},P_{Y(1)}) = \mathcal{H}^O(P_{Y(0)},P_{Y(1)})$. Then, $\mathcal{H}(\tau) = \mathcal{H}^O(\tau)$, since $\tau$ is a functional of $(P_{Y(0)},P_{Y(1)})$. Therefore, if $\mathcal{H}(\tau) \subsetneq \mathcal{H}^O(\tau)$ and \ref{ass:rand_assign}, \ref{ass:ex_validity} hold, additional assumptions must be imposed. If $\mathcal{H}(\tau)\neq\emptyset$, these assumptions are not refuted. Henceforth, denote by $\mathcal{H}^A(\cdot)$ and $\mathcal{H}(\cdot)$ the identified sets with and without the additional assumptions. Let $\mathcal{H}^{A,O}(\cdot)$ be identified sets with additional assumptions, using only observational data.

Given that $\mathcal{H}^A(\tau) $ and $\mathcal{H}^{O,A}(\tau) $ are images of $T$ over $\mathcal{H}^A(m,\gamma)$ and $\mathcal{H}^{O,A}(m,\gamma) $, respectively, the additional assumptions must directly or indirectly restrict $m$, $\gamma$ or $(m,\gamma)$. Otherwise, $\mathcal{H}^O(m,\gamma) = \mathcal{H}^{A,O}(m,\gamma) $ and $\mathcal{H}(m,\gamma) = \mathcal{H}^{A}(m,\gamma)$, so  $\mathcal{H}^A(\tau) = \mathcal{H}(\tau) = \mathcal{H}^O(\tau) = \mathcal{H}^{A,O}(\tau)$ where the second equality is by \Cref{lem:marginals_y} $i)$. 
 
By way of contradiction, suppose that $\emptyset\neq \mathcal{H}^A(\tau)\subsetneq \mathcal{H}^{A,O}(\tau)$ and that only $\gamma$ is further restricted by the assumptions, or equivalently $\gamma\in\Gamma^A$, for some $\Gamma^A\subsetneq (\mathcal{P}^\mathcal{S})^2$. Then:
 \begin{align}\label{eq:intersection}
 \begin{split}
          \mathcal{H}^A(P_{Y(0)},P_{Y(1)},\gamma_0,\gamma_1) &=\mathcal{H}(P_{Y(0)},P_{Y(1)},\gamma_0,\gamma_1)\cap\left((\mathcal P^\mathcal{Y})^2\times\Gamma^A\right)\\
     &=\mathcal{H}(P_{Y(0)},P_{Y(1)})\times\left(\mathcal{H}(\gamma_0,\gamma_1)\cap\Gamma^A\right)
 \end{split}
 \end{align}

where the first equality is by the fact that only $\gamma$ is restricted by assumption, and the second is by \Cref{lem:marginals_y} $ii)$. It is then immediate that:
\begin{align*}
    \mathcal{H}^A(P_{Y(0)},P_{Y(1)}) =\mathcal{H}(P_{Y(0)},P_{Y(1)}) = \mathcal{H}^O(P_{Y(0)},P_{Y(1)})
\end{align*}
where the first equality is by \eqref{eq:intersection} and the definition of a projection, and the second is by \Cref{lem:marginals_y} $i)$. Thus, $\mathcal{H}^A(\tau)  = \mathcal{H}^O(\tau)$. But since $\mathcal{H}^A(\tau)\subseteq\mathcal{H}^{A,O}(\tau)\subseteq\mathcal{H}^O(\tau)$, then it must also be that $\mathcal{H}^A(\tau)=\mathcal{H}^{A,O}(\tau) $, contradicting $\mathcal{H}^A(\tau)\subsetneq \mathcal{H}^{A,O}(\tau)$. Therefore, the additional assumptions must restrict $m$ or $(m,\gamma)$. In either case, $m\in\mathcal{M}^A\subsetneq\mathcal{M}$.
\end{proofEnd}

\Cref{prop:role_experiment} shows that it is \textit{necessary} to restrict $m$ by assumption in order to gain identifying power from the experimental data. Equivalently, when the temporal link functions $m$ are left unrestricted, incorporating the experimental data provides \textit{no additional identifying power} for $\tau$. Once such restrictions are imposed, the inclusion of experimental data may yield additional identifying power, enabling $\mathcal{H}(\tau) \subsetneq \mathcal{H}^O(\tau)$. The proposition thus isolates the class of assumptions that could be made more informative about $\tau$ by the addition of experimental data. In this sense, adding experimental data \textit{may amplify} the identifying power of restrictions on $m$.

To gain intuition for the result, note that because $S(d)$ is revealed in the observational data whenever $D = d$, the experimental data can \textit{only} provide additional information about the distribution of $S(d)$ for those individuals with $D \neq d$. For these individuals, however, the data impose no restrictions on the relationship between $S(d)$ and the mean $Y(d)$. When this relationship is left unrestricted, additional information about the distribution of $S(d)$ does not translate into additional information about the average $Y(d)$, and thus about $\tau$. In contrast, when the relationship between $S(d)$ and the average $Y(d)$ is restricted by assumption, additional information about the distribution of $S(d)$ can yield further information about the mean of $Y(d)$ and hence about $\tau$.

\begin{remark}\label{rem:sufficient_conditions}
Existing approaches typically do not impose explicit restrictions on $m$. However, their identification arguments ultimately reduce to such restrictions implied by the maintained assumptions. These point-identification results are thus nested within the identification framework here. For example, \citet{athey2025combining} maintain \textit{latent unconfoundedness (LUC)} -- $Y(d)\independent D|S(d),G=O$ for $d\in\{0,1\}$, while the identification result uses its implication $m_d(s) = E_O[Y|S=s,D=d]$ for $d\in\{0,1\}$ and $s\in\mathcal{S}$. For \citet[Theorem 1]{imbens2024long}, let $S_t$ be subvectors of $S$ for $t\in\{1,2,3\}$. Maintained assumptions imply the restriction used for identification: $m_d(s_3,s_2) = h(s_3,s_2,d)$ where $h$ solves $E_O[Y | S_2, S_1, D] = E_O[h(S_3, S_2, D) | S_2, S_1, D]$.
\end{remark}

Two points are worth emphasizing. First, the addition of experimental data need not necessarily amplify the identifying power of restrictions on $m$. In particular, it is possible that $\mathcal{M}^A \subsetneq \mathcal{M}$ while $\mathcal{H}(\tau) = \mathcal{H}^O(\tau)$. Second, restrictions on $m$ may have identifying power for $\tau$ even in the absence of experimental data. Whether either case arises depends on the underlying data distribution and on the specific restriction imposed. These observations highlight that restrictions on $m$ are \textit{central} for identifying $\tau$ in this setting, while the experimental data play an auxiliary role by potentially amplifying their identifying power.

\Cref{prop:selection_assumptions_without_experiment} in Appendix \ref{sect:appendix_lemmas} illustrates these points using the latent unconfoundedness assumption defined in \Cref{rem:sufficient_conditions}, and frequently invoked in this context (see, for example, \citet{hu2022identification}, \citet{park2024informativeness}, and \citet{Aizer2024}). Specifically, the proposition shows that if the data distribution is such that $m_d(s) = E_O[Y \mid S = s, D = d]$ are constant in $s$, both the combined and observational data alone point identify $\tau$ under LUC. In this case, $\mathcal{M}^A \subsetneq \mathcal{M}$ but $\mathcal{H}(\tau) = \mathcal{H}^O(\tau)$, so the experimental data do not amplify the identifying power of LUC. The proposition further shows that for a broad class of empirically relevant data distributions, the observational-data bound $\mathcal{H}^O(\tau)$ under LUC is strictly more informative than the worst-case bounds of \citet{manski1990nonparametric} that would be sharp absent LUC.

\section{Modeling Assumptions}\label{sect:modeling_assumptions}

For the purpose of identifying $\tau$, the experimental data serve only to potentially amplify the identifying power of restrictions imposed on $m$. Consequently, plausible inference hinges on the plausibility of the restrictions within this class, despite the use of experimental data. However, commonly used assumptions implying restrictions on $m$ may be challenging to interpret or justify in economically relevant settings (see e.g. \citet{ghassami2022combining}, \citet{van2023estimating}, \citet{imbens2024long} and \citet{park2024bracketing}). Hence, I explore alternative restrictions within this class that may be defensible based on economic theory or intuition.

\begin{myassump}{LIV}{(Latent Monotone Instrumental Variables)}\label{ass:LMIV}
For every $m\in\mathcal M^A$ and each $d\in\{0,1\}$, $m_d(s)$ is nondecreasing in $s$ with respect to the product order. That is, for all $s,s'\in\mathcal S$ such that $s\leq s'$ (componentwise): $ m_d(s)\leq m_d(s')$.
\end{myassump}

Assumption \ref{ass:LMIV} states that the conditional means of $Y(d)$ are nondecreasing in any individual short-term potential outcome $S(d)$.\footnote{One can also assume that $m_d(s)$ is nonincreasing for any sub-collection of elements $\{s_j\}_{j\in\mathcal{J}}$ of $s$. Results follow directly by defining $\Tilde{S}(d)$ with components $\Tilde{S}_j(d) = \mathbbm{1}[j\in\mathcal J]\left(- S_j(d)\right)+\mathbbm{1}[j\not\in\mathcal J] S_j(d)$ and observing that $E[Y(d)|\Tilde{S}(d)=s]$ satisfies LIV.} Intuitively, this would be plausible if researchers are willing to maintain that long-term outcomes would be monotonic in short-term ones, \textit{absent selection into treatment}.

\setcounter{example}{1}
\begin{example}\label{ex:liv}{\textit{(LIV and Head Start)}}
Suppose first that $S(d)$ consists of a single childhood cognitive test score. Assumption \ref{ass:LMIV} then states that people with a higher \textit{potential} test score $S(d)$ have a weakly higher average of \textit{potential} high school degree status in adulthood $Y(d)$ than people with a lower \textit{potential} test score. In other words, under an exogenously fixed treatment, high school graduation \textit{rates} are nondecreasing in the childhood test score. When $S(d)$ comprises multiple test scores, the assumption asserts that the high school graduation rates are nondecreasing in each individual test score, holding the remaining ones constant.
\end{example}

LIV is related to the monotone instrumental variable (MIV) assumption of \citet{manski2000monotone, manski2009more}. MIV posits the existence of a variable $V\in\mathcal{V}$ such that $E[Y(d)|V=v]$ is nondecreasing in $v\in\mathcal{V}$, where $V$ is observed for \textit{all} individuals. The critical distinction is that the conditioning variables in Assumption \ref{ass:LMIV} are latent counterfactuals. This feature introduces additional complexity, which is addressed by the identification framework.

\begin{myassump}{TI}{(Treatment Invariance)}\label{ass:TI}
    For all $m\in\mathcal{M}^A$ and  $s\in\mathcal{S}$, $m_1(s) = m_0(s)$.
\end{myassump}

The assumption states that the relationship between the \textit{potential} short-term outcomes $S(d)$ and the mean of the long-term \textit{potential} outcomes $Y(d)$ does not vary with the underlying treatment $d$. Intuitively, it means that the treatment affects \textit{average} long-term outcomes only through short-term ones. Importantly, TI follows from theoretical models previously used in empirical research, as outlined by the following example. It is also implied by the statistical surrogacy assumption of \citet{prentice1989surrogate} $Y\independent D|S,G=E$ in the special case of perfect compliance.\footnote{See \Cref{lem:rel_to_surrogacy} \textit{i) in Appendix \ref{sect:appendix_lemmas}}. An important distinction relative to work relying on the surrogacy assumption in this literature is that the comparability assumption $G\independent Y|S$ is \textit{not} imposed here (e.g. see \citet{athey2024surrogate} and \citet{chen2023semiparametric}). Together, comparability and surrogacy imply a selection assumption on $m$ under Assumption \ref{ass:ex_validity}; see \Cref{lem:rel_to_surrogacy} \textit{vi)}.}

\begin{example}\label{ex:TI_head_start}{\textit{(TI and Head Start)}}
Consider the following separable model:
\begin{align}\label{eq:separable_model}
    Y(d) = \phi_d(S(d))+\varepsilon_d = \phi(S(d))+\varepsilon_d, \enskip \varepsilon_d\sim\varepsilon, \enskip\varepsilon_{d'}\independent S(d), \forall d,d'\in\{0,1\}
\end{align}
where $S(d)$ is a vector of short-term potential outcomes, including test scores and measures of non-cognitive skills. $S(d)$ are inputs in the production function $\phi_d$ for $Y(d)$. The production function $\phi_d$ and the distributions of unobservables $\varepsilon_d$ do not depend on Head Start participation $d$. Therefore, $E[Y(d)|S(d)=s] = \phi(s)+E[\varepsilon]$ which is invariant to $d$, so TI is implied by the model. Researchers may thus use TI whenever they find the model plausible. \citet{garcia2020quantifying} argue the plausibility of such a model in the context of identifying the long-term effects of an early childhood program.
\end{example}

Note that neither of the two assumptions imposes restrictions on selection, i.e., how individuals choose their treatment. As such, they represent \textit{treatment response assumptions}, whereas related preceding work predominantly relies on selection assumptions.

\section{Identification Framework}\label{sect:identification}

The preceding sections argue that restrictions on temporal link functions $m$ are central, and that experimental data can make them more informative. Here, I introduce a novel framework that enables the identification of $\tau$ under a general class of restrictions on $m$.

Define the functional $T:\mathcal{M}\times\mathcal{P}^\mathcal{S}\times\mathcal{P}^\mathcal{S}\rightarrow\Bar{\mathbb{R}}$, where $\mathcal{P}^\mathcal{S}$ is the set of distribution functions supported on $\mathcal{S}$:
\begin{align}\label{eq:identification_functional}
    T(m,\gamma) = \int_\mathcal{S} m_1(s) d\gamma_{1}(s) - \int_\mathcal{S} m_0(s) d\gamma_{0}(s).
\end{align}

Recalling that $E[Y(d)] =  \int_\mathcal{S} m_d(s) d\gamma_{d}(s)$ by \eqref{eq:identification_identity}, $T$ produces the corresponding value of $\tau$ for a given $(m,\gamma)$. One can then incorporate assumptions on $m$ to identify $\tau$ in two steps. First, find all $(m,\gamma)$ consistent with the data, Assumptions \ref{ass:rand_assign}, \ref{ass:ex_validity}, and any restriction \ref{ass:modeling_generic} on $m$. This yields their identified set $\mathcal{H}(m,\gamma)$. Then, the identified set $\mathcal{H}(\tau)$ is the set of all values of $\tau$ consistent with the feasible $(m,\gamma)$, or: 
\begin{align}\label{eq:identified_set_tau_definition}
    \begin{split}
        \mathcal{H}(\tau) &:= \Big\{T(m,\gamma): \text{$(m,\gamma)\in\mathcal{H}(m,\gamma)$} \Big\}.
    \end{split}
\end{align}

\subsection[Identifying (m,γ)]{Identifying $(m,\gamma)$}\label{sect:m_gamma}

This section shows how to represent all information about $(m,\gamma)$. $(m,\gamma)$ are features of \textit{latent} random variables $Y(d)$ and $S(d)$. I exploit this fact to construct $\mathcal{H}(m,\gamma)$ for a large class of restrictions on $m$. 

Because at least some potential outcomes are unobserved for each unit, the data and maintained assumptions are consistent with a set of \textit{random vectors} $(S(0),S(1),Y(0),Y(1))$. Intuitively, let $\mathcal{Q}$ be the set of all such random vectors that are compatible with the data, and Assumptions \ref{ass:rand_assign} and \ref{ass:ex_validity}. To obtain $\mathcal{H}(m,\gamma) $, it suffices to collect the corresponding pairs $(m,\gamma)$ that additionally satisfy the modeling restriction $m\in\mathcal{M}^A$. By definition:

\begin{align}\label{eq:id_set_intuitive}
    \mathcal{H}(m,\gamma) = \left\{(m,\gamma):\begin{array}{cc}
         \overbrace{\text{$m\in\mathcal{M}^A$, }}^{\text{\textit{Modeling assumption}}}\overbrace{\text{$\exists (S(0),S(1),Y(0),Y(1))\in\mathcal{Q}$}}^{\text{\textit{Data + Assumptions RA/EV}}}  \\
         \underbrace{\text{$\forall d\in\{0,1\}:$ $\gamma_d\eqd S(d)$, }\text{$m_d(S(d)) = E[Y(d)|S(d)]$ a.s.}}_{\text{\textit{$(m,\gamma)$ correspond to $S(d)$ and $Y(d)$}}}
    \end{array}\right\}.
\end{align}

Random set theory provides a convenient way to deliver a \textit{sharp} characterization of $(m,\gamma)$ based on this intuition. It first allows one to characterize $\mathcal{Q}$ and, consequently, $\mathcal{H}(m,\gamma)$. It then yields an \textit{equivalent} representation of \eqref{eq:id_set_intuitive} via moment inequalities that exhaust \textit{all} information contained in the data and the maintained assumptions. To formalize the argument, I introduce the necessary basic definitions specialized to finite-dimensional Euclidean spaces. I henceforth maintain that all random elements are defined on a non-atomic probability space $(\Omega, \mathcal{F},P)$.\footnote{That is, for any $A\in\mathcal{F}$ with positive measure there exists a measurable $B\subset A$ such that $0<P(B)<P(A)$.}

\noindent\textbf{Notation}: $B$ and $K$ denote sets. $\mathcal{K}(B)$, and $\mathcal{C}(B)$ denote the families of all compact, and closed subsets of the set $B$, respectively. $co(B)$ is the closed convex hull of the set $B$.

\begin{definition}
    A measurable map $\mathbf{R}:\Omega\rightarrow \mathcal{C}(\mathbb{R}^{d_R})$ is called a \textit{random (closed) set}.\footnote{$\mathbf{R}$ is measurable if for every compact set $K\in\mathcal{K}(\mathbb{R}^{d_R})$: $\{\omega\in\Omega: \mathbf{R}(\omega)\cap K\neq \emptyset\}\in\mathcal{F}$. The codomain $\mathcal{C}(\mathbb{R}^{d_R})$ is equipped with the $\sigma-$algebra generated by the families of sets $\{B\in \mathcal{C}(\mathbb{R}^{d_R}):B\cap K\neq \emptyset\}$ over $K\in\mathcal{K}(\mathbb{R}^{d_R})$.} 
\end{definition}

\begin{definition}
    A random vector $R:\Omega\rightarrow \mathbb{R}^{d_R}$ such that $R\in\mathbf{R}$ a.s. is called a \textit{(measurable) selection} of $\mathbf{R}$. $Sel(\mathbf{R})$ and $Sel^1(\mathbf{R})$ are the sets of all selections, and all integrable selections of $\mathbf{R}$, respectively. 
\end{definition}

Define the following closed random sets for $d\in\{0,1\}$:
\begin{align}
           \mathbf{Y}_d := \begin{cases}
            \{Y\}, \text{ if $(D,G)=(d,O)$}\\
            \mathcal{Y}, \text{ otherwise}
        \end{cases}, \enskip
        \mathbf{S}_d := \begin{cases}
            \{S\}, \text{ if $(D,G)\in\{(d,E),(d,O)\}$}\\
            \mathcal{S}, \text{ otherwise}
        \end{cases}.
\end{align}
By construction, random sets $\mathbf{Y}_d $ and $\mathbf{S}_d $ summarize all information about $Y(d)$ and $S(d)$ contained in \textit{the data}, respectively. As \citet{beresteanu2012partial} explain, \textit{all} information in the data about $(S(0),S(1),Y(0),Y(1))$ can be represented by stating that they can be any selection of the corresponding random sets. Additional assumptions, such as Assumptions \ref{ass:rand_assign} and \ref{ass:ex_validity}, may be imposed by further restricting the admissible selections, which are represented by the set $\mathcal{Q}$ in \eqref{eq:id_set_intuitive}. The following lemma formalizes $\mathcal{Q}$ and uses it to equivalently characterize $\mathcal{H}(m,\gamma)$.

\begin{theoremEnd}[proof at the end, no link to proof]{lemma}\label{thm:sharp_set}
    Let Assumptions \ref{ass:rand_assign},  \ref{ass:ex_validity}, and \ref{ass:modeling_generic} hold. The identified set for $(m,\gamma)$ is:
    \begin{align}\label{eq:joint_conditions}
    \begin{split}
    \mathcal{H}(m,\gamma) = \left\{\begin{array}{ll} (m,\gamma)\in\mathcal{M}^A\times (\mathcal{P}^{\mathcal{S}})^2:
          \forall d\in\{0,1\},  \enskip  \exists S(d)\in Sel(\mathbf{S}_d)\cap I,\\
     \exists Y(d)\in Sel(\mathbf{Y}_d), \enskip \gamma_d \eqd S(d),\enskip m_d(S(d)) = E_O[Y(d)|S(d)] \text{ a.s.}
\end{array}\right\}.
    \end{split}
    \end{align}
    where $I$ is the set of random elements $E_1\in\mathcal{S}$ such that $E_1\independent G$ and $E_1\independent Z|G=E$.
\end{theoremEnd}

\begin{proofEnd}
Recall by \eqref{eq:tildez} that $ \Tilde{Z} = \mathbbm{1}[G=E] Z+\mathbbm{1}[G=O](\sup\mathcal{Z}+1)\in \Tilde{\mathcal{Z}}$. Note that $\Tilde{Z} =Z$ when $G=E$ and $\Tilde{Z}$ equals a distinct constant when $G=O$. Therefore, Assumptions \ref{ass:rand_assign} and \ref{ass:ex_validity} hold if and only if $\Tilde{Z}\independent (Y(d),S(d))$ for all $d\in\{0,1\}$. Let $\Tilde{I}$ be the set of random elements $(E_1,E_2,E_3)$ such that $(E_1,E_2,E_3)\in \mathcal{Y}\times\mathcal{S}\times\Tilde{\mathcal{Z}}$ and $(E_1,E_2)\independent E_3$, i.e. that satisfy the two assumptions. Define the random set:
    \begin{align}
        (\mathbf{Y}_0,\mathbf{S}_0,\mathbf{Y}_1,\mathbf{S}_1) :=\begin{cases}
            \{(Y,S)\}\times\mathcal{Y}\times\mathcal{S}, \text{ if $(D,G)=(0,O)$}\\
            \mathcal{Y}\times\mathcal{S}\times\{(Y,S)\}, \text{ if $(D,G)=(1,O)$}\\
            \mathcal{Y}\times\{S\}\times\mathcal{Y}\times\mathcal{S}, \text{ if $(D,G)=(0,E)$}\\
            \mathcal{Y}\times\mathcal{S}\times\mathcal{Y}\times\{S\}, \text{ if $(D,G)=(1,E)$}\\
        \mathcal{Y}\times\mathcal{S}\times\mathcal{Y}\times\mathcal{S}, \text{ otherwise}
        \end{cases}.
    \end{align}

As in the proof of \citet[Proposition 2.3]{beresteanu2012partial}, by definition of $(\mathbf{Y}_0,\mathbf{S}_0,\mathbf{Y}_1,\mathbf{S}_1)$, all information on $(Y(0),S(0),Y(1),S(1))$ in the observed data can be summarized by $(Y(0),S(0),Y(1),S(1))\in Sel((\mathbf{Y}_0,\mathbf{S}_0,\mathbf{Y}_1,\mathbf{S}_1))$. Also, by definition of the random set, $(\mathbf{Y}_0,\mathbf{S}_0,\mathbf{Y}_1,\mathbf{S}_1) =(\mathbf{Y}_0,\mathbf{S}_0)\times(\mathbf{Y}_1,\mathbf{S}_1)$, where for $d\in\{0,1\}$:
    \begin{align}
        (\mathbf{Y}_d,\mathbf{S}_d) =\begin{cases}
            \{(Y,S)\}, \text{ if $(D,G)=(d,O)$}\\
                \mathcal{Y}\times\{S\}, \text{ if $(D,G)=(d,E)$}\\
            \mathcal{Y}\times\mathcal{S}, \text{ otherwise}
        \end{cases}.
    \end{align}

By applying \Cref{lem:rectangular_rset} twice, $Sel((\mathbf{Y}_0,\mathbf{S}_0,\mathbf{Y}_1,\mathbf{S}_1)) =Sel((\mathbf{Y}_0,\mathbf{S}_0))\times Sel((\mathbf{Y}_1,\mathbf{S}_1))=Sel(\mathbf{Y}_0)\times Sel(\mathbf{S}_0)\times Sel(\mathbf{Y}_1)\times Sel(\mathbf{S}_1)$. All information  about $(Y(0),S(0),Y(1),S(1))$ in the data, Assumptions \ref{ass:rand_assign} \ref{ass:ex_validity} and $E[|Y(d)|]<\infty$ for $d\in\{0,1\}$ can thus equivalently be expressed as $(Y(d),S(d),\Tilde{Z})\in Sel(\mathbf{Y}_d)\times Sel (\mathbf{S}_d)\times\{\Tilde{Z}\}\cap \Tilde{I}$ for $d\in\{0,1\}$. If Assumptions \ref{ass:rand_assign} and \ref{ass:ex_validity} hold, the intersection is non-empty. Let $\Bar{I}$ be the set of random elements $(E_1,E_2)$ such that $(E_1,E_2)\in \mathcal{S}\times\Tilde{\mathcal{Z}}$ and $E_1\independent E_2$. The set of $(m,\gamma)$ consistent with the data and the two assumptions follows by definition as:
\begin{align}
\begin{split}
    &\mathcal{H}^{EV/RA}(m,\gamma) \\
    &:= \left\{\begin{array}{ll} (m,\gamma)\in\mathcal{M}\times \left(\mathcal{P}^{\mathcal{S}}\right)^2: \forall d\in\{0,1\}, \enskip \exists(\upsilon_d,\varsigma_d,\Tilde{Z})\in Sel(\mathbf{Y}_d)\times Sel (\mathbf{S}_d)\times\{\Tilde{Z}\}\cap \Tilde{I},\\
    \gamma_d \eqd \varsigma_d, \enskip m_d(\varsigma_d) = E[\upsilon_d|\varsigma_d] \text{ a.s.}
    \end{array}
\right\}\\
&= \left\{\begin{array}{ll} (m,\gamma)\in\mathcal{M}\times \left(\mathcal{P}^{\mathcal{S}}\right)^2: \forall d\in\{0,1\}, \enskip \exists(\varsigma_d,\Tilde{Z})\in Sel((\mathbf{S}_d,              \Tilde{Z}))\cap \Bar{I},\\
    \exists\upsilon_d\in Sel(\mathbf{Y}_d),\enskip (\upsilon_d,\varsigma_d)\independent \Tilde{Z}, \enskip \gamma_d \eqd \varsigma_d, \enskip m_d(\varsigma_d) = E[\upsilon_d|\varsigma_d] \text{ a.s.}
\end{array}\right\}\\
&= \left\{\begin{array}{ll} (m,\gamma)\in\mathcal{M}\times \left(\mathcal{P}^{\mathcal{S}}\right)^2: \forall d\in\{0,1\}, \enskip \exists(\varsigma_d,\Tilde{Z})\in Sel((\mathbf{S}_d,              \Tilde{Z}))\cap \Bar{I},\\
    \exists\upsilon_d\in Sel(\mathbf{Y}_d),\enskip (\upsilon_d,\varsigma_d)\independent \Tilde{Z}, \enskip \gamma_d \eqd \varsigma_d, \enskip m_d(\varsigma_d) = E_O[\upsilon_d|\varsigma_d] \text{ a.s.}
\end{array}\right\}\\
&= \left\{\begin{array}{ll} (m,\gamma)\in\mathcal{M}\times \left(\mathcal{P}^{\mathcal{S}}\right)^2: \forall d\in\{0,1\}, \enskip \exists(\varsigma_d,\Tilde{Z})\in Sel((\mathbf{S}_d,              \Tilde{Z}))\cap  \Bar{I},\\
    \exists\upsilon_d\in Sel(\mathbf{Y}_d), \enskip \gamma_d \eqd \varsigma_d, \enskip m_d(\varsigma_d) = E_O[\upsilon_d|\varsigma_d] \text{ a.s.}
\end{array}\right\}\\
\end{split}
\end{align}
where the second equality is by \Cref{lem:rectangular_rset}, and the third is by $(\upsilon_d,\varsigma_d)\independent \Tilde{Z}$ and the fourth is by \Cref{lem:removing_joint_independence}. It only remains to impose Assumption \ref{ass:modeling_generic}. Then, the identified set is:
\begin{align}
    \begin{split}
    \mathcal{H}(m,\gamma) &=  \mathcal{H}^{EV/RA}(m,\gamma)\cap (\mathcal{M}^A\times (\mathcal{P}^{\mathcal{S}})^2)\\
       &= \left\{\begin{array}{ll} (m,\gamma)\in\mathcal{M}^A\times (\mathcal{P}^{\mathcal{S}})^2:\enskip \forall d\in\{0,1\}, \enskip \exists(\varsigma_d,\Tilde{Z})\in Sel((\mathbf{S}_d,              \Tilde{Z}))\cap \Bar{I},\\
   \exists\upsilon_d\in Sel(\mathbf{Y}_d),\enskip\gamma_d \eqd \varsigma_d, \enskip m_d(\varsigma_d) = E_O[\upsilon_d|\varsigma_d] \text{ a.s.}
\end{array}\right\}\\
&= \left\{\begin{array}{ll} (m,\gamma)\in\mathcal{M}^A\times (\mathcal{P}^{\mathcal{S}})^2:\enskip \forall d\in\{0,1\}, \enskip \exists\varsigma_d\in Sel(\mathbf{S}_d)\cap I,\\
   \exists\upsilon_d\in Sel(\mathbf{Y}_d),\enskip\gamma_d \eqd \varsigma_d, \enskip m_d(\varsigma_d) = E_O[\upsilon_d|\varsigma_d] \text{ a.s.}
\end{array}\right\}.
\end{split}
\end{align}

where the first equality is by observation and the second is since $\varsigma_d\in Sel(\mathbf{S}_d)\cap I$ can be equivalently stated as $(\varsigma_d,\Tilde{Z})\in Sel(\mathbf{S}_d,\Tilde{Z})\cap \Bar{I}$. 

Next note that for every $(m,\gamma)\in \mathcal{H}(m,\gamma)$, there exist $(\upsilon_0,\varsigma_0,\upsilon_1,\varsigma_1)$ that generate them and that are consistent with the data, modeling assumption, Assumptions \ref{ass:rand_assign} and \ref{ass:ex_validity}. Therefore, $\mathcal{H}(m,\gamma)$ is sharp.
\end{proofEnd}

The lemma exploits the fact that $Y$ is observed only in the observational sample to eliminate restrictions on $(m,\gamma)$ imposed by Assumptions \ref{ass:rand_assign} and \ref{ass:ex_validity} that are redundant in this setting. In particular, the independence restrictions that remain relevant are $S(d)\independent G$ and $S(d)\independent Z|G=E$. Moreover, only observational data impose restrictions on $m$ directly, as reflected by $m_d(S(d)) = E_O[Y(d)|S(d)]$. These observations enable the following sharp characterization of $\mathcal{H}(m,\gamma)$ via moment inequalities identified by the data.

\begin{theoremEnd}[proof at the end, no link to proof]{theorem}\label{thm:tractable_set}
    Let Assumptions \ref{ass:rand_assign},  \ref{ass:ex_validity} and \ref{ass:modeling_generic} hold. Suppose that $E[|Y(d)|]<\infty$ for $d\in\{0,1\}$. The identified set for $(m,\gamma)$ is:
\begin{align}\label{eq:joint_characterization}
\begin{split}
        \mathcal{H}(m,\gamma)   = \left\{\begin{array}{l}
        (m,\gamma)\in\mathcal{M}^A\times(\mathcal{P}^{\mathcal{S}})^2: \text{$\forall d\in\{0,1\}$, $\forall B\in\mathcal{C}(\mathcal{S})$, } \\
         \text{$\gamma_d(B)\geq \max\left(ess\sup_{Z}P_E(S\in B,D=d|Z),P_O(S\in B,D=d) \right)$,}\\
          \text{$\forall u\in\{-1,1\}$: $um_d(s)\leq u\mu_d(s)\pi_{\gamma_d}(s)+ h_{co(\mathcal{Y})}(u)(1-\pi_{\gamma_d}(s))$ $ \gamma_d-$a.e.}
    \end{array}\right\}
\end{split}
\end{align}
where $h_{co(\mathcal{Y})}(u) := \sup_{y\in co(\mathcal{Y})} u y$, $\mu_d(s) := E_O[Y|S=s,D=d]$, and $\pi_{\gamma_d} := d P_O(S,D=d)/d\gamma_d$. If a collection of sets $\mathfrak{C}$ is a core determining class for the containment functional of $\mathbf{S}_d$, then the condition $\forall B\in\mathcal{C}(\mathcal{S})$ can be replaced with $\forall B\in\mathfrak{C}$.
\end{theoremEnd}

\begin{proofEnd}

Recall that $\Bar{I}$ is the set of random elements $(E_1,E_2)$ such that $(E_1,E_2)\in \mathcal{S}\times\Tilde{\mathcal{Z}}$ and $E_1\independent E_2$. Likewise $I$ is defined by \Cref{thm:sharp_set} and $\Tilde{Z}$ by \eqref{eq:tildez}. The proof then proceeds through a series of steps.  I first show that $\mathcal{H}(m,\gamma) = \Tilde{\mathcal{H}}(m,\gamma)$ where:  
    \begin{align}\label{eq:intermediate_id_set_1}
 \begin{split}
    &\Tilde{\mathcal{H}}(m,\gamma) \\
&:=\left\{\begin{array}{l}
(m,\gamma)\in\mathcal{M}^A\times (\mathcal{P}^{\mathcal{S}})^2: \forall d\in\{0,1\}, \text{ $\exists (\varsigma_d,\Tilde{Z})\in Sel((\mathbf{S}_d,\Tilde{Z}))\cap \Bar{I}$, $ \gamma_d\eqd \varsigma_d$,} \\
 \text{ $\forall u\in\{-1,1\}:\enskip u m_d(s)\leq u\mu_d(s)\pi_{\gamma_d}(s)+ h_{co(\mathcal{Y})}(u)(1-\pi_{\gamma_d}(s)),\text{  $\gamma_d-$a.e.}$} \end{array}\right\}
\end{split}
\end{align}

I then show that $\Tilde{\mathcal{H}}(m,\gamma)$ is equivalent to \eqref{eq:joint_characterization}. Observing that $\varsigma_d\in Sel(\mathbf{S}_d)\cap I$ can be equivalently stated as $(\varsigma_d,\Tilde{Z})\in Sel(\mathbf{S}_d,\Tilde{Z})\cap \Bar{I}$, by \Cref{thm:sharp_set}:
\begin{align}
  \mathcal{H}(m,\gamma) &= \left\{\begin{array}{ll} (m,\gamma)\in\mathcal{M}^A\times (\mathcal{P}^{\mathcal{S}})^2:
          \forall d\in\{0,1\},  \enskip  \exists(\varsigma_d,\Tilde{Z})\in Sel((\mathbf{S}_d,\Tilde{Z}))\cap \Bar{I},\\
     \exists\upsilon_d\in Sel(\mathbf{Y}_d), \enskip \gamma_d \eqd \varsigma_d,\enskip m_d(\varsigma_d) = E_O[\upsilon_d|\varsigma_d] \text{ a.s.}
\end{array}\right\}
\end{align}

noting that $\mathcal M^A$ is such that for any $(m,\gamma)$ under consideration, each $m_d$ must be $\gamma_d$-integrable since $E[|Y(d)|]<\infty$ implies
$\int |m_d|d\gamma_d = E[|E[Y(d)|S(d)]|]\leq E[|Y(d)|]<\infty$. It is immediate that then also:
\begin{align}\label{eq:intermediate_id_set}
  \mathcal{H}(m,\gamma) &= \left\{\begin{array}{ll} (m,\gamma)\in\mathcal{M}^A\times (\mathcal{P}^{\mathcal{S}})^2:
          \forall d\in\{0,1\},  \enskip  \exists(\varsigma_d,\Tilde{Z})\in Sel((\mathbf{S}_d,\Tilde{Z}))\cap \Bar{I},\\
     \exists\upsilon_d\in Sel^1(\mathbf{Y}_d), \enskip \gamma_d \eqd \varsigma_d,\enskip m_d(\varsigma_d) = E_O[\upsilon_d|\varsigma_d] \text{ a.s.}
\end{array}\right\}.
\end{align}
\underline{$\mathcal{H}(m,\gamma) = \Tilde{\mathcal{H}}(m,\gamma)$}

To show $\mathcal{H}(m,\gamma)\subseteq \Tilde{\mathcal{H}}(m,\gamma)$, fix $(m,\gamma)\in\mathcal{H}(m,\gamma)$ and $d\in\{0,1\}$. Then there exist $(\varsigma_d,\Tilde Z)\in Sel((\mathbf S_d,\Tilde Z))\cap \Bar I$ and $\upsilon_d\in Sel^1(\mathbf Y_d)$ such that $\gamma_d\eqd \varsigma_d$ and $m_d(\varsigma_d)=E_O[\upsilon_d\mid \varsigma_d]$ a.s.
 Let $\sigma(\varsigma_d|G=O)$ be the sub-$\sigma$-algebra generated by $\varsigma_d$ given $\{\omega\in\Omega:G=O\}$. Let $\mathbb{E}_O[\mathbf{Y}_d|\varsigma_d]:= cl\{E_O[\upsilon_d|\varsigma_d]: \upsilon_d\in Sel^1(\mathbf{Y}_d)\}$, where the closure is taken in $L^1$ space of all $\sigma(\varsigma_d|G=O)$-measurable functions. $\mathbb{E}_O[\mathbf{Y}_d|\varsigma_d]$ exists, is a unique random set, and has at least one integrable selection (\citet[Theorem 2.1.71]{molchanov2017theory}).
By definition of $\mathbb{E}_O[\mathbf{Y}_d|\varsigma_d]$, it is immediate that:
\begin{align}\label{eq:aumann_equivalence_proof}
\begin{split}
     \exists \upsilon_d\in Sel^1(\mathbf{Y}_d): \enskip m_d(\varsigma_d)=E_O[\upsilon_d|\varsigma_d] \text{ a.s.} &\Rightarrow m_d(\varsigma_d)\in \mathbb{E}_O[\mathbf{Y}_d|\varsigma_d] \text{ a.s.}
\end{split}
\end{align}

By assumption, the probability space is non-atomic. By \Cref{lem:prob_space}, $P$ has no atoms over $\sigma(\varsigma_d|G=O)$ for any measurable selection $\varsigma_d$. Since $E[|Y(d)|]<\infty$ for all $d\in\{0,1\}$, $\mathbf{Y}_d$ is integrable. Thus, $\mathbb{E}_O[\mathbf{Y}_d|\varsigma_d]$ is almost surely convex and equal to $\mathbb{E}_O[co(\mathbf{Y}_d)|\varsigma_d]$ (\citet[Theorem 2.1.77]{molchanov2017theory}). Therefore, $h_{\mathbb{E}_O[\mathbf{Y}_d|\varsigma_d]}(u) = h_{\mathbb{E}_O[co(\mathbf{Y}_d)|\varsigma_d]}(u)$ a.s. for all $u\in\mathbb{R}$ by definition of the support function $h$. By $\mathbb{E}_O[co(\mathbf{Y}_d)|\varsigma_d]=\mathbb{E}_O[\mathbf{Y}_d|\varsigma_d]$ and integrability of the latter, the former set is also integrable. It then follows that $h_{\mathbb{E}_O[co(\mathbf{Y}_d)|\varsigma_d]}(u) = E_O[h_{co(\mathbf{Y}_d)}(u)|\varsigma_d]$ a.s. for all $u\in\mathbb{R}$ (\citet[Theorem 2.1.72]{molchanov2017theory}). Hence, recalling that $\mathbb{E}_O[co(\mathbf{Y}_d)|\varsigma_d]=\mathbb{E}_O[\mathbf{Y}_d|\varsigma_d]$, also $h_{\mathbb{E}_O[\mathbf{Y}_d|\varsigma_d]}(u) =h_{\mathbb{E}_O[co(\mathbf{Y}_d)|\varsigma_d]}(u)= E_O[h_{co(\mathbf{Y}_d)}(u)|\varsigma_d]$ a.s. for all $u\in\mathbb{R}$. Then:
\begin{align}\label{eq:md_equivalent}
\begin{split}
   m_d(\varsigma_d)\in\mathbb{E}_O[\mathbf{Y}_d|\varsigma_d] \text{ a.s.}
     \Leftrightarrow&\forall u\in\{-1,1\}:\enskip u m_d(\varsigma_d)\leq h_{\mathbb{E}_O[\mathbf{Y}_d|\varsigma_d]}(u) \text{ a.s.}\\
  \Leftrightarrow&\forall u\in\{-1,1\}:\enskip u m_d(\varsigma_d)\leq E_O[h_{co(\mathbf{Y}_d)}(u)|\varsigma_d] \text{ a.s.}
\end{split}
\end{align}
where the first line is by \citet[Theorem 13.1]{rockafellar1970convex} and almost sure convexity of $\mathbb{E}_O[\mathbf{Y}_d|\varsigma_d]$, and the second is by $h_{\mathbb{E}_O[\mathbf{Y}_d|\varsigma_d]}(u) = E_O[h_{co(\mathbf{Y}_d)}(u)|\varsigma_d]$ a.s. for all $u\in\mathbb{R}$. Moreover:
    \begin{align}\label{eq:support_lie}
        \begin{split}
            E_O[h_{co(\mathbf{Y}_d)}(u)|\varsigma_d] &= E_O[h_{co(\mathbf{Y}_d)}(u)|\varsigma_d,D=d]P_O(D=d|\varsigma_d)\\
            &+ E_O[h_{co(\mathbf{Y}_d)}(u)|\varsigma_d,D\neq d]P_O(D\neq d|\varsigma_d)\\
            &= uE_O[Y|\varsigma_d,D=d]P_O(D=d|\varsigma_d)+ h_{co(\mathcal{Y})}(u)P_O(D\neq d|\varsigma_d)
                 \\
             &= uE_O[Y|S,D=d]P_O(D=d|\varsigma_d)+ h_{co(\mathcal{Y})}(u)P_O(D\neq d|\varsigma_d)
                 \\
              &= u\mu_d(\varsigma_d)P_O(D=d|\varsigma_d)+ h_{co(\mathcal{Y})}(u)P_O(D\neq d|\varsigma_d)\\
              &= u\mu_d(\varsigma_d)\pi_{\gamma_d}(\varsigma_d)+ h_{co(\mathcal{Y})}(u)(1-\pi_{\gamma_d}(\varsigma_d))
        \end{split}
    \end{align}
where the first equality is by LIE. The second follows because $co(\mathbf{Y}_d) = \{Y\}$ whenever $D=d$, $h_{\{Y\}}(u) = u Y$, and $co(\mathbf{Y}_d) = co(\mathcal{Y})$ when $D\neq d$. The third is by observing that $P_O(\varsigma_d = S|D=d)=1$ since $\varsigma_d \in Sel(\mathbf{S}_d)$ and $\mathbf{S}_d = \{S\}$ when $D=d$. The fourth is by definition of $\mu_d$ and $P_O(\varsigma_d = S|D=d)=1$. The final equality is by \Cref{lem:prop_score}. Then observe that:
    \begin{align}\label{eq:support_lie_2}
        \begin{split}
            &\forall u\in\{-1,1\}:\enskip u m_d(\varsigma_d)\leq E_O[h_{co(\mathbf{Y}_d)}(u)|\varsigma_d] \text{ a.s.}\\
             \Leftrightarrow&\forall u\in\{-1,1\}:\enskip u m_d(\varsigma_d)\leq u\mu_d(\varsigma_d)\pi_{\gamma_d}(\varsigma_d)+ h_{co(\mathcal{Y})}(u)(1-\pi_{\gamma_d}(\varsigma_d))\text{ a.s.}\\
           \Leftrightarrow&\forall u\in\{-1,1\}:\enskip u m_d(s)\leq u\mu_d(s)\pi_{\gamma_d}(s)+ h_{co(\mathcal{Y})}(u)(1-\pi_{\gamma_d}(s))\text{  $\gamma_d-$a.e.}
        \end{split}
    \end{align}
    where the second line follows by \eqref{eq:support_lie} and the third by $\varsigma_d\eqd \gamma_d$. Since $\varsigma_d$ was an arbitrary selection such that $(\varsigma_d,\Tilde{Z})\in Sel((\mathbf{S}_d,\Tilde{Z}))\cap \Bar{I}$ and $\varsigma_d\eqd \gamma_d$ for any $d\in\{0,1\}$, by \eqref{eq:aumann_equivalence_proof}, \eqref{eq:md_equivalent} and \eqref{eq:support_lie_2} then $\mathcal{H}(m,\gamma)\subseteq \Tilde{\mathcal{H}}(m,\gamma)$.

    For the converse, for any $d\in\{0,1\}$ fix any $\varsigma_d$ such that $(\varsigma_d,\Tilde{Z})\in Sel(\mathbf{S}_d)\times\{\Tilde{Z}\}\cap \Bar{I}$, letting $\varsigma_d\eqd \gamma_d$. Let $m\in\mathcal{M}^A$ be a pair of arbitrary $m_d$ such that for $d\in\{0,1\}$:
    \begin{align}
          \forall u\in\{-1,1\}:\enskip u m_d(s)\leq u\mu_d(s)\pi_{\gamma_d}(s)+ h_{co(\mathcal{Y})}(u)(1-\pi_{\gamma_d}(s))\text{  $\gamma_d-$a.e.}
    \end{align}
    Since $E[|Y(d)|]<\infty$ and $P(G=O)>0$, $E_O[|Y(d)|]<\infty$. Given that $\mathcal{M}^A$ must be such that $m_d$ is $\gamma_d$ integrable for any $d\in\{0,1\}$, by \Cref{lem:reverse_aumann} then for $d\in\{0,1\}$ there exist $\upsilon_d\in Sel^1(\mathbf{Y}_d)$ such that $m_d(\varsigma_d) = E_O[\upsilon_d|\varsigma_d]$ a.s. It is then immediate that $ \Tilde{\mathcal{H}}(m,\gamma)\subseteq \mathcal{H}(m,\gamma)$.

\underline{$\Tilde{\mathcal{H}}(m,\gamma)$ is equivalent to \eqref{eq:joint_characterization}}

By \citet[Theorem 2.1]{artstein1983distributions}, a distribution function characterizes a selection in $Sel((\mathbf{S}_d,\Tilde{Z}))$ if and only if:
\begin{align}
   &\forall B\in\mathcal{C}(\mathcal{S}\times \Tilde{\mathcal{Z}}):\enskip P((S(d),\Tilde{Z})\in B) \geq P((\mathbf{S}_d,\Tilde{Z})\subseteq B) \\
   \Leftrightarrow  &\forall B\in\mathcal{C}(\mathcal{S}):\enskip P(S(d)\in B|\Tilde{Z}) \geq P(\mathbf{S}_d\subseteq B|\Tilde{Z}) \text{ a.s.} 
\end{align}
where the second line follows by \citet[Theorem 2.33]{molchanov2018random}. Now consider the containment functional $P(\mathbf{S}_d\subseteq B|\Tilde{Z})$. If $B = \mathcal{S}$, $P(\mathbf{S}_d\subseteq B|\Tilde{Z})=1$. If $B \subsetneq \mathcal{S}$, then $P(\mathbf{S}_d\subseteq B|\Tilde{Z})=P(S\in B,D=d|\Tilde{Z})$. Hence, $\exists (\varsigma_d,\Tilde{Z})\in Sel((\mathbf{S}_d,\Tilde{Z}))$ such that $\gamma_d\in\mathcal{P}^\mathcal{S}$ and $\gamma_d\eqd \varsigma_d$ if and only if:
\begin{align}\label{eq:gamma_intermediate}
  \forall B\in\mathcal{C}(\mathcal{S}):\enskip P(\varsigma_d\in B|\Tilde{Z}) \geq P(S\in B,D=d|\Tilde{Z})  \text{ a.s.}
\end{align}

Since $(\varsigma_d,\Tilde{Z})\in \Bar{I}$, \eqref{eq:gamma_intermediate} is equivalent to:
\begin{align}
  \forall B\in\mathcal{C}(\mathcal{S}):\enskip P(\varsigma_d\in B)&\geq ess\sup_{\Tilde{Z}}P(S\in B,D=d|\Tilde{Z})\\
   &=\max\left(ess\sup_{Z}P_E(S\in B,D=d|Z),P_O(S\in B,D=d) \right)
\end{align}
where the first line follows since, by definition of $\Bar{I}$, $\varsigma_d\independent \Tilde{Z}$. The second is by definition of $\Tilde{Z}$ and $P(G=O)>0$ given that two datasets are observed. Therefore:
\begin{align}\label{eq:gamma_equivalent}
\begin{split}
&\text{$\exists (\varsigma_d,\Tilde{Z})\in Sel((\mathbf{S}_d,\Tilde{Z}))\cap \Bar{I}$ s.t. $\gamma_d\eqd \varsigma_d$}\\
  \Leftrightarrow&\forall B\in\mathcal{C}(\mathcal{S}):\enskip  P(\varsigma_d\in B)\geq\max\left(ess\sup_{Z}P_E(S\in B,D=d|Z),P_O(S\in B,D=d) \right)
\end{split}
\end{align}

By definition of $\Tilde{\mathcal{H}}(m,\gamma)$ and \eqref{eq:gamma_equivalent}:
\begin{align}\label{eq:intermediate_id_set_final}
 \begin{split}
  \Tilde{\mathcal{H}}(m,\gamma)&= \left\{\begin{array}{l}
(m,\gamma)\in\mathcal{M}^A\times (\mathcal{P}^{\mathcal{S}})^2: \forall d\in\{0,1\}, \text{$\forall B\in\mathcal{C}(\mathcal{S}),$} \\
  \text{$\gamma_d(B)\geq\max\left(ess\sup_{Z}P_E(S\in B,D=d|Z),P_O(S\in B,D=d) \right)$,} \\
 \text{ $\forall u\in\{-1,1\}:\enskip u m_d(s)\leq u\mu_d(s)\pi_{\gamma_d}(s)+ h_{co(\mathcal{Y})}(u)(1-\pi_{\gamma_d}(s)),\text{  $\gamma_d-$a.e.}$}\end{array}\right\}
\end{split}
\end{align}

By its definition, if $\mathfrak{C}$ is a core determining class, \eqref{eq:gamma_equivalent} is equivalent to:
\begin{align}\label{eq:gamma_equivalent_core}
\begin{split}
&\text{$\exists (\varsigma_d,\Tilde{Z})\in Sel((\mathbf{S}_d,\Tilde{Z}))\cap \Bar{I}$ s.t. $\gamma_d\eqd \varsigma_d$}\\
  \Leftrightarrow&\forall B\in\mathfrak{C}:\enskip  P(\varsigma_d\in B)\geq\max\left(ess\sup_{Z}P_E(S\in B,D=d|Z),P_O(S\in B,D=d) \right)
\end{split}
\end{align}
It is then immediate that in \eqref{eq:intermediate_id_set_final} the condition $\forall B\in\mathcal{C}(\mathcal{S})$ can be replaced by $\forall B\in\mathfrak{C}$. 

The result then follows by $\mathcal{H}(m,\gamma) = \Tilde{\mathcal{H}}(m,\gamma)$.
\end{proofEnd}

The argument combines Artstein’s theorem (\citet[Theorem 2.1]{artstein1983distributions}) with the conditional Aumann expectation, tools that are typically applied separately, to address a novel technical challenge: jointly bounding (i) the distribution function of a selection of a random set and (ii) a conditional mean whose conditioning $\sigma$-algebra is generated by that selection.

To build intuition for the bounds, consider conditions that are necessary for $(m,\gamma)$ to be compatible with the data and maintained assumptions. The proof shows that these are also sufficient, and hence deliver sharp bounds. First, recalling $\gamma_d:=P_{S(d)}$, it is immediate that for any closed subset $B\in\mathcal{C}(\mathcal S)$: $\gamma_d(B)\geq P(S(d)\in B,D=d)$. By Assumptions \ref{ass:rand_assign} and \ref{ass:ex_validity}, $S(d)\independent G$ and $S(d)\independent Z|G=E$. It is therefore also necessary that each $\gamma_d$ satisfies:
\begin{align}\label{eq:gamma_dominance}
    \forall B\in\mathcal{C}(\mathcal{S}):\enskip \gamma_d(B)\geq P_O(S(d)\in B,D=d) \quad \text{and} \quad  \gamma_d(B)\geq P_E(S(d)\in B,D=d|Z) \text{ a.s.}
\end{align}
which coincides with the restrictions on $\gamma_d$ in the theorem. This represents an intersection bound on $\gamma_d$ reflecting that $S(d)$ is observed in both datasets when $D=d$.

Second, recalling that under Assumption \ref{ass:ex_validity} $m_d(S(d)) = E_O[Y(d)|S(d)]$, by the law of iterated expectations:
\begin{align*}
    m_d(S(d))= E_O[Y|D=d,S(d)]P_O(D=d|S(d))+E_O[Y(d)|D\neq d,S(d)] P_O(D\neq d|S(d)) \text{ a.s.}
\end{align*}
Let $\pi_{\gamma_d}(s):=P_O(D=d|S(d)=s)$ denote the latent propensity score for a given $\gamma_d\eqd S(d)$. Since $S(d)=S$ whenever $D=d$ and $S(d)\independent G$, $\pi_{\gamma_d}$ is the Radon-Nikodym derivative such that $\pi_{\gamma_d}(S(d))=\frac{dP_O(S,D=d)}{d\gamma_d}$. Because $E_O[Y(d)|D\neq d,S(d)]\in\mathcal{Y}$ is never observed, the usual worst-case bounds arguments imply:
\begin{align}\label{eq:m_wc}
\begin{split}
        m_d(s)&\geq E_O[Y|D=d,S=s]\pi_{\gamma_d}(s)+\inf\mathcal{Y} (1-\pi_{\gamma_d}(s)) \text{ $\gamma_d-$a.e.}\\
        m_d(s)&\leq E_O[Y|D=d,S=s]\pi_{\gamma_d}(s)+\sup\mathcal{Y} (1-\pi_{\gamma_d}(s)) \text{ $\gamma_d-$a.e.}
\end{split}
\end{align}
which is identical to the restrictions on $m_d$ in the theorem. Accordingly, these conditions represent worst-case bounds on $m_d$ for a \textit{given} $\gamma_d$. Notably, the bounds on $m_d$ may change with the particular $\gamma_d$ under consideration.

Finally, \Cref{thm:tractable_set} also enables the use of a \textit{core-determining class} (CDC) to remove redundant constraints on $\gamma$ without loss of information. Intuitively, a CDC does so by eliminating constraints that are implied by the remaining ones.

\subsection[Characterizing H(τ)]{Characterizing $\mathcal{H}(\tau)$}

The sharp identified set $\mathcal{H}(\tau)$ follows as the image of the functional $T$ over the set of possible $(m,\gamma)$. The next result shows that under easily verifiable high-level conditions on $\mathcal{M}^A$, $\mathcal{H}(\tau)$ is an interval which can be characterized by solving two optimization problems.

\begin{theoremEnd}[proof at the end, no link to proof]{theorem}\label{thm:interval_set_general}
Let Assumptions \ref{ass:rand_assign}, \ref{ass:ex_validity}, and \ref{ass:modeling_generic} hold. Suppose that $E[|Y(d)|]<\infty$ for $d\in\{0,1\}$ and that $\mathcal{M}^A$ is closed and convex. Then the closure of $\mathcal{H}(\tau)$ is:
\begin{align}\label{eq:general_characterization}
   \left[\inf_{(\tilde{m},\tilde{\gamma})\in\mathcal{H}(m,\gamma)}T(\tilde{m},\tilde{\gamma}),\sup_{(\tilde{m},\tilde{\gamma})\in\mathcal{H}(m,\gamma)}T(\tilde{m},\tilde{\gamma})\right].
\end{align}
\end{theoremEnd}

\begin{proofEnd}

Recall $\mathcal{H}(\tau)=\{T(m,\gamma):(m,\gamma)\in\mathcal{H}(m,\gamma)\}$. The proof proceeds in three steps.

\underline{\textbf{Step 1:} $\mathcal{H}(m,\gamma)$ is convex.}

Take any $(m,\gamma),(m',\gamma')\in\mathcal{H}(m,\gamma)$ and any $a\in(0,1)$. Define:
\begin{align}
m^a:=a m+(1-a)m',\qquad \gamma^a:=a\gamma+(1-a)\gamma'.
\end{align}
Since $\mathcal{M}^A$ is convex by assumption, $m^a\in\mathcal{M}^A$. $\mathcal{P}^\mathcal{S}$ is the space of probability measures and is therefore convex, so $\gamma_d^a\in\mathcal{P}^\mathcal{S}$ for each $d$. It remains to verify that set constraints on $m_d$ and $\gamma_d$ from \Cref{thm:tractable_set} are satisfied for $d\in\{0,1\}$ .

For $\gamma_d$ constraints, note that for any $d\in\{0,1\}$ and $B\in\mathcal{C}(\mathcal{S})$:
\begin{align}
\gamma_d^a(B)&=a\gamma_d(B)+(1-a)\gamma'_d(B)\\
&\geq a\max\!\left(ess\sup_{Z}P_E(S\in B,D=d\mid Z),P_O(S\in B,D=d)\right)\notag\\
&\quad+(1-a)\max\!\left(ess\sup_{Z}P_E(S\in B,D=d\mid Z),P_O(S\in B,D=d)\right)\notag\\
&=\max\!\left(ess\sup_{Z}P_E(S\in B,D=d\mid Z),P_O(S\in B,D=d)\right),\notag
\end{align}
where the first line is by definition of $\gamma^a$, the second is by $(m,\gamma), (m',\gamma')\in \mathcal{H}(m,\gamma)$ and \Cref{thm:tractable_set}, and the third is by observation.

For the $m_d$ constraint, fix any $d\in\{0,1\}$. For any $B$ with $\gamma_d(B)=0$, $P_O(S\in B,D=d)\leq \gamma_d(B)=0$. Hence, $P_O(S\in\cdot,D=d)\ll \gamma_d $, and $P_O(S\in\cdot,D=d)\ll \gamma'_d$. Since $\gamma_d^a=a\gamma_d+(1-a)\gamma'_d$, also $P_O(S\in\cdot,D=d)\ll \gamma_d^a$. Then by the Radon-Nikodym theorem there exist measurable functions:
\begin{align}
\pi_{\gamma_d}:=\frac{dP_O(S\in\cdot,D=d)}{d\gamma_d},\qquad
\pi_{\gamma'_d}:=\frac{dP_O(S\in\cdot,D=d)}{d\gamma'_d},\qquad
\pi^a:=\frac{dP_O(S\in\cdot,D=d)}{d\gamma_d^a}.
\end{align}
Since $\gamma_d^a=a\gamma_d+(1-a)\gamma'_d$ with $a\in(0,1)$, $\gamma_d \ll \gamma_d^a$ and $\gamma'_d \ll \gamma_d^a$. Therefore, similarly, the following exist:
\begin{align}
\rho:=\frac{d\gamma_d}{d\gamma_d^a},\qquad \rho':=\frac{d\gamma'_d}{d\gamma_d^a}.
\end{align}
Then $\rho,\rho'\geq 0$ by nonnegativity of measures, and:
\begin{align}\label{eq:rho_identity_general} a\rho+(1-a)\rho'=1\qquad \gamma_d^a-\text{a.e.}
\end{align}
by $\gamma_d^a=a\gamma_d+(1-a)\gamma'_d$ and linearity of the Radon--Nikodym derivative. By the chain rule for Radon--Nikodym derivatives, set:
\begin{align}\label{eq:pi_chain_general}
\pi^a=\pi_{\gamma_d}\rho=\pi_{\gamma'_d}\rho'\qquad \gamma_d^a-\text{a.e.}
\end{align}
with the convention $\pi_{\gamma_d}=0$ on the event $\{\rho=0\}$ and similarly $\pi_{\gamma'_d}=0$ on $\{\rho'=0\}$.

Since $(m,\gamma)\in\mathcal{H}(m,\gamma)$, for each $u\in\{-1,1\}$:
\begin{align}\label{eq:moment_gamma_general}
u m_d(s)\leq u\mu_d(s)\pi_{\gamma_d}(s)+h_{co(\mathcal{Y})}(u)\bigl(1-\pi_{\gamma_d}(s)\bigr)\qquad \gamma_d-\text{a.e.}
\end{align}
by \Cref{thm:tractable_set}, and similarly with $(m',\gamma')$. Since \eqref{eq:moment_gamma_general} holds $\gamma_d$-a.e. and $\gamma_d(B)=\int_B \rho(s)d\gamma_d^a(s)$ for every measurable $B$, it follows that \eqref{eq:moment_gamma_general}
also holds $\gamma_d^a$-a.e. on $\{\rho>0\}$. Analogously, the corresponding inequality for $(m_d',\gamma_d')$ holds $\gamma_d^a$-a.e. on $\{\rho'>0\}$.

Fix $u\in\{-1,1\}$. On $\{\rho>0\}\cap\{\rho'>0\}$, by \eqref{eq:pi_chain_general} rewrite \eqref{eq:moment_gamma_general} as:
\begin{align}
u m_d(s)\leq u\mu_d(s)\frac{\pi^a(s)}{\rho(s)}+h_{co(\mathcal{Y})}(u)\Big(1-\frac{\pi^a(s)}{\rho(s)}\Big)\qquad \gamma_d^a-\text{a.e.\ on }\{\rho>0\},
\end{align}
and similarly:
\begin{align}
u m'_d(s)\leq u\mu_d(s)\frac{\pi^a(s)}{\rho'(s)}+h_{co(\mathcal{Y})}(u)\Big(1-\frac{\pi^a(s)}{\rho'(s)}\Big)\qquad \gamma_d^a-\text{a.e.\ on }\{\rho'>0\}.
\end{align}
Then define $A(s):=\frac{a}{\rho(s)}+\frac{1-a}{\rho'(s)}$ on the event $\{\rho>0\}\cap\{\rho'>0\}$ and note:
\begin{align}\label{eq:convex_combination_general}
\begin{split}
u m_d^a(s)
&=au m_d(s)+(1-a)u m'_d(s)\\
&\leq a\Big(u\mu_d(s)\frac{\pi^a(s)}{\rho(s)}+h_{co(\mathcal{Y})}(u)\Big(1-\frac{\pi^a(s)}{\rho(s)}\Big)\Big)\\
&\quad +(1-a)\Big(u\mu_d(s)\frac{\pi^a(s)}{\rho'(s)}+h_{co(\mathcal{Y})}(u)\Big(1-\frac{\pi^a(s)}{\rho'(s)}\Big)\Big)\\
&=u\mu_d(s)\pi^a(s) \Big(\frac{a}{\rho(s)}+\frac{1-a}{\rho'(s)}\Big) + h_{co(\mathcal{Y})}(u)\Big(1-\pi^a(s)\Big(\frac{a}{\rho(s)}+\frac{1-a}{\rho'(s)}\Big)\Big)\\
&=u\mu_d(s)\pi^a(s) A(s) + h_{co(\mathcal{Y})}(u)\bigl(1-\pi^a(s) A(s)\bigr), \text{ on }\{\rho>0\}\cap\{\rho'>0\}
\end{split}
\end{align}
where the first line is by definition of $m^a_d$, the second is by the two inequalities above, the third is by rearrangement and factoring out $\pi^a(s)$, and the fourth is by definition of $A(s)$.  Observe that:
\begin{align}\label{eq:A_ge_1_general}
\begin{split}
A(s)-1
&=\frac{a}{\rho(s)}+\frac{1-a}{\rho'(s)}-1\\
&=\frac{a\rho'(s)+(1-a)\rho(s)}{\rho(s)\rho'(s)}-\frac{1}{a\rho(s)+(1-a)\rho'(s)}\\
&=\frac{(a\rho'(s)+(1-a)\rho(s))(a\rho(s)+(1-a)\rho'(s))-\rho(s)\rho'(s)}{\rho(s)\rho'(s)(a\rho(s)+(1-a)\rho'(s))}\\
&=\frac{a(1-a)(\rho(s)-\rho'(s))^2}{\rho(s)\rho'(s)(a\rho(s)+(1-a)\rho'(s))}\\
&=\frac{a(1-a)(\rho(s)-\rho'(s))^2}{\rho(s)\rho'(s)}\geq 0
\qquad\text{on }\{\rho>0\}\cap\{\rho'>0\},
\end{split}
\end{align}
where the first line is by definition of $A(s)$, the second is by \eqref{eq:rho_identity_general}, the third and fourth are by observation, and the fifth is by \eqref{eq:rho_identity_general}. Note that $\mu_d(s)=E_O[Y|S=s,D=d]$ is well-defined and finite 
$P_O(S\in\cdot|D=d)-$a.e. because $E_O[|Y|\mathbbm{1}[D=d]]<\infty$ 
by assumption. Moreover, since $Y\in\mathcal{Y}$ a.s., 
$\mu_d(s)\in\mathrm{co}(\mathcal{Y})$ wherever it is defined and hence $u\mu_d(s)\leq h_{\mathrm{co}(\mathcal{Y})}(u),\text{ for }u\in\{-1,1\}$. Since $A(s)\geq 1$ and $u\mu_d(s)-h_{\mathrm{co}(\mathcal{Y})}(u)\leq 0$, by \eqref{eq:convex_combination_general}:
\begin{align}
um_d^a(s)\leq h_{\mathrm{co}(\mathcal{Y})}(u)+(u\mu_d(s)-
h_{\mathrm{co}(\mathcal{Y})}(u))\pi^a(s)A(s)\leq 
u\mu_d(s)\pi^a(s)+h_{\mathrm{co}(\mathcal{Y})}(u)(1-\pi^a(s))
\end{align}
on $\{\rho>0\}\cap\{\rho'>0\}$.

On $\{\rho=0\}\cup\{\rho'=0\}$, note that $\gamma_d^a(\{\rho=0\}
\cap\{\rho'=0\})=0$ by \eqref{eq:rho_identity_general}, so the constraint holds trivially there. On the remaining events 
$\{\rho>0\}\cap\{\rho'=0\}$ and $\{\rho=0\}\cap\{\rho'>0\}$, 
\eqref{eq:pi_chain_general} implies $\pi^a=0$ $\gamma_d^a$-a.e., 
so the constraint reduces to $um_d^a(s)\leq h_{\mathrm{co}
(\mathcal{Y})}(u)$. Since $(m,\gamma),(m',\gamma')\in\mathcal{H}
(m,\gamma)$ implies $m\in\mathcal{M}^A\subseteq\mathcal{M}$ and 
$m'\in\mathcal{M}^A\subseteq\mathcal{M}$, we have $m_d(s),m'_d(s)
\in\mathcal{Y}\subseteq\mathrm{co}(\mathcal{Y})$ everywhere. 
Therefore $m_d^a(s)=am_d(s)+(1-a)m'_d(s)\in\mathrm{co}(\mathcal{Y})$ 
everywhere by convexity, and hence $um_d^a(s)\leq h_{\mathrm{co}
(\mathcal{Y})}(u)$ holds everywhere.

Combining the cases $\{\rho>0\}\cap\{\rho'>0\}$ and $\{\rho=0\}\cup\{\rho'=0\}$, for every $u\in\{-1,1\}$:
\begin{align}
u m_d^a(s)\leq u\mu_d(s)\pi^a(s)+h_{co(\mathcal{Y})}(u)\bigl(1-\pi^a(s)\bigr)
\qquad \gamma_d^a-\text{a.e.}
\end{align}
Since $\pi^a=dP_O(S\in\cdot,D=d)/d\gamma_d^a=\pi_{\gamma_d^a}$, the moment constraints hold for $(m^a,\gamma^a)$. Thus $(m^a,\gamma^a)\in\mathcal{H}(m,\gamma)$ and $\mathcal{H}(m,\gamma)$ is convex.

\underline{\textbf{Step 2:} $T$ is continuous on line segments over $\mathcal{H}(m,\gamma)$.}

Fix arbitrary $(m^1,\gamma^1),(m^0,\gamma^0)\in\mathcal{H}(m,\gamma)$ and define for $t\in[0,1]$:
\begin{align}
m^t:=t m^1+(1-t)m^0,\qquad \gamma^t:=t\gamma^1+(1-t)\gamma^0.
\end{align}
By Step 1, $(m^t,\gamma^t)\in\mathcal{H}(m,\gamma)$ for every $t\in[0,1]$, so $T(m^t,\gamma^t)$ is well-defined and finite for every $t$ by $E[|Y(d)|]<\infty$. Fix any $t\in(0,1)$ and $d\in\{0,1\}$. Since $(m^t,\gamma^t)\in\mathcal{H}(m,\gamma)$:
\begin{align}\label{eq:mt_integrable_gt_general}
\int |m^t_d|d\gamma^t_d<\infty.
\end{align}
Then:
\begin{align}\label{eq:mt_integrable_cross_general}
\int |m^t_d|d\gamma^0_d\leq \frac{1}{1-t}\int |m^t_d|d\gamma^t_d<\infty,
\qquad
\int |m^t_d|d\gamma^1_d\leq \frac{1}{t}\int |m^t_d|d\gamma^t_d<\infty,
\end{align}
where the inequalities are by $\gamma^t_d=t\gamma^1_d+(1-t)\gamma^0_d\geq (1-t)\gamma^0_d$ and
$\gamma^t_d\geq t\gamma^1_d$, and the finiteness is by \eqref{eq:mt_integrable_gt_general}. Moreover, by $|t x+(1-t)y|\geq t|x|-(1-t)|y|$:
\begin{align}\label{eq:revtri_pointwise}
|m^t_d|
=|t m^1_d+(1-t)m^0_d|
\geq t|m^1_d|-(1-t)|m^0_d|.
\end{align}
Rearranging yields $t|m^1_d|\leq |m^t_d|+(1-t)|m^0_d|$. Integrating with respect to $\gamma^0_d$ gives:
\begin{align}\label{eq:cross_int_m1_g0}
t\int |m^1_d|d\gamma^0_d
\leq \int |m^t_d|d\gamma^0_d+(1-t)\int |m^0_d|d\gamma^0_d<\infty,
\end{align}
where the finiteness is by \eqref{eq:mt_integrable_cross_general} and $\int |m^0_d|d\gamma^0_d<\infty$, since
$(m^0,\gamma^0)\in\mathcal{H}(m,\gamma)$. Since $t>0$, $\int |m^1_d|d\gamma^0_d<\infty$. The argument for $\int |m^0_d|d\gamma^1_d<\infty$ is analogous. Hence for each $d$, the integrals $\int m^1_dd\gamma^1_d$, $\int m^1_dd\gamma^0_d$, $
\int m^0_dd\gamma^1_d$, $ \int m^0_dd\gamma^0_d$ are finite. Therefore:
\begin{align}
\begin{split}
\int m^t_dd\gamma^t_d
&=\int (t m^1_d+(1-t)m^0_d)d(t\gamma^1_d+(1-t)\gamma^0_d)\\
&=t^2\int m^1_dd\gamma^1_d
+t(1-t)\int m^1_dd\gamma^0_d
+t(1-t)\int m^0_dd\gamma^1_d
+(1-t)^2\int m^0_dd\gamma^0_d,
\end{split}
\end{align}
where the first line is by definition of $m^t_d$ and $\gamma^t_d$, and the second by bilinearity of the integral and the fact that individual integrals are finite. Then $f(t): = \int m^t_1d\gamma^t_1-\int m^t_0d\gamma^t_0$ is a quadratic polynomial in $t$, hence continuous on $[0,1]$. Because $f(t) = T(m^t,\gamma^t)$, $T$ is continuous on line segments.

\underline{\textbf{Step 3:} $cl(\mathcal{H}(\tau))$ is a closed interval.}

Step 1 implies $\mathcal{H}(m,\gamma)$ is convex, hence path-connected by line segments: for any two points $(m,\gamma),(m',\gamma')\in\mathcal{H}(m,\gamma)$, the line segment $\{(m^t,\gamma^t):t\in[0,1]\}$ lies entirely in $\mathcal{H}(m,\gamma)$. By Step 2, $T$ is continuous along every such line segment. Therefore, for any two values $\tau_0=T(m,\gamma),\tau_1=T(m',\gamma')\in\mathcal{H}(\tau)$, the set $\mathcal{H}(\tau)$ contains the continuous image $\{T(m^t,\gamma^t):t\in[0,1]\}$ of $[0,1]$ connecting $\tau_0$ and $\tau_1$. By the intermediate value theorem, this image contains the whole interval $[\min(\tau_0,\tau_1),\max(\tau_0,\tau_1)]$. Hence $\mathcal{H}(\tau)\subseteq\mathbb{R}$ is an interval. The closure of an interval in $\mathbb{R}$ is the closed interval with endpoints given by its infimum and supremum. Thus:
\begin{align}
cl(\mathcal{H}(\tau))
=
\left[\inf_{(\tilde{m},\tilde{\gamma})\in\mathcal{H}(m,\gamma)}T(\tilde{m},\tilde{\gamma}),\;
\sup_{(\tilde{m},\tilde{\gamma})\in\mathcal{H}(m,\gamma)}T(\tilde{m},\tilde{\gamma})\right].
\end{align}

\end{proofEnd}

Using optimization problems to characterize identified sets has become common in partial identification. Such representations typically follow from linear objectives and polyhedral, hence convex, constraint sets. \Cref{thm:interval_set_general} requires a different argument since $T$ is bilinear -- and thus separately continuous -- and the constraint set is bilinearly constrained and therefore need not be convex. The proof shows that $\mathcal{H}(m,\gamma)$ is nevertheless convex under the assumptions of the theorem and that $T$ is continuous along line segments in $\mathcal{H}(m,\gamma)$. Then $\mathcal{H}(\tau)$ is a line-continuous image of a convex set, hence a connected set in $\mathbb{R}$, i.e., an interval.

The restrictions $\mathcal{M}^A$ enter as constraints in the optimization problems. This offers three advantages. First, it enables the implementation of LIV and TI. Second, it nests existing point-identification results and generalizes them by allowing for imperfect compliance (for examples see \Cref{rem:sufficient_conditions}). Third, it provides a tool that facilitates development of new restrictions on $m$ tailored for specific empirical settings by computationally producing the corresponding sharp bounds for $\tau$. This removes the need to derive closed-form expressions and to prove sharpness for each set of assumptions.

\begin{remark}\label{rem:regularity}
Assumptions \ref{ass:LMIV} and \ref{ass:TI} are representable via linear equality and inequality restrictions on $m$. Therefore, each resulting $\mathcal{M}^A$ is an intersection of (possibly infinitely many) affine subspaces and halfspaces in a linear space of functions on $\mathcal{S}$, and is thus convex. Moreover, because the restrictions are imposed pointwise (or $\gamma_d-$a.e.) and are preserved under limits, the corresponding $\mathcal{M}^A$ are closed. Furthermore, whenever $m$ is identified by the data, such as under existing assumptions in \Cref{rem:sufficient_conditions}, $\mathcal{M}^A$ is a singleton ($\gamma_d-$a.e.) and hence closed and convex.
\end{remark}

\subsection{Implementation and Estimation}\label{sect:implementation}

This section discusses how the optimization problems can be used to tractably characterize and estimate $\mathcal{H}(\tau)$. The problems become finite-dimensional when $\mathcal{S}$ is a finite set, as shown by the following corollary. 

\begin{theoremEnd}[proof at the end, no link to proof]{corollary}\label{cor:interval_set}
Suppose that assumptions of \Cref{thm:interval_set_general} hold and that $k:=|\mathcal{S}|<\infty$. Let $\Delta(k)$ denote the $k-$dimensional simplex. Then the closure of $\mathcal{H}(\tau) $ is:
    \begin{align}\label{eq:joint_characterization_discrete_optimization}
        \left[\inf_{(\Tilde{m},\Tilde{\gamma})\in\mathcal{H}(m,\gamma)}T(\Tilde{m},\Tilde{\gamma}),\sup_{(\Tilde{m},\Tilde{\gamma})\in\mathcal{H}(m,\gamma)}T(\Tilde{m},\Tilde{\gamma})\right].
    \end{align}
where:
\begin{align}\label{eq:joint_characterization_discrete}
\begin{split}
         \mathcal{H}(m,\gamma)   =&\left\{\begin{array}{l}
        (m,\gamma)\in\mathcal{M}^A\times (\Delta(k))^2: \text{$\forall d\in\{0,1\}$, $ \forall s\in\mathcal{S}$,}\\
         \text{$\gamma_d(s)\geq \max\left(ess\sup_{Z}P_E(S = s,D=d|Z),P_O(S=s,D=d) \right),$}\\
    \text{$\left(m_d(s)-\inf\mathcal{Y}\right)\gamma_d(s)\geq \left(E_{O}[Y|S=s,D=d]-\inf\mathcal{Y}\right)P_{O}(S=s,D=d)$ },\\
   \text{$\left(\sup\mathcal{Y}-m_d(s)\right)\gamma_d(s)\geq \left(\sup\mathcal{Y}-E_{O}[Y|S=s,D=d]\right)P_{O}(S=s,D=d)$ }
    \end{array}\right\}
\end{split}
\end{align}
\end{theoremEnd}
\begin{proofEnd}

Since $\mathcal{S} = \{1,2,\hdots,k\}$, represent $\gamma_d$ as an element of the $k-$dimensional simplex $\Delta(k)$ and $m_d\in\mathcal{Y}^k$. Let $\gamma_d(s)$ and $m_d(s)$ denote the $s-$th element of the corresponding vectors. Then:
\begin{align}\label{eq:discrete_intermediate}
\begin{split}
        \mathcal{H}(m,\gamma)   &= \left\{\begin{array}{l}
     (m,\gamma)\in\mathcal{M}^A\times(\Delta(k))^2: \text{$\forall d\in\{0,1\}$, $\forall B\in\mathcal{C}(\mathcal{S})$, } \\
         \text{$\gamma_d(B)\geq \max\left(ess\sup_{Z}P_E(S\in B,D=d|Z),P_O(S\in B,D=d) \right)$,}\\
          \text{$\forall u\in\{-1,1\}$: $um_d(s)\leq u\mu_d(s)\pi_{\gamma_d}(s)+ h_{co(\mathcal{Y})}(u)(1-\pi_{\gamma_d}(s))$ $\gamma_d-$a.e.}
    \end{array}\right\}\\
      &= \left\{\begin{array}{l}
     (m,\gamma)\in\mathcal{M}^A\times(\Delta(k))^2: \text{$\forall d\in\{0,1\}$, $\forall B\in\mathcal{C}(\mathcal{S})$, } \\
         \text{$\gamma_d(B)\geq \max\left(ess\sup_{Z}P_E(S\in B,D=d|Z),P_O(S\in B,D=d) \right)$, $\forall u\in\{-1,1\}$:}\\
          \text{ $um_d(s)\leq u\mu_d(s)\frac{P_O(S=s,D=d)}{\gamma_d(s)}+ h_{co(\mathcal{Y})}(u)\left(1-\frac{P_O(S=s,D=d)}{\gamma_d(s)}\right)$ $\gamma_d-$a.e.}
    \end{array}\right\}
\end{split}
\end{align}
 where the first line is by \Cref{thm:tractable_set}. The second is by definition of $\pi_{\gamma_d}(s)$ and $\gamma_d$ being supported on $\mathcal{S}$ with $|\mathcal{S}|<\infty$. $\mathcal{S}$ is closed by definition. Since it is finite, it is bounded. Hence, $\mathbf{S}_d$ is almost surely compact, by definition. Then, by \citet[Lemma B.1]{beresteanu2012partial} $\{\{s\}: s\in\mathcal{S}\}$ is a core-determining class for the containment functional of $\mathbf{S}_d$. Then:
\begin{align}\label{eq:identified_set_worked}
\begin{split}
         \mathcal{H}(m,\gamma)   =&\left\{\begin{array}{l}
        (m,\gamma)\in\mathcal{M}^A\times (\Delta(k))^2: \text{$\forall d\in\{0,1\}$, $ \forall s\in\mathcal{S}$, $\forall u\in\{-1,1\},$}\\
         \text{$\gamma_d(s)\geq \max\left(ess\sup_{Z}P_E(S = s,D=d|Z),P_O(S=s,D=d) \right),$}\\
      \text{$um_d(s)\leq u\mu_d(s)\frac{P_O(S=s,D=d)}{\gamma_d(s)}+ h_{co(\mathcal{Y})}(u)\left(1-\frac{P_O(S=s,D=d)}{\gamma_d(s)}\right) $ $\gamma_d-$a.e.}
    \end{array}\right\}\\
          =&\left\{\begin{array}{l}
        (m,\gamma)\in\mathcal{M}^A\times (\Delta(k))^2: \text{$\forall d\in\{0,1\}$, $ \forall s\in\mathcal{S}$, $\forall u\in\{-1,1\},$}\\
         \text{$\gamma_d(s)\geq \max\left(ess\sup_{Z}P_E(S = s,D=d|Z),P_O(S=s,D=d) \right),$}\\
      \text{$um_d(s)\leq uE[Y|S=s,D=d]\frac{P_O(S=s,D=d)}{\gamma_d(s)}+ h_{co(\mathcal{Y})}(u)\left(1-\frac{P_O(S=s,D=d)}{\gamma_d(s)}\right) $  $\gamma_d-$a.e.}
    \end{array}\right\}\\
    =&\left\{\begin{array}{l}
        (m,\gamma)\in\mathcal{M}^A\times (\Delta(k))^2: \text{$\forall d\in\{0,1\}$, $ \forall s\in\mathcal{S}$,}\\
         \text{$\gamma_d(s)\geq \max\left(ess\sup_{Z}P_E(S = s,D=d|Z),P_O(S=s,D=d) \right),$}\\
    \text{$\left(m_d(s)-\inf\mathcal{Y}\right)\gamma_d(s)\geq \left(E_{O}[Y|S=s,D=d]-\inf\mathcal{Y}\right)P_{O}(S=s,D=d)$ $\gamma_d-$a.e.},\\
     \text{$\left(\sup\mathcal{Y}-m_d(s)\right)\gamma_d(s)\geq \left(\sup\mathcal{Y}-E_{O}[Y|S=s,D=d]\right)P_{O}(S=s,D=d)$ $\gamma_d-$a.e.}
    \end{array}\right\}\\
    =&\left\{\begin{array}{l}
        (m,\gamma)\in\mathcal{M}^A\times (\Delta(k))^2: \text{$\forall d\in\{0,1\}$, $ \forall s\in\mathcal{S}$,}\\
         \text{$\gamma_d(s)\geq \max\left(ess\sup_{Z}P_E(S = s,D=d|Z),P_O(S=s,D=d) \right),$}\\
    \text{$\left(m_d(s)-\inf\mathcal{Y}\right)\gamma_d(s)\geq \left(E_{O}[Y|S=s,D=d]-\inf\mathcal{Y}\right)P_{O}(S=s,D=d)$},\\
     \text{$\left(\sup\mathcal{Y}-m_d(s)\right)\gamma_d(s)\geq \left(\sup\mathcal{Y}-E_{O}[Y|S=s,D=d]\right)P_{O}(S=s,D=d)$}
    \end{array}\right\}.
\end{split}
\end{align}
where the first line is by \Cref{thm:tractable_set} and \eqref{eq:discrete_intermediate}, the second line is by definition of $\mu_d(s)$, the third is by definition of $h_{co(\mathcal{Y})}(u)$ and rearrangement, and the fourth is by observation. The result then follows from \Cref{thm:tractable_set} and \Cref{thm:interval_set_general}.
\end{proofEnd}

When short-term outcomes have finite support, the infinite-dimensional problems in \Cref{thm:interval_set_general} reduce to finite-dimensional \textit{generalized bilinear} optimization programs (\citet{al1992generalized}, for other examples of bilinear programs see \citet{dutz2021selection} and \citet{shea2022testing}). Although such problems are generally nonconvex, modern general-purpose solvers can compute globally optimal solutions via spatial branch-and-bound methods (e.g. \citet{gurobi}). \Cref{sect:estimation} provides a consistent criterion-based estimator based on the corollary.

Focusing on cases where relevant variables are finitely supported or discretized has been the predominant approach in work that relies on Artstein's theorem (\citet{galichon2011set}, \citet{russell2021sharp}, \citet{luo2024selecting}). Discretization is also common in related settings (\citet{rambachan2024program}, \citet{park2024informativeness}). Appendix \ref{sect:discretization} discusses the interpretation of results when $S$ is discretized. General-purpose methods for solving bilinear infinite-dimensional programs remain an interesting open problem. In special cases, ongoing work shows that additional structure permits tractable implementation via optimal transport theory with approximation guarantees even when $\mathcal{S}$ is infinite; see \Cref{rem:simple_optimization}. 

A few details are worth highlighting. First, $\mathcal{Y}$ is unrestricted. When it is unbounded, then the $m_d$ constraints involving $\mathcal{Y}$ are trivially satisfied. Second, to reduce computational complexity, the corollary uses $\{\{s\}:s\in\mathcal{S}\}$ as a CDC. \Cref{tab:complexity} illustrates the resulting reduction as a function of $|\mathcal{S}|$.\footnote{The number of constraints on $m$ imposed by the data given $\gamma$ is $4|\mathcal{S}|$; the total number depends on the modeling assumptions. The number of constraints may potentially be reduced further by adapting methods such as \citet{luo2024selecting}, as $\{\{s\}:s\in\mathcal{S}\}$ is not necessarily the smallest CDC.} If $S(d)$ represents percentiles, then not using the CDC already results in a prohibitively complex constraint set.

\begin{table}[h!]
\centering
\caption{Number of constraints on $\gamma$ in $\mathcal{H}(m,\gamma)$.}
\begin{tabular}{c c c c c c}
\toprule
& \multicolumn{5}{c}{$|\mathcal{S}|$} \\
\cline{2-6}
Constraint $\#$ for $\gamma$ & $2$ & $5$ & $10$ & $20$ & $100$ \\
\hline
Without CDC & $8$ & $64$ & $2048$ & $2{,}097{,}152$ & $>10^{30}$ \\
With CDC $\{\{s\}:s\in\mathcal{S}\}$ & $2$ & $8$ & $18$ & $38$ & $198$ \\
\bottomrule
\end{tabular}
\label{tab:complexity}
\end{table}

\begin{remark}\label{rem:simple_optimization}
Problems $\sup/\inf_{(\Tilde{m},\Tilde{\gamma})\in\mathcal{H}(m,\gamma)} T(\Tilde{m},\Tilde{\gamma})$ become linear in certain cases. This substantially simplifies computation and occurs under either of the following two conditions. First, if $\mathcal{H}(m,\gamma)=\{m\}\times\mathcal{H}(\gamma)$, the problems reduce to linear programs over $\gamma$. For example, assumptions that point-identify $m$ independently of $\gamma$, such as LUC in \Cref{rem:sufficient_conditions}, yield $\mathcal{H}(m,\gamma)=\{m\}\times\mathcal{H}(\gamma)$. In this case, one may use optimal transport theory to characterize the identified set even when $\mathcal{S}$ is infinite \citep{voronin2025generalized}. Second, if $\mathcal{H}(m,\gamma)=\mathcal{H}(m)\times\{\gamma\}$ and $\mathcal{M}^A$ is representable by linear constraints, the problems reduce to linear programs over $m$. This case arises under Assumptions \ref{ass:LMIV} and \ref{ass:TI} when the right-hand sides of the constraints for each $\gamma_d$ sum to one, as under perfect compliance or when the datasets jointly point-identify $\gamma$.
\end{remark}

\subsubsection{Identifying Power}

Closed-form bounds, when available, can provide intuition about the sources of identifying power. Optimization-based bounds typically trade off this transparency for the ability to computationally deliver sharp bounds under a wide range of assumptions (see, e.g., \citet{mogstad2018using}, \citet{torgovitsky2019nonparametric}). Under Assumptions \ref{ass:LMIV} and \ref{ass:TI}, however, one can still shed some light on the sources of identifying power by characterizing $m^{UB}$ and $m^{LB}$ that attain the upper and lower bounds on $\tau$, given a feasible $\gamma$. To simplify exposition, focus on the bounds for $E[Y(1)]$ for a fixed feasible $\gamma$, as the intuition for $E[Y(0)]$ (and hence $\tau$) is analogous.\footnote{If compliance is perfect, there is a unique feasible $\gamma$. Otherwise, there may be a set of feasible $\gamma$, and the bounds are obtained by taking the union over all feasible $\gamma$.} Based on \Cref{lem:min_max_selectors} in Appendix \ref{sect:appendix_lemmas} and letting $s'\leq s$ be the product order, under Assumption \ref{ass:LMIV}:
    \begin{align}
    \begin{split}
        m_1^{LB}(s) &= \sup_{s'\leq s}\left[E_O[Y|S=s',D=1]\pi_{\gamma_1}(s')+\inf\mathcal{Y}\left(1-\pi_{\gamma_1}(s')\right)\right],\\
        m_1^{UB}(s) &= \inf_{s'\geq s}\left[E_O[Y|S=s',D=1]\pi_{\gamma_1}(s')+\sup\mathcal{Y}\left(1-\pi_{\gamma_1}(s')\right)\right].
    \end{split}
    \end{align}
These represent the extremal monotonic $m_1$ that are compatible with worst-case (WC) bounds on $m_1$ given $\gamma$ in \eqref{eq:m_wc}. Similarly, under Assumption \ref{ass:TI}:
\begin{align}
    \begin{split}
        m_1^{LB}(s) &= \sup_{d\in\{0,1\}}\left[E_O[Y|S=s,D=d]\pi_{\gamma_d}(s)+\inf\mathcal{Y}\left(1-\pi_{\gamma_d}(s)\right)\right],\\
        m_1^{UB}(s) &= \inf_{d\in\{0,1\}}\left[E_O[Y|S=s,D=d]\pi_{\gamma_d}(s)+\sup\mathcal{Y}\left(1-\pi_{\gamma_d}(s)\right)\right].
    \end{split}
    \end{align}
Under \ref{ass:TI}, $m_1$ must satisfy the WC bounds for both $m_1$ and $m_0$. Hence the bound on it is given by the intersection of the two WC bounds at each support point $s$.

Based on these results, identifying power under Assumption \ref{ass:LMIV} comes from the extent to which the WC bounds on $m_1$ violate monotonicity. Imposing monotonicity replaces the WC lower (upper) endpoint by its left-envelope $\sup_{s'\le s}(\cdot)$ (right-envelope $\inf_{s'\ge s}(\cdot)$), which raises the lower bound (lowers the upper bound) wherever the WC endpoints are decreasing. Under Assumption \ref{ass:TI}, identifying power comes from how small the intersection of the WC bounds for $m_1$ and $m_0$ is at each $s$. The bounds tighten when those intervals overlap little, e.g., when some $\pi_{\gamma_d}(s)$ are large, so the corresponding WC interval is narrow. Both mechanisms are illustrated for a non-pathological DGP in \Cref{fig:liv_ti_illustration}. Finally, bounded $\mathcal{Y}$ is necessary for either assumption to yield identifying power, since otherwise WC bounds are trivial for any feasible $\gamma$.

 \begin{figure}[H]                                                                                                 
      \centering                                                      
      \begin{subfigure}[b]{0.4\textwidth}                                                                                
          \centering                                                                                                      
          \includegraphics[width=\textwidth]{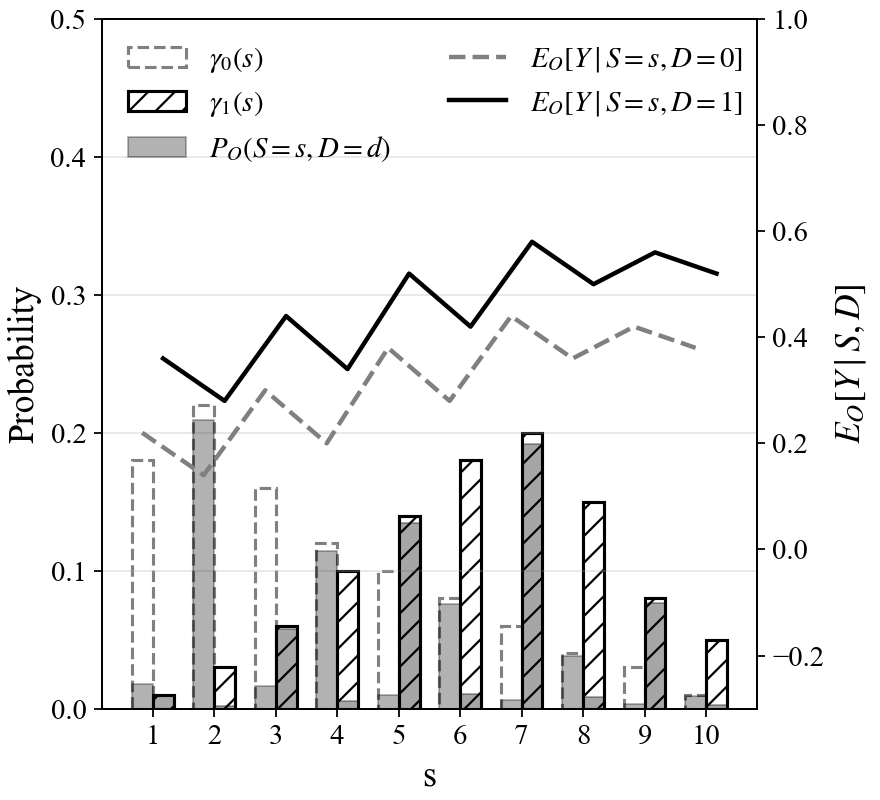}
          \caption{Data / DGP}                                                                                            
          \label{fig:liv_ti_dgp}                            
      \end{subfigure}
      \hfill
      \begin{subfigure}[b]{0.4\textwidth}
          \centering
          \includegraphics[width=\textwidth]{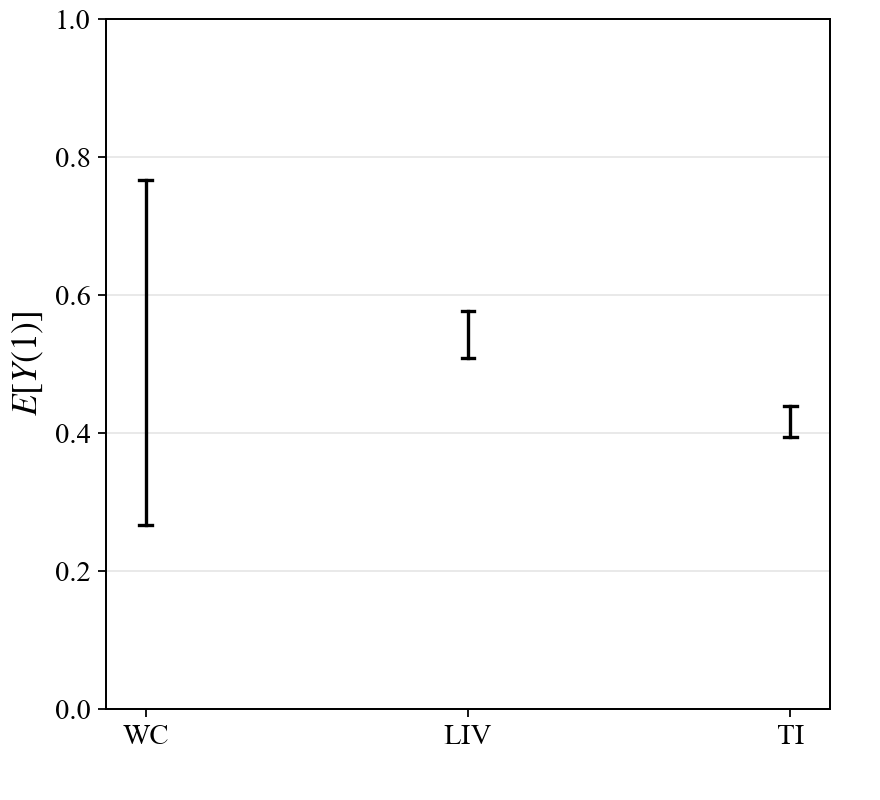}
            \caption{Bounds on $E[Y(1)]$}
          \label{fig:liv_ti_ey}
      \end{subfigure}

      \begin{subfigure}[b]{0.4\textwidth}
          \centering
          \includegraphics[width=\textwidth]{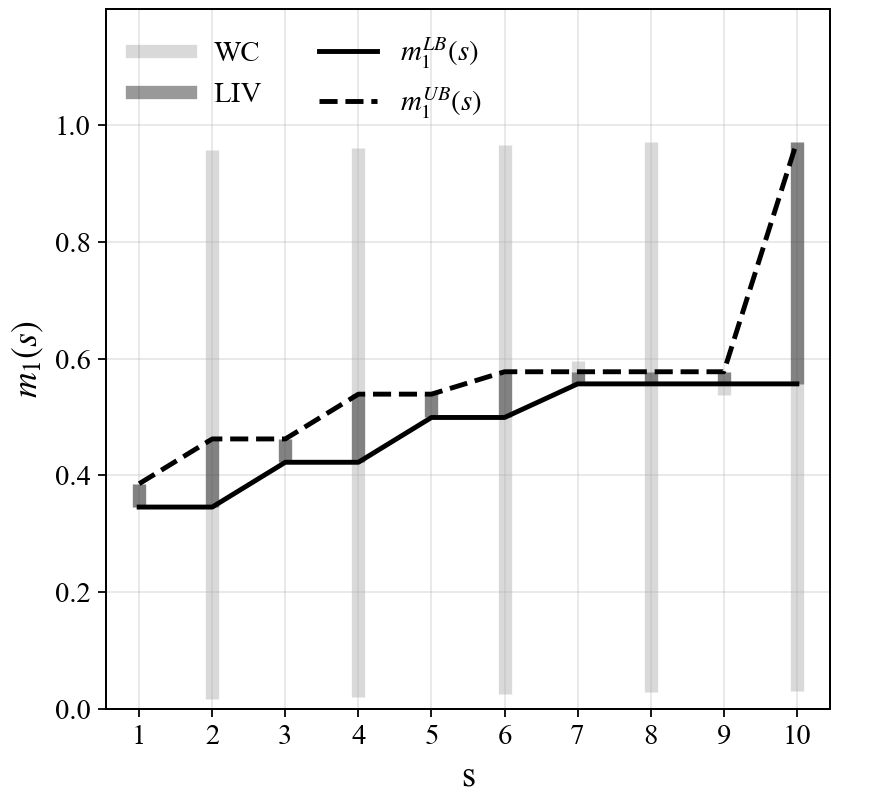}
        \caption{LIV bounds on $m_1(s)$}
          \label{fig:liv_ti_liv}
      \end{subfigure}
      \hfill
      \begin{subfigure}[b]{0.4\textwidth}
          \centering
          \includegraphics[width=\textwidth]{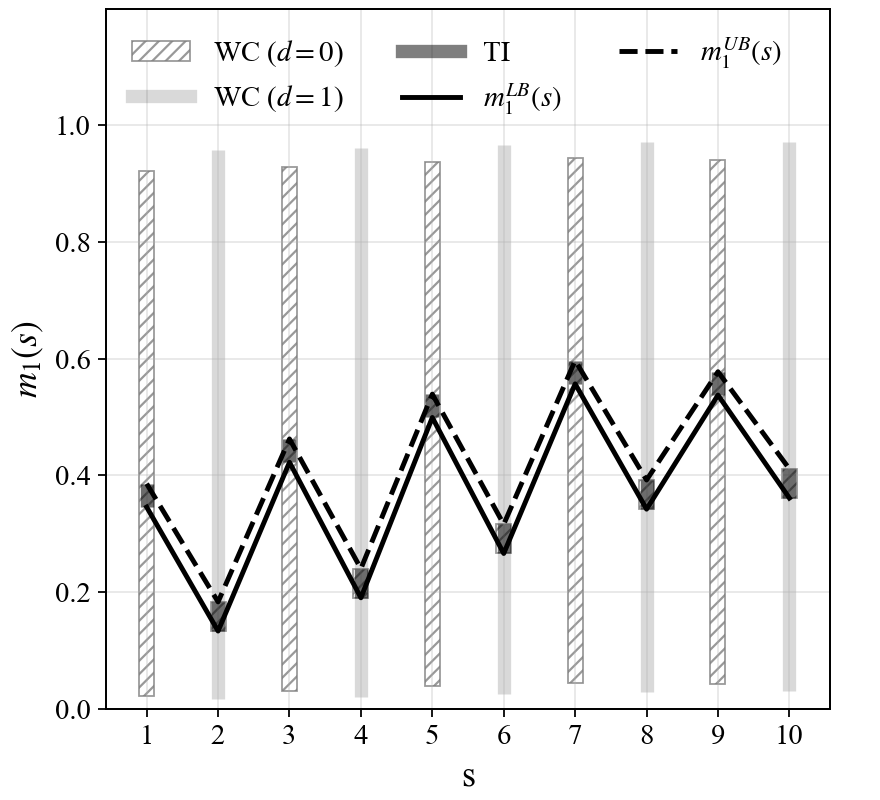}
          \caption{TI bounds on $m_1(s)$}
          \label{fig:liv_ti_ti}
      \end{subfigure}
      \caption{Illustrative example: LIV and TI bounds}
\footnotesize \emph{Notes:} Panel \ref{fig:liv_ti_dgp} specifies the DGP and $\gamma$. Panels \ref{fig:liv_ti_liv} and \ref{fig:liv_ti_ti} depict the worst-case (WC) bounds on $m_d$ given by \eqref{eq:m_wc}, $m_1^{LB}$ and $m_1^{UB}$ which deliver the lower and upper bounds on $E[Y(1)]$.
      \label{fig:liv_ti_illustration}
  \end{figure}

\begin{remark}
Both Assumptions \ref{ass:LMIV} and \ref{ass:TI} can yield point identification of $m$, given $\gamma$. For Assumption \ref{ass:TI}, an important case is when $D$ is constant in the observational data ($G=O$), for instance, when one treatment is not available in that dataset. For Assumption \ref{ass:LMIV}, point identification arises when the WC endpoints are sufficiently non-monotonic that imposing nondecreasing $m_d$ makes it constant (since the lower and upper monotone envelopes coincide). This may happen, for example, when the WC endpoints are strictly decreasing in $s$, or when the endpoints exhibit more pronounced oscillations than in \Cref{fig:liv_ti_illustration}.
\end{remark}

\section{Empirical Illustration: Long-term Effects of Head Start Participation}\label{sect:app}

Head Start is the largest early childhood education program in the United States, serving approximately 730{,}000 low-income preschool-age children in 2023.\footnote{Link: \href{https://headstart.gov/program-data/article/head-start-program-facts-fiscal-year-2023}{https://headstart.gov/program-data/article/head-start-program-facts-fiscal-year-2023} (Last accessed 01/14/2025).} Established in 1965 as part of the ``War on Poverty,'' it was intended to help narrow gaps between disadvantaged and more advantaged children on a national scale. Its long-term effects have been studied extensively, primarily using observational designs. A common approach compares siblings within the same family who did and did not participate in Head Start (\citet{currie1995does}, \citet{garces2002longer}, \citet{deming2009early}, \citet{bauer2016long}), while others exploit variation in program funding, income-based eligibility, or rollout timing to identify local average treatment effects (\citet{ludwig2007does}, \citet{carneiro2014long}, \citet{bailey2021prep}). \citet{kline2016evaluating} instead monetize the experimental LATE of Head Start on test scores using estimates of the test-score--earnings relationship from administrative follow-up of Tennessee Project STAR participants (\citet{chetty2011does}). Despite this long history, the literature has yet to reach consensus on the program’s long-term impacts (\citet{gibbs2011does}, \citet{pages2020elusive}), as the assumptions and generalizability of parameters identified by existing approaches remain debated (\citet{ludwig2008long}, \citet{elango2015early}, \citet{gonzalez2020within}, \citet{garcia2020quantifying}, \citet{miller2023selection}). I illustrate an alternative approach by applying the proposed method to estimate long-term average treatment effects of Head Start for eligible individuals under assumptions that do not restrict selection into treatment.

\subsection{Data}

I combine individual-level data from the Head Start Impact Study (HSIS) and the Child and Young Adult Supplement to the National Longitudinal Survey of Youth 1979 cohort (CNLSY).

HSIS was an experimental trial of Head Start mandated by the 105th US Congress as part of the program’s 1998 reauthorization. In fall of 2002, $4{,}667$ children from nationally representative cohorts aged 3 and 4 were enrolled in the experiment. Across 383 randomly selected Head Start centers, participants were randomized either to a treatment group assigned to enroll in Head Start or to a control group barred from enrolling, yielding $Z\in\{0,1\}$. Since participants were followed only through third grade, HSIS does not contain adolescence or adulthood outcomes that may be of interest. However, it does contain childhood math and reading test scores. I follow \citet{kline2016evaluating} and \citet{kamat2024identifying} by pooling all children into a single cohort, then retaining only those with available scores, resulting in an experimental sample size of $n_E = 3{,}540$.

CNLSY is a biennial longitudinal survey introduced in 1986 that has tracked 11,545 children born to participants in the National Longitudinal Survey of Youth 1979 cohort (NLSY79), which, like HSIS, was designed to be nationally representative. It reveals long-term outcomes previously studied in the literature and non-randomized Head Start participation. As in \citet{deming2009early}, I consider eight long-term outcomes: grade repetition, diagnosis of a learning disability, high school graduation, ``idleness'', criminal involvement, teenage parenthood, self-reported health status, and average earnings. CNLSY also contains comparable childhood math and reading test scores that can be used as $S$, in conjunction with HSIS.\footnote{Fadeout, i.e. disappearance of \textit{average} effects on test-scores over time,  need not invalidate their use as $S$. Here, distributional effects are relevant, not only mean effects, as captured by $\gamma$ (see also \citet{bitler2014experimental}). More generally, even if there are no distributional effects on $S$, $m_d$ may vary across $d$ for the scores.} Following \citet[Section IIA]{carneiro2014long}, I construct the observational sample by selecting individuals who were either eligible for Head Start or reported having participated in the program. This yields an observational sample of size $n_O=2{,}535$. \Cref{app:data_cleaning} provides additional details on data construction.

\Cref{tab:summary_stats} presents descriptive statistics for the two samples. They have comparable gender compositions and similar shares of white and non-white individuals. Moreover, both NLSY79 and HSIS were designed to be nationally representative, and both pertain to the same treatment, lending support to Assumption \ref{ass:ex_validity}. While lower than in the experiment, a significant proportion of CNLSY individuals participate in the program. The rate of compliance with the assigned treatment is $83.8\%$ in the experiment. Imperfect compliance and the availability of the treatment in the observational population preclude direct application of existing data-combination methods that assume either perfect compliance or that the treatment is available only in the experiment. Nevertheless, the method developed here can still be used to estimate the effects under their corresponding identifying assumptions. I illustrate this in the following section by reporting bounds under relevant restrictions following from \citet{athey2025combining} and \citet{garcia2020quantifying}.\footnote{Both assume perfect compliance. In \citet{garcia2020quantifying}, the intervention is available only in the experiment, which is typical of “model” early childhood intervention programs such as the Carolina Abecedarian and Perry Preschool Projects, but not of large-scale programs such as Head Start.}
Despite imperfect compliance, the estimated lower bounds on 
$\gamma_d$ are close to summing to one, suggesting that the combined data are highly informative about the short-term potential outcome distributions.\footnote{In the empirical implementation, I fix 
$\gamma_d$ at these bounds. This linearizes the optimization programs and substantially simplifies computation without resorting to coarse discretization; see Appendix \ref{sect:computational}.}

\begin{table}[htbp]\centering
\caption{Summary Statistics}\label{tab:summary_stats}
\begin{threeparttable}
\begin{tabular}{lcccc}
\toprule
& \multicolumn{2}{c}{HSIS} & \multicolumn{2}{c}{CNLSY} \\
\cmidrule(lr){2-3} \cmidrule(lr){4-5}
Variable & Mean & SD & Mean & SD \\
\toprule
\multicolumn{5}{l}{\textbf{A} \textit{Individual Characteristics}} \\
\hspace{1em}Male                &       0.504&       0.500&       0.512&       0.500\\
\hspace{1em}White               &       0.319&       0.466&       0.278&       0.448\\
\hspace{1em}Math Score          &      50.921&      24.649&      40.963&      26.066\\
\hspace{1em}Reading Score       &      55.225&      24.413&      52.500&      25.853\\
\hspace{1em}Repeat Grade        &       -   &        -   &       0.320&       0.467\\
\hspace{1em}Learning Disability \hspace{3em}&       -   &        -   &       0.057&       0.231\\
\hspace{1em}HS Graduate         &       -   &        -   &       0.847&       0.360\\
\hspace{1em}Idle                &       -   &        -   &       0.173&       0.379\\
\hspace{1em}Crime               &       -   &        -   &       0.389&       0.488\\
\hspace{1em}Teen Pregnancy      &       -      &        -    &      0.244&       0.430\\
\hspace{1em}Poor Health         &       -   &        -   &       0.166&       0.373\\
\hspace{1em}Average Earnings (in 000)&    -        &      -      &   22.166&   17.237\\
\midrule
\multicolumn{5}{l}{\textbf{B} \textit{Program Characteristics}} \\
\hspace{1em}$D$                     &       0.531&       0.499&       0.443&       0.497\\
\hspace{1em}$Z$                     &       0.596&       0.491&     -       &     -       \\
\hspace{1em}$\mathbbm{1}[D=Z]$                     &       0.838&       0.369&     -       &      -      \\
\midrule
Observations & \multicolumn{2}{c}{3,540} & \multicolumn{2}{c}{2,535} \\
\bottomrule
\end{tabular}
\begin{tablenotes}[flushleft]
\footnotesize
\item\textit{Notes}: Summary statistics from the Head Start Impact Study (HSIS) and the Child and Young Adult Supplement of the National Longitudinal Survey of Youth 1979 cohort (CNLSY). $D$ denotes Head Start participation and $Z$ is experimental treatment assignment. $\mathbbm{1}$ denotes the indicator variable, and SD the sample standard deviation.
\end{tablenotes}
 \end{threeparttable}
\end{table}

\subsection{Results}

\Cref{tab:bounds_estimates} reports bound estimates using the consistent estimator proposed in \Cref{sect:estimation} under previously discussed assumptions. Worst-case bounds impose no restrictions on the temporal link functions $m$, coincide with the bounds of \citet{manski1990nonparametric}, as shown in \Cref{sect:roles}, and therefore necessarily include zero. Thus, identifying even the sign of the effect requires additional assumptions. For all outcomes, each of the previously mentioned assumptions substantially reduces the width of the estimated bounds.

\begin{table}[htbp]
\centering
\caption{Bounds Estimates}
\label{tab:bounds_estimates}
\begin{threeparttable}
\begin{tabular}{lccccc}
\toprule
{} &  & \multicolumn{4}{c}{Modeling Assumption} \\
\cmidrule(lr){3-6}
Outcome & $n_O$ & \phantom{-}Worst-case & \phantom{-}LIV & \phantom{-}TI & \phantom{-}LUC \\
\midrule
\multirow{2}{*}{Repeat Grade}
& \multirow{2}{*}{${2{,}455}$} & $-0.476$ & $-0.053$ & $-0.075$ & $-0.008$ \\
& & $\phantom{-}0.524$ & $-0.011$ & $\phantom{-}0.050$ &$-0.008$ \\[4pt]
\multirow{2}{*}{Learning Disability}
& \multirow{2}{*}{${2{,}527}$} & $-0.448$ & $-0.008$ & $-0.076$ & $-0.001$ \\
& & $\phantom{-}0.552$ & $\phantom{-}0.008$ & $\phantom{-}0.050$ & $-0.001$ \\[4pt]
\multirow{2}{*}{HS Graduate}
& \multirow{2}{*}{${2{,}031}$} & $-0.525$ & $\phantom{-}0.019$ & $-0.079$ & $-0.002$ \\
& & $\phantom{-}0.475$ & $\phantom{-}0.032$ & $\phantom{-}0.050$ &$-0.002$ \\[4pt]
\multirow{2}{*}{Idle}
& \multirow{2}{*}{${2{,}517}$} & $-0.474$ & $-0.046$ & $-0.072$ & $-0.030$ \\
& & $\phantom{-}0.526$ & $-0.015$ & $\phantom{-}0.051$ & $-0.030$\\[4pt]
\multirow{2}{*}{Crime}
& \multirow{2}{*}{${2{,}517}$} & $-0.499$ & $-0.040$ & $-0.074$ & $-0.032$ \\
& & $\phantom{-}0.501$ & $-0.012$ & $\phantom{-}0.048$ & $-0.032$\\[4pt]
\multirow{2}{*}{Teen Pregnancy}
& \multirow{2}{*}{${2{,}517}$} & $-0.460$ & $\phantom{-}0.003$ & $-0.072$ & $\phantom{-}0.013$ \\
& & $\phantom{-}0.540$ & $\phantom{-}0.022$ & $\phantom{-}0.050$ &$\phantom{-}0.013$ \\[4pt]
\multirow{2}{*}{Poor Health}
& \multirow{2}{*}{${2{,}517}$} & $-0.466$ & $-0.050$ & $-0.078$ & $-0.010$ \\
& & $\phantom{-}0.534$ & $-0.029$ & $\phantom{-}0.044$ &$-0.010$ \\[4pt]
\multirow{2}{*}{Average earnings (in 000)}
& \multirow{2}{*}{${2{,}395}$} & $-110.895$ & $-3.841$ & $-19.413$ & $-2.449$ \\
& & $\phantom{-}131.208$ & $\phantom{-}1.621$ & $\phantom{-}11.583$ &$-2.449$ \\
\bottomrule
\end{tabular}
\begin{tablenotes}[flushleft]
\footnotesize
\item\textit{Notes}: Estimated bounds, represented as $\genfrac{}{}{0pt}{}{\text{Lower Bound}}{\text{Upper Bound}}$, for different long-term outcomes. Worst-case bounds impose no restrictions on $m$. The remaining bounds impose only the noted modeling assumption.
\end{tablenotes}
\end{threeparttable}
\end{table}

Assumption \ref{ass:LMIV} maintains that temporal link functions are monotonic in each of the two test-score components. For high school graduation and earnings, I impose weak monotonicity increasing in each potential test score. For grade repetition, learning disability diagnosis, ``idleness'', crime, teen pregnancy, and poor health, I impose weak monotonicity decreasing in each potential test score. This assumption may be particularly appealing for outcomes related to education, employment, crime, and earnings. Taking high school graduation as an example, it would hold if, fixing any Head Start participation, individuals with higher math or reading scores are equally or more likely to graduate from high school than individuals with lower scores.

Estimates under Assumption \ref{ass:LMIV} reveal the sign of the effect for all but two outcomes. They further indicate that Head Start participation reduces the probability of grade repetition by at least $1.1$ percentage points (pp), idleness by at least $1.5$ pp, criminal involvement by at least $1.2$ pp, and reporting poor health by at least $2.9$ pp. The corresponding upper bounds are also quite informative. Overall, the estimates suggest that beneficial effects on grade repetition, learning disability, high school graduation, idleness, and poor health may be more modest than reported by the sibling study in \citet{deming2009early}. Compared to the same study, the effect on criminal involvement found here has the expected sign, while the impact on teen pregnancy does not. Moreover, because the present data allow for longer follow-up than in \citet{deming2009early}, they reveal earnings for a larger fraction of individuals. However, the sign of the effect on earnings remains unidentified.

Next, I turn to illustrative results relying on temporal link function restrictions that follow from previously proposed methods. These methods cannot be directly utilized here because of imperfect compliance and the availability of the treatment in CNLSY. However, as argued above, the proposed framework can be combined with the relevant assumptions, thereby extending their applicability. Estimates under Assumption \ref{ass:TI} also lead to a substantial reduction in the width of the bounds, though none exclude zero. Assumption \ref{ass:TI} would hold if the model proposed by \citet{garcia2020quantifying} is credible. That said, HSIS does not contain all of the short-term outcomes included by the authors in their model of the Carolina Abecedarian and CARE programs, nor the medium-term experimental outcomes they use to validate their choice of short-term outcomes. The corresponding results should therefore be interpreted with caution. The final column reports estimates under Assumption LUC from \citet{athey2025combining}. Since this assumption point-identifies $m$ and the $\gamma_d$ are fixed by empirical bounds, the estimated bounds for the LTE collapse to points. However, the estimated signs are not aligned with previous findings for some outcomes. As emphasized by \citet{imbens2024long}, one reason for this may be that the estimates are biased by so-called long-term confounders--unobservables that relate Head Start participation, the short-term test scores, and the long-term outcome.

\section{Conclusion}\label{sect:conclusion}

Recent literature proposes augmenting long-term observational studies with short-term experiments to provide alternatives to conventional long-term observational studies. This paper shows that such data combination is not a substitute for credible modeling assumptions. Nevertheless, it remains appealing for this purpose. Assumptions relating short-term to long-term potential outcomes may be defensible based on economic theory or intuition, and thus conducive to plausible inference. Data combination may be used to amplify the identifying power of such restrictions and thereby may yield more informative plausible inference than observational data alone.

This paper introduces two assumptions that exploit this feature of data combination. It also provides a general identification approach that enables computational derivation of bounds under new modeling assumptions, facilitating further developments. Tailor-made assumptions that are plausible in specific empirical settings are an interesting topic for future research, which may benefit from these results.

\printbibliography

@article{ghassami2022combining,
  title={Combining experimental and observational data for identification and estimation of long-term causal effects},
  author={Ghassami, AmirEmad and Yang, Alan and Richardson, David and Shpitser, Ilya and Tchetgen, Eric Tchetgen},
  journal={arXiv preprint arXiv:2201.10743},
  year={2022}
}

@techreport{athey2024surrogate,
  title={Estimating treatment effects using multiple surrogates: The role of the surrogate score and the surrogate index},
  author={Athey, Susan and Chetty, Raj and Imbens, Guido and Kang, Hyunseung},
  journal={arXiv preprint arXiv:1603.09326},
  year={2024}
}

@article{athey2025combining,
  title={Using Experiments to Correct for Selection in Observational Studies}, 
  author={Susan Athey and Raj Chetty and Guido Imbens},
  year={2025},
  journal={arXiv preprint arXiv:2006.09676},

}

@article{manski1990nonparametric,
  title={Nonparametric bounds on treatment effects},
  author={Manski, Charles F},
  journal={The American Economic Review},
  volume={80},
  number={2},
  pages={319--323},
  year={1990},
  publisher={JSTOR}
}

@article{garcia2020quantifying,
  title={Quantifying the life-cycle benefits of an influential early-childhood program},
  author={Garc{\'\i}a, Jorge Luis and Heckman, James J and Leaf, Duncan Ermini and Prados, Mar{\'\i}a Jos{\'e}},
  journal={Journal of Political Economy},
  volume={128},
  number={7},
  pages={2502--2541},
  year={2020},
  publisher={The University of Chicago Press Chicago, IL}
}

@article{hu2022identification,
  title={Identification and estimation of treatment effects on long-term outcomes in clinical trials with external observational data},
  author={Hu, Wenjie and Zhou, Xiaohua and Wu, Peng},
  journal={arXiv preprint arXiv:2208.10163},
  year={2022}
}

@article{chen2023semiparametric,
  title={Semiparametric estimation of long-term treatment effects},
  author={Chen, Jiafeng and Ritzwoller, David M},
  journal={Journal of Econometrics},
  volume={237},
  number={2},
  pages={105545},
  year={2023},
  publisher={Elsevier}
}

@article{prentice1989surrogate,
  title={Surrogate endpoints in clinical trials: definition and operational criteria},
  author={Prentice, Ross L},
  journal={Statistics in medicine},
  volume={8},
  number={4},
  pages={431--440},
  year={1989},
  publisher={Wiley Online Library}
}

@article{dynarski2021closing,
  title={Closing the gap: The effect of reducing complexity and uncertainty in college pricing on the choices of low-income students},
  author={Dynarski, Susan and Libassi, CJ and Michelmore, Katherine and Owen, Stephanie},
  journal={American Economic Review},
  volume={111},
  number={6},
  pages={1721--1756},
  year={2021},
  publisher={American Economic Association}
}

@article{heckman2013understanding,
  title={Understanding the mechanisms through which an influential early childhood program boosted adult outcomes},
  author={Heckman, James and Pinto, Rodrigo and Savelyev, Peter},
  journal={American Economic Review},
  volume={103},
  number={6},
  pages={2052--2086},
  year={2013},
  publisher={American Economic Association}
}

@article{park2024bracketing,
  title={A Bracketing Relationship for Long-Term Policy Evaluation with Combined Experimental and Observational Data},
  author={Park, Yechan and Sasaki, Yuya},
  journal={arXiv preprint arXiv:2401.12050},
  year={2024}
}

@article{park2024informativeness,
      title={The Informativeness of Combined Experimental and Observational Data under Dynamic Selection}, 
      author={Yechan Park and Yuya Sasaki},
        journal={arXiv preprint arXiv:2403.16177},
      year={2024}
}

@article{chesher2017generalized,
  title={Generalized instrumental variable models},
  author={Chesher, Andrew and Rosen, Adam M},
  journal={Econometrica},
  volume={85},
  number={3},
  pages={959--989},
  year={2017},
  publisher={Wiley Online Library}
}

@book{manski2009identification,
  title={Identification for prediction and decision},
  author={Manski, Charles F},
  year={2009},
  publisher={Harvard University Press}
}

@article{Aizer2024,
  title={The Lifetime Impacts of the New Deal's Youth Employment Program},
  author={Aizer, Anna and Early, Nancy and Eli, Shari and Imbens, Guido and Lee, Keyoung and Lleras-Muney, Adriana and Strand, Alexander},
  journal={The Quarterly Journal of Economics},
  year={2024},
  publisher={Oxford University Press}
}

@article{imbens2024long,
  title={Long-term causal inference under persistent confounding via data combination},
  author={Imbens, Guido and Kallus, Nathan and Mao, Xiaojie and Wang, Yuhao},
  journal={arXiv preprint arXiv:2202.07234},
  year={2024}
}

@article{kline2016evaluating,
  title={Evaluating public programs with close substitutes: The case of Head Start},
  author={Kline, Patrick and Walters, Christopher R},
  journal={The Quarterly Journal of Economics},
  volume={131},
  number={4},
  pages={1795--1848},
  year={2016},
  publisher={MIT Press}
}

@article{gupta2019top,
  title={Top challenges from the first practical online controlled experiments summit},
  author={Gupta, Somit and Kohavi, Ronny and Tang, Diane and Xu, Ya and Andersen, Reid and Bakshy, Eytan and Cardin, Niall and Chandran, Sumita and Chen, Nanyu and Coey, Dominic and others},
  journal={ACM SIGKDD Explorations Newsletter},
  volume={21},
  number={1},
  pages={20--35},
  year={2019},
  publisher={ACM New York, NY, USA}
}

@book{molchanov2018random,
  title={Random Sets in Econometrics},
  author={Molchanov, Ilya and Molinari, Francesca},
  volume={60},
  year={2018},
  publisher={Cambridge University Press}
}

@article{molchanov2014applications,
  title={Applications of random set theory in econometrics},
  author={Molchanov, Ilya and Molinari, Francesca},
  journal={Annu. Rev. Econ.},
  volume={6},
  number={1},
  pages={229--251},
  year={2014},
  publisher={Annual Reviews}
}

@book{molchanov2017theory,
  author    = {Ilya Molchanov},
  title     = {Theory of Random Sets},
  edition   = {2},
  series    = {Probability Theory and Stochastic Modelling},
  volume    = {87},
  year      = {2017},
  publisher = {Springer}
}

@article{beresteanu2012partial,
  title={Partial identification using random set theory},
  author={Beresteanu, Arie and Molchanov, Ilya and Molinari, Francesca},
  journal={Journal of Econometrics},
  volume={166},
  number={1},
  pages={17--32},
  year={2012},
  publisher={Elsevier}
}

@article{artstein1983distributions,
  title={Distributions of random sets and random selections},
  author={Artstein, Zvi},
  journal={Israel Journal of Mathematics},
  volume={46},
  pages={313--324},
  year={1983},
  publisher={Springer}
}

@inproceedings{van2023estimating,
  title={Estimating long-term causal effects from short-term experiments and long-term observational data with unobserved confounding},
  author={Van Goffrier, Graham and Maystre, Lucas and Gilligan-Lee, Ciar{\'a}n Mark},
  booktitle={Conference on Causal Learning and Reasoning},
  pages={791--813},
  year={2023},
  organization={PMLR}
}

@article{imbens1994identification,
  title={Identification and Estimation of Local Average Treatment Effects},
  author={Imbens, Guido W and Angrist, Joshua D},
  journal={Econometrica},
  volume={62},
  number={2},
  pages={467--475},
  year={1994}
}

@article{manski2000monotone,
  title={Monotone Instrumental Variables: With an Application to the Returns to Schooling},
  author={Manski, Charles F and Pepper, John V},
  journal={Econometrica},
  volume={68},
  number={4},
  pages={997--1010},
  year={2000},
  publisher={JSTOR}
}

@article{manski2009more,
  title={More on monotone instrumental variables},
  author={Manski, Charles F and Pepper, John V},
  journal={The Econometrics Journal},
  volume={12},
  number={suppl\_1},
  pages={S200--S216},
  year={2009},
  publisher={Oxford University Press Oxford, UK}
}

@article{galichon2011set,
  title={Set identification in models with multiple equilibria},
  author={Galichon, Alfred and Henry, Marc},
  journal={The Review of Economic Studies},
  volume={78},
  number={4},
  pages={1264--1298},
  year={2011},
  publisher={Oxford University Press}
}

@article{heckman1999local,
  title={Local instrumental variables and latent variable models for identifying and bounding treatment effects},
  author={Heckman, James J and Vytlacil, Edward J},
  journal={Proceedings of the national Academy of Sciences},
  volume={96},
  number={8},
  pages={4730--4734},
  year={1999},
  publisher={National Acad Sciences}
}

@book{rockafellar1970convex,
url = {https://doi.org/10.1515/9781400873173},
title = {Convex Analysis},
author = {Ralph Tyrell Rockafellar},
publisher = {Princeton University Press},
address = {Princeton},
doi = {doi:10.1515/9781400873173},
isbn = {9781400873173},
year = {1970},
lastchecked = {2024-08-09}
}

@article{russell2021sharp,
  title={Sharp bounds on functionals of the joint distribution in the analysis of treatment effects},
  author={Russell, Thomas M},
  journal={Journal of Business \& Economic Statistics},
  volume={39},
  number={2},
  pages={532--546},
  year={2021},
  publisher={Taylor \& Francis}
}

@incollection{currie2011human,
  title={Human capital development before age five},
  author={Currie, Janet and Almond, Douglas},
  booktitle={Handbook of labor economics},
  volume={4},
  pages={1315--1486},
  year={2011},
  publisher={Elsevier}
}

@article{garces2002longer,
  title={Longer-term effects of Head Start},
  author={Garces, Eliana and Thomas, Duncan and Currie, Janet},
  journal={American economic review},
  volume={92},
  number={4},
  pages={999--1012},
  year={2002},
  publisher={American Economic Association}
}

@article{bauer2016long,
  title={The long-term impact of the Head Start program},
  author={Bauer, Lauren and Schanzenbach, Diane Whitmore},
  journal={The Hamilton Project},
  year={2016},
  publisher={Brookings Institution}
}

@article{carneiro2014long,
  title={Long-term impacts of compensatory preschool on health and behavior: Evidence from Head Start},
  author={Carneiro, Pedro and Ginja, Rita},
  journal={American Economic Journal: Economic Policy},
  volume={6},
  number={4},
  pages={135--173},
  year={2014},
  publisher={American Economic Association}
}

@article{deming2009early,
  title={Early childhood intervention and life-cycle skill development: Evidence from Head Start},
  author={Deming, David},
  journal={American Economic Journal: Applied Economics},
  volume={1},
  number={3},
  pages={111--134},
  year={2009},
  publisher={American Economic Association}
}

@article{ludwig2007does,
  title={Does Head Start improve children's life chances? Evidence from a regression discontinuity design},
  author={Ludwig, Jens and Miller, Douglas L},
  journal={The Quarterly journal of economics},
  volume={122},
  number={1},
  pages={159--208},
  year={2007},
  publisher={MIT Press}
}

@article{currie1995does,
  title={Does Head Start Make a Difference?},
  author={Currie, Janet and Thomas, Duncan},
  journal={The American Economic Review},
  pages={341--364},
  year={1995},
  publisher={JSTOR}
}

@article{pages2020elusive,
  title={Elusive longer-run impacts of head start: Replications within and across cohorts},
  author={Pages, Remy and Lukes, Dylan J and Bailey, Drew H and Duncan, Greg J},
  journal={Educational Evaluation and Policy Analysis},
  volume={42},
  number={4},
  pages={471--492},
  year={2020},
  publisher={Sage Publications Sage CA: Los Angeles, CA}
}

@article{d2024partially,
  title={Partially linear models under data combination},
  author={D’Haultf{\oe}uille, Xavier and Gaillac, Christophe and Maurel, Arnaud},
  journal={Review of Economic Studies},
  pages={rdae022},
  year={2024},
  publisher={Oxford University Press UK}
}

@article{todd2023best,
  title={The best of both worlds: combining randomized controlled trials with structural modeling},
  author={Todd, Petra E and Wolpin, Kenneth I},
  journal={Journal of Economic Literature},
  volume={61},
  number={1},
  pages={41--85},
  year={2023},
  publisher={American Economic Association}
}

@article{cross2002regressions,
  title={Regressions, short and long},
  author={Cross, Philip J and Manski, Charles F},
  journal={Econometrica},
  volume={70},
  number={1},
  pages={357--368},
  year={2002},
  publisher={JSTOR}
}

@article{molinari2006generalization, title={Generalization of a Result on 'Regressions, Short and Long}, volume={22}, number={1}, journal={Econometric Theory}, author={Molinari, Francesca and Peski, Marcin}, year={2006}, pages={159–163}}

@article{fan2014identifying,
  title={Identifying treatment effects under data combination},
  author={Fan, Yanqin and Sherman, Robert and Shum, Matthew},
  journal={Econometrica},
  volume={82},
  number={2},
  pages={811--822},
  year={2014},
  publisher={Wiley Online Library}
}

@article{ridder2007econometrics,
  title={The econometrics of data combination},
  author={Ridder, Geert and Moffitt, Robert},
  journal={Handbook of econometrics},
  volume={6},
  pages={5469--5547},
  year={2007},
  publisher={Elsevier}
}

@incollection{chesher2020generalized,
  title={Generalized instrumental variable models, methods, and applications},
  author={Chesher, Andrew and Rosen, Adam M},
  booktitle={Handbook of Econometrics},
  volume={7},
  pages={1--110},
  year={2020},
  publisher={Elsevier}
}

@article{griffen2017assessing,
  title={Assessing the performance of nonexperimental estimators for evaluating Head Start},
  author={Griffen, Andrew S and Todd, Petra E},
  journal={Journal of Labor Economics},
  volume={35},
  number={S1},
  pages={S7--S63},
  year={2017},
  publisher={University of Chicago Press Chicago, IL}
}

@article{torgovitsky2019nonparametric,
  title={Nonparametric inference on state dependence in unemployment},
  author={Torgovitsky, Alexander},
  journal={Econometrica},
  volume={87},
  number={5},
  pages={1475--1505},
  year={2019},
  publisher={Wiley Online Library}
}

@article{mogstad2018using,
  title={Using instrumental variables for inference about policy relevant treatment parameters},
  author={Mogstad, Magne and Santos, Andres and Torgovitsky, Alexander},
  journal={Econometrica},
  volume={86},
  number={5},
  pages={1589--1619},
  year={2018},
  publisher={Wiley Online Library}
}

@article{kamat2024identifying,
  title={Identifying the effects of a program offer with an application to head start},
  author={Kamat, Vishal},
  journal={Journal of Econometrics},
  volume={240},
  number={1},
  pages={105679},
  year={2024},
  publisher={Elsevier}
}

@article{todd2006assessing,
  title={Assessing the impact of a school subsidy program in Mexico: Using a social experiment to validate a dynamic behavioral model of child schooling and fertility},
  author={Todd, Petra E and Wolpin, Kenneth I},
  journal={American economic review},
  volume={96},
  number={5},
  pages={1384--1417},
  year={2006},
  publisher={American Economic Association}
}

@article{attanasio2012education,
  title={Education choices in Mexico: using a structural model and a randomized experiment to evaluate Progresa},
  author={Attanasio, Orazio P and Meghir, Costas and Santiago, Ana},
  journal={The Review of Economic Studies},
  volume={79},
  number={1},
  pages={37--66},
  year={2012},
  publisher={Oxford University Press}
}

@techreport{hoynes2018safety,
  title={Safety net investments in children},
  author={Hoynes, Hilary W and Schanzenbach, Diane Whitmore},
  year={2018},
  institution={National Bureau of Economic Research}
}

@article{shi2015simple,
  title={Simple two-stage inference for a class of partially identified models},
  author={Shi, Xiaoxia and Shum, Matthew},
  journal={Econometric Theory},
  volume={31},
  number={3},
  pages={493--520},
  year={2015},
  publisher={Cambridge University Press}
}

@techreport{dutz2021selection,
  title={Selection in surveys: Using randomized incentives to detect and account for nonresponse bias},
  author={Dutz, Deniz and Huitfeldt, Ingrid and Lacouture, Santiago and Mogstad, Magne and Torgovitsky, Alexander and Van Dijk, Winnie},
  year={2021},
  institution={National Bureau of Economic Research}
}

@article{al1992generalized,
  title={Generalized bilinear programming: Part I. Models, applications and linear programming relaxation},
  author={Al-Khayyal, Faiz A},
  journal={European Journal of Operational Research},
  volume={60},
  number={3},
  pages={306--314},
  year={1992},
  publisher={Elsevier}
}

@article{shea2022testing,
  title={Testing for racial bias in police traffic searches},
  author={Shea, Joshua},
  journal={University of Illinois, Champaign Urbana, USA},
  year={2022}
}

@misc{gurobi,
  author = {{Gurobi Optimization}},
  title = {{Gurobi Optimizer Reference Manual}},
  year = 2024,
  url = "https://www.gurobi.com"
}

@book{nocedal1999numerical,
  title={Numerical optimization},
  author={Nocedal, Jorge and Wright, Stephen J},
  year={1999},
  publisher={Springer}
}

@article{chetty2011does,
  title={How does your kindergarten classroom affect your earnings? Evidence from Project STAR},
  author={Chetty, Raj and Friedman, John N and Hilger, Nathaniel and Saez, Emmanuel and Schanzenbach, Diane Whitmore and Yagan, Danny},
  journal={The Quarterly journal of economics},
  volume={126},
  number={4},
  pages={1593--1660},
  year={2011},
  publisher={MIT Press}
}

@incollection{elango2015early,
  title={Early childhood education},
  author={Elango, Sneha and Garc{\'\i}a, Jorge Luis and Heckman, James J and Hojman, Andr{\'e}s},
  booktitle={Economics of Means-Tested Transfer Programs in the United States, Volume 2},
  pages={235--297},
  year={2015},
  publisher={University of Chicago Press}
}

@article{gonzalez2020within,
  title={Within-family differences in Head Start participation and parent investment},
  author={Gonzalez, Kathryn E},
  journal={Economics of Education Review},
  volume={74},
  pages={101950},
  year={2020},
  publisher={Elsevier}
}

@article{miller2023selection,
  title={Selection into identification in fixed effects models, with application to Head Start},
  author={Miller, Douglas L and Shenhav, Na’ama and Grosz, Michel},
  journal={Journal of Human Resources},
  volume={58},
  number={5},
  pages={1523--1566},
  year={2023},
  publisher={University of Wisconsin Press}
}

@article{bailey2021prep,
  title={Prep school for poor kids: The long-run impacts of Head Start on human capital and economic self-sufficiency},
  author={Bailey, Martha J and Sun, Shuqiao and Timpe, Brenden},
  journal={American Economic Review},
  volume={111},
  number={12},
  pages={3963--4001},
  year={2021},
  publisher={American Economic Association}
}

@article{ludwig2008long,
  title={Long-term effects of Head Start on low-income children},
  author={Ludwig, Jens and Phillips, Deborah A},
  journal={Annals of the New York Academy of Sciences},
  volume={1136},
  number={1},
  pages={257--268},
  year={2008},
  publisher={Wiley Online Library}
}

@techreport{gibbs2011does,
  title={Does Head Start do any lasting good?},
  author={Gibbs, Chloe and Ludwig, Jens and Miller, Douglas L},
  year={2011},
  institution={National Bureau of Economic Research}
}

@article{heckman2000substitution,
  title={Substitution and dropout bias in social experiments: A study of an influential social experiment},
  author={Heckman, James and Hohmann, Neil and Smith, Jeffrey and Khoo, Michael},
  journal={The Quarterly Journal of Economics},
  volume={115},
  number={2},
  pages={651--694},
  year={2000},
  publisher={MIT Press}
}

@article{puma2010head,
  title={Head Start Impact Study. Final Report.},
  author={Puma, Michael and Bell, Stephen and Cook, Ronna and Heid, Camilla and Shapiro, Gary and Broene, Pam and Jenkins, Frank and Fletcher, Philip and Quinn, Liz and Friedman, Janet and others},
  journal={Administration for Children \& Families},
  year={2010},
  publisher={ERIC}
}

@article{rambachan2024program,
  title={Program evaluation with remotely sensed outcomes},
  author={Rambachan, Ashesh and Singh, Rahul and Viviano, Davide},
  journal={arXiv preprint arXiv:2411.10959},
  year={2024}
}

@article{voronin2025generalized,
  title={Generalized Optimal Transport},
  author={Voronin, Andrei},
  journal={arXiv preprint arXiv:2507.22422},
  year={2025}
}

@techreport{luo2024selecting,
  title={Selecting Inequalities for Sharp Identification in Models with Set-Valued Predictions},
  author={Luo, Ye and Ponomarev, Kirill and Wang, Hai},
  year={2024},
  institution={Working Paper, University of Chicago}
}

@article{bugni2017inference,
  title={Inference for subvectors and other functions of partially identified parameters in moment inequality models},
  author={Bugni, Federico A and Canay, Ivan A and Shi, Xiaoxia},
  journal={Quantitative Economics},
  volume={8},
  number={1},
  pages={1--38},
  year={2017},
  publisher={Wiley Online Library}
}

@techreport{bitler2014experimental,
  title={Experimental evidence on distributional effects of Head Start},
  author={Bitler, Marianne P and Hoynes, Hilary W and Domina, Thurston},
  year={2014},
  institution={National Bureau of Economic Research}
}

\begin{appendices}

\section{Additional Discussions}\label{sect:extensions}
\titleformat{\subsection}[block]{\normalfont\large\bfseries}{A.\arabic{subsection}}{1em}{}

\subsection{Discretization of Short-term Outcomes}\label{sect:discretization}

In this section, I clarify the implications of discretizing short-term outcomes. To this end, let a researcher pose a surjective discretization function $\lambda: \mathcal{S}\rightarrow \mathcal{S}^D:= \{1,2,\hdots, k\}$ for some $k<\infty$, and define $S^D(d) = \lambda(S(d))$. Note that this subsumes the case in which $S(d)$ is finitely supported, since then $\lambda(s) = s$ for all $s\in\mathcal{S}$. I introduce $\lambda$ to clarify the subtle differences in applications of results of \Cref{sect:implementation} when $S(d)$ is finitely supported and discretized. Similarly define discretized temporal link functions $m^D_d:\mathcal{S}^D\rightarrow\mathcal{Y}$, given by $m^D_d = E[Y(d)|S^D(d)] = E[Y(d)|\lambda(S(d))]$, and let $m^D = (m_0^D,m_1^D)$. Pose the following analog of Assumption \ref{ass:modeling_generic} under the discretization.

\begin{myassump}{MA:D}\label{ass:selection_discrete}
 Suppose $\mathcal{M}^A$ and $\mathcal{M}^D$ are known or identified sets, and that $m\in\mathcal{M}^A\subseteq\mathcal{M}$.
    Then $\lambda$ is such that $m^D\in\mathcal{M}^D$.
\end{myassump}

 Assumptions \ref{ass:modeling_generic} and \ref{ass:selection_discrete} are closely related. The former maintains that the researcher imposes some modeling assumption that will restrict feasible $m$, as in \Cref{sect:modeling_assumptions}. The latter strengthens this notion and assumes that additionally $m^D$ satisfies known restrictions after discretization. Of course, if Assumption \ref{ass:modeling_generic} holds for a finitely supported $S(d)$, then Assumption \ref{ass:selection_discrete} trivially follows by taking $\lambda$ to be an identity function up to necessary relabeling of $S(d)$ values, if any. The remark below explains that for some modeling assumptions and discretization functions, \ref{ass:selection_discrete} follows immediately from \ref{ass:modeling_generic}, but that it may be restrictive for others.

\begin{remark}\label{rem:discretization}
Consider Assumption \ref{ass:LMIV} when $S(d)$ is a scalar stating that $E[Y(d)|S(d)=s]$ are nondecreasing functions. Then $E[Y(d)|S^D=s]$ must also be nondecreasing for any order-preserving $\lambda$, so Assumption \ref{ass:selection_discrete} holds for an appropriately chosen $\lambda$. However, LUC states that $m_d(s) = E_O[Y|S=s,D=d]$, which does not directly imply that $m_d^D(s)=E[Y|S^D=s,D=d]$. A similar remark can be made for treatment invariance.
\end{remark}

If $S(d)$ is finitely supported, \ref{ass:modeling_generic} and  \ref{ass:selection_discrete} are equivalent and \Cref{sect:implementation} characterizes the identified set. If $S(d)$ is discretized and Assumption \ref{ass:selection_discrete} holds as a direct consequence of Assumption \ref{ass:modeling_generic}, such as under LIV, then results characterize the identified set $\mathcal{H}(\tau)$ that is sharp \textit{under finitely-supported short-term outcomes}.\footnote{Note that this set may be larger than the intractable identified set that would have been obtained using non-discretized data.} This is also the case if the researcher believes the modeling assumption holds under discretized data, i.e., is willing to maintain \ref{ass:selection_discrete} directly. Otherwise, the results in \Cref{sect:implementation} should be viewed as providing an approximation of the identified set.

\subsection{Data Creation}\label{app:data_cleaning}

The construction of the two samples follows previous related work. The definitions of the long-term outcomes in CNLSY follow \citet{deming2009early}. Grade retention and diagnosis of a learning disability are defined as having reported being retained in any grade in school and being diagnosed with a learning disability, respectively. High school graduation is defined as reporting having graduated from high school or completing General Educational Development certification. Individuals are classified as idle if, in their most recent interview year, they report neither wages nor school attendance. Criminal involvement is defined as ever reporting having been convicted of a crime, placed on probation, sentenced by a judge, or incarcerated. Health status is measured by averaging responses to a Likert scale item on self-reported health status and generating an indicator equal to one if it is below three on a five-point scale. Finally, earnings are averaged across all reported values and converted to 2020 dollars using the Bureau of Labor Statistics Consumer Price Index.

Program participation in CNLSY is determined by the question of whether the child has ever attended Head Start, while HSIS contains indicators for true participation.\footnote{In this paper, I abstract away from substitution bias (\citet{heckman2000substitution}), as in the main analysis of \citet{garcia2020quantifying}. I consider the effects of Head Start participation compared to non-participation, irrespective of the take-up of alternatives.} Eligibility of non-participants in CNLSY must be inferred. For this, I rely on eligibility variables constructed by \citet[Section IIA]{carneiro2014long}. Eligibility is inferred by determining whether the child met the contemporaneous program requirements based on survey responses. Children ages three to five have been eligible if their family income is below the federal poverty line, or if their family has been eligible for any of the following public assistance programs: Aid to Families with Dependent Children (AFDC) or Temporary Assistance for Needy Families (TANF) after 1996, or Supplemental Security Income (SSI). Poverty status is verified by comparing the reported family income at ages three to five with the relevant federal poverty line, which is dependent on the family size and year. Eligibility for AFDC/TANF is determined based on two family income tests: the gross income test and the countable income test, as well as other pertinent categorical requirements. The income tests have state-specific thresholds that may vary by year and family size. Additionally, AFDC requires a specific family structure: either it must be female-headed, or with an unemployed main earner. The observational sample consists of individuals who were either determined to have been eligible for Head Start, or have participated in the program based on the relevant responses. 

HSIS contains Woodcock-Johnson III (WJ-III) cognitive assessment scores, which are used as short-term outcomes. For compatibility with the observational data, following \citet{griffen2017assessing}, I create composite scores for math and reading ability by averaging the national-level percentile scores on the corresponding components of the cognitive ability test and using them as a two-dimensional $S$ binned to a total of $400$ support points. As corresponding measures of math and reading ability in CNLSY, I take the percentile scores on the Peabody Individual Achievement Math and Reading Recognition subtests, also binned to a total of $400$ support points. 

\section{Proofs}\label{sect:proofs}
\titleformat{\subsection}[block]{\normalfont\large\bfseries}{B.\arabic{subsection}}{1em}{}

This section contains the proofs of the main results. It begins by summarizing notation. Supplemental Appendix \ref{sect:appendix_lemmas} contains auxiliary lemmas and propositions, with their proofs.

\

{\noindent\large \textbf{Preliminaries and Notation}}

\
Equality of distribution of two random elements or a random element and a law is denoted by $\eqd$ (e.g. $R\eqd P_R$ and $R\eqd R'$). I denote random sets with boldface letters (e.g. $\mathbf{R}$) and use $\mathbf{R},Z$ to denote the random set $\mathbf{R}\times \{Z\}$. $\mathbb{E}(\mathbf{R}|X)$ is used for the conditional Aumann expectation of a random set $\mathbf{R}$ given a sigma-algebra generated by a random vector $X$. The set of all selections of $\mathbf{R}$ is denoted by $Sel(\mathbf{R})$. The set of all random vectors $R\in Sel(\mathbf{R})$ such that $E[||R||]<\infty$ is denoted by $Sel^1(\mathbf{R})$. I use $\eqd$ to denote that a random element has a law, or an equivalent distribution-determining functional. $A$, $B$ and $K$ represent sets.  $\mathcal{C}(A)$ and $\mathcal{B}(A)$ are the families of all closed and Borel subsets of the set $A$, respectively. $co(A)$ is the closed convex hull of the set $A$. The identified set for a generic parameter $\theta$ is $\mathcal{H}(\theta)$. The set of distribution functions of random vectors with support $\mathcal{R}$ is $\mathcal{P}^\mathcal{R}$. I assume throughout that $\mathcal{Y}\times\mathcal{S}$ is a locally compact, second countable Hausdorff space, more precisely $\mathbb{R}^{1+d}$ endowed with its natural topology, while any of its subspaces inherit their relative topologies.

In the proofs for simpler notation, I will use the following random variable:
\begin{align}\label{eq:tildez}
    \Tilde{Z} = \mathbbm{1}[G=E] Z+\mathbbm{1}[G=O](sup\mathcal{Z}+1).
\end{align}

I use LIE to refer to the “law of iterated expectations".

\subsection{Main Results}\label{sect:main_proofs}
\printProofs

\begin{lemma}\label{lem:reverse_aumann}
    Let $\Bar{I}$ be the set of random elements $(E_1,E_2)$ such that $(E_1,E_2)\in \mathcal{S}\times\Tilde{\mathcal{Z}}$ and $E_1\independent E_2$. Suppose that $E_O[|Y(d)|]<\infty$ for any $d\in\{0,1\}$. For any $(m,\gamma)$ such that $\forall d\in\{0,1\}:$ \text{ $\exists (\varsigma_d,\Tilde{Z})\in Sel((\mathbf{S}_d,\Tilde{Z}))\cap \Bar{I}$ with $\gamma_d\eqd \varsigma_d$,} and $m_d$ that is $\gamma_d$ integrable satisfying:
    \begin{align}\label{eq:m_conditions}
         \text{ $\forall u\in\{-1,1\}:\enskip u m_d(s)\leq u\mu_d(s)\pi_{\gamma_d}(s)+ h_{co(\mathcal{Y})}(u)(1-\pi_{\gamma_d}(s)),\text{  $\gamma_d-$a.e.}$} 
    \end{align}
    there exist $\upsilon_d\in Sel^1(\mathbf{Y}_d)$ such that $m_d(\varsigma_d) = E_O[\upsilon_d|\varsigma_d]$ a.s. for $d\in\{0,1\}$.
\end{lemma}
\begin{proof}
Fix any $(m,\gamma)$ satisfying the conditions. For $d\in\{0,1\}$, let $(\varsigma_d,\Tilde{Z})\in Sel((\mathbf{S}_d,\Tilde{Z}))\cap \Bar{I}$ be such that
$\varsigma_d\eqd \gamma_d$. Note that $P_O(\varsigma_d \in\cdot ) =\gamma_d(\cdot)$ since $\varsigma_d\eqd \gamma_d$ and $(\varsigma_d,\Tilde{Z})\in\Bar{I}$. Recall that $\mu_d(s) = E_O[Y|S=s,D=d] = E_O[Y|\varsigma_d=s,D=d]$ $\gamma_d-$a.e., where the second equality follows by observing that $P_O(\varsigma_d = S|D=d)=1$ since $\varsigma_d \in Sel(\mathbf{S}_d)$ and $\mathbf{S}_d = \{S\}$ when $D=d$. $E_O[|Y\mathbbm{1}[D=d]|]\in\mathbb{R}$ since $E_O[|Y\mathbbm{1}[D=d]|]\leq E_O[|Y(d)|]<\infty$ where the final inequality is by assumption. Also, $\gamma_d-$a.e:
\begin{align}\label{eq:pi_mu_identity_rev}
\begin{split}
E_O\big(Y\mathbbm{1}[D=d]| \varsigma_d=s\big)
&=E_O(Y| \varsigma_d=s,D=d)P_O(D=d| \varsigma_d=s)=\mu_d(s)\pi_{\gamma_d}(s),
\end{split}
\end{align}
and thus $\mu_d(s)\pi_{\gamma_d}(s)\in\mathbb{R}$ $\gamma_d-$a.e. Fix an arbitrary $y_0\in\mathcal Y$. Define:
\begin{align}\label{eq:c_def_rev}
c_d(s):=
\begin{cases}
\dfrac{m_d(s)-\pi_{\gamma_d}(s)\mu_d(s)}{1-\pi_{\gamma_d}(s)}, & \text{if }\pi_{\gamma_d}(s)<1,\\[6pt]
y_0, & \text{if }\pi_{\gamma_d}(s)=1.
\end{cases}
\end{align}

Note that $c_d(s)\in\mathbb R$ $\gamma_d-$a.e., since $c_d(s) = y_0\in\mathbb{R}$ on $\{\pi_{\gamma_d}=1\}$, and on $\{\pi_{\gamma_d}<1\}$ $m_d$ is finite $\gamma_d-$a.e. and $\mu_d(s)\pi_{\gamma_d}(s)\in\mathbb R$ $\gamma_d-$a.e. By \eqref{eq:m_conditions}, $\mu_d(s) = m_d(s)$ whenever $\pi_{\gamma_d}(s)=1$, so then:
\begin{align}\label{eq:key_id_rev}
(1-\pi_{\gamma_d})c_d=m_d-\pi_{\gamma_d}\mu_d \qquad \gamma_d-\text{a.e.}
\end{align}

For any $s$ such that $\pi_{\gamma_d}(s)<1$ note that $\forall u \in\{-1,1\}$:
\begin{align}
    \begin{split}
u m_d(s)\leq u\mu_d(s)\pi_{\gamma_d}(s)+h_{co(\mathcal{Y})}(u)(1-\pi_{\gamma_d}(s)) \Rightarrow u c_d(s) \leq h_{co(\mathcal{Y})}(u).
    \end{split}
\end{align}
Hence, by \citet[Theorem 13.1]{rockafellar1970convex}, $c_d(s)\in co(\mathcal{Y})$ whenever $\pi_{\gamma_d}(s)<1$. Thus also $c_d(s)\in co(\mathcal{Y})$, since $c_d(s) = y_0\in\mathcal{Y}$ when $\pi_{\gamma_d}(s)=1$. Define for $t\in\mathbb{R}$:
\begin{align}\label{eq:y_pm_rev}
y_-(t):=\sup(\mathcal{Y}\cap(-\infty,t]),\qquad
y_+(t):=\inf(\mathcal{Y}\cap[t,\infty)).
\end{align}
For $t\in co(\mathcal{Y})\cap\mathbb{R}$, $y_-(t),y_+(t)\in\mathcal{Y}$ and $y_-(t)\leq t\leq y_+(t)$.
Moreover, $y_-$ and $y_+$ are monotone, hence Borel measurable. $\mathcal{Y}$ is closed, by definition. Fix values:
\begin{align}\label{eq:anchors_rev}
a:=
\begin{cases}
\max(\mathcal{Y}\cap(-\infty,0]), & \text{if }\mathcal{Y}\cap(-\infty,0]\neq\emptyset,\\
\min\mathcal{Y}, & \text{otherwise},
\end{cases}
\enskip
b:=
\begin{cases}
\min(\mathcal{Y}\cap[0,\infty)), & \text{if }\mathcal{Y}\cap[0,\infty)\neq\emptyset,\\
\max\mathcal{Y}, & \text{otherwise}.
\end{cases}
\end{align}
Let $C_\mathcal{Y}:=|a|+|b|$. Observe that $a,b\in\mathbb{R}\cap\mathcal{Y}$ and hence $C_\mathcal{Y}<\infty$. Define on the event $\{\pi_{\gamma_d}<1\}$:
\begin{align}\label{eq:y12_rev}
(y_1(s),y_2(s)):=
\begin{cases}
(y_-(c_d(s)),a), & \text{if }c_d(s)\leq a,\\
(a,b), & \text{if }a<c_d(s)<b,\\
(b,y_+(c_d(s))), & \text{if }c_d(s)\geq b,
\end{cases}
\end{align}

On the event $\{\pi_{\gamma_d}=1\}$, set $y_1(s)=y_2(s)=y_0$. Since $y_-(t),y_+(t)\in\mathcal{Y}$ for any $t\in co(\mathcal{Y})$, $c_d(s)\in co(\mathcal{Y})$, and $a,b\in\mathbb{R}\cap\mathcal{Y}$, then $P_O(y_i(\varsigma_d)\in\mathcal{Y}) = 1$ for $d\in\{0,1\}$ and $i\in\{1,2\}$. Finally, define:
\begin{align}\label{eq:p_rev}
p(s):=
\begin{cases}
\dfrac{y_2(s)-c_d(s)}{y_2(s)-y_1(s)}, & \text{if }y_2(s)>y_1(s),\\[6pt]
1, & \text{if }y_2(s)=y_1(s).
\end{cases}
\end{align}

Recall that for any $t\in co(\mathcal{Y})\cap\mathbb{R}$, $y_-(t)\leq t\leq y_+(t)$. Since $c_d(s)\in co(\mathcal Y)$ and $c_d(s)\in \mathbb{R}$, then it must be that $y_-(c_d(s))\leq c_d(s)\leq y_+(c_d(s))$ and therefore by \eqref{eq:y12_rev} $y_1(s)\leq c_d(s)\leq y_2(s)$ $\gamma_d-$a.e. on the event $\{\pi_{\gamma_d}<1\}$. On $\{\pi_{\gamma_d}=1\}$, $y_1(s)=y_2(s) = c_d(s) = y_0$. Then $p(s)\in[0,1]$ and $\gamma_d-$a.e:
\begin{align}\label{eq:mix_mean_rev}
p(s)y_1(s)+(1-p(s))y_2(s)=c_d(s).
\end{align}

After possibly enlarging the probability space, there exists a random variable $U:\Omega\to[0,1]$ which is uniformly distributed on $[0,1]$ with $U\independent(Y,D,\varsigma_d)|G=O$. Define:
\begin{align}\label{eq:tildeY_def_rev}
\Tilde{Y}_d:=\Big(y_1(\varsigma_d)\mathbbm{1}[U\leq p(\varsigma_d)]+y_2(\varsigma_d)\mathbbm{1}[U>p(\varsigma_d)]\Big)\mathbbm{1}[G=O]+y_0\mathbbm{1}[G=E].
\end{align}
Since $P_O(y_i(\varsigma_d)\in\mathcal{Y}) = 1$ for $d\in\{0,1\}$ and $i\in\{1,2\}$, it is immediate that $\Tilde{Y}_d\in\mathcal{Y}$ a.s. It is also direct that $\Tilde{Y}_d\independent D|\varsigma_d, G=O$ since $U\independent(Y,D,\varsigma_d)|G=O$. Then:
\begin{align}\label{eq:tildeY_condmean_rev}
E_O[\Tilde{Y}_d| \varsigma_d]= p(\varsigma_d)y_1(\varsigma_d)+(1-p(\varsigma_d))y_2(\varsigma_d)=c_d(\varsigma_d)\qquad \text{a.s.}
\end{align}
where the first equality is by $U\independent(Y,D,\varsigma_d)|G=O$ and the definition of $\Tilde{Y}_d$, and the second is by \eqref{eq:mix_mean_rev}. Finally, define:
\begin{align}\label{eq:upsilon_def_rev}
\upsilon_d:=Y\mathbbm{1}[D=d,G=O]+\Tilde{Y}_d\mathbbm{1}[D\neq d\vee G\neq O].
\end{align}
By construction, $\upsilon_d=Y$ on $\{D=d,G=O\}$ and $\upsilon_d\in\mathcal{Y}$ on $\{D\neq d\vee G\neq O\}$, hence $\upsilon_d\in Sel(\mathbf{Y}_d)$. 
Then $\gamma_d-$a.e.:
\begin{align}\label{eq:condmean_final_rev}
\begin{split}
E_O[\upsilon_d| \varsigma_d=s]
&=E_O\big[Y\mathbbm{1}[D=d]| \varsigma_d=s\big]+E_O\big[\Tilde{Y}_d\mathbbm{1}[D\neq d]| \varsigma_d=s\big]\\
&=\pi_{\gamma_d}(s)\mu_d(s)+(1-\pi_{\gamma_d}(s))c_d(s)
=m_d(s),
\end{split}
\end{align}
where the first equality is by observation, the second is by \eqref{eq:pi_mu_identity_rev}, $\Tilde{Y}_d\independent D|\varsigma_d, G=O$ and \eqref{eq:tildeY_condmean_rev}, and the last equality is by \eqref{eq:key_id_rev}. Therefore, $m_d(\varsigma_d)=E_O(\upsilon_d| \varsigma_d)$ a.s.

It remains to show that $E[|\upsilon_d|]<\infty$ so that $\upsilon_d\in Sel^1(\mathbf{Y}_d)$. To that end, observe that whenever $c_d(s)\leq a$ or $c_d(s)\geq b$, $y_1(s)$ and $y_2(s)$ have the same sign. If $\mathcal Y\cap(-\infty,0]=\emptyset$, then $\min(\mathcal Y)>0$ and by definition $a=\min\mathcal Y$ and $b=\min(\mathcal Y\cap[0,\infty))=\min\mathcal Y$. Hence $a=b>0$ and $\mathcal Y\subseteq [a,\infty)\subset(0,\infty)$.
Therefore, since $y_1(s),y_2(s)\in\mathcal Y$,  $y_1(s)\geq 0$ and $y_2(s)\geq 0$. If $\mathcal Y\cap[0,\infty)=\emptyset$, then $\max(\mathcal Y)<0$ and similarly $\mathcal Y\subset(-\infty,0)$, so $y_1(s)\leq 0$ and $y_2(s)\leq 0$. If both $\mathcal Y\cap(-\infty,0]\neq\emptyset$ and $\mathcal Y\cap[0,\infty)\neq\emptyset$, then by definition
$a=\max(\mathcal Y\cap(-\infty,0])\leq 0$ and $b=\min(\mathcal Y\cap[0,\infty))\geq 0$.
When $c_d(s)\leq a$, $(y_1(s),y_2(s))=(y_-(c_d(s)),a)$, and since $y_-(c_d(s))\leq c_d(s)\leq a\leq 0$ then
$y_1(s)\leq 0$ and $y_2(s)\leq 0$. When $c_d(s)\geq b$, $(y_1(s),y_2(s))=(b,y_+(c_d(s)))$, and since $y_+(c_d(s))\geq c_d(s)\geq b\geq 0$ then
$y_1(s)\geq 0$ and $y_2(s)\geq 0$.

Hence, for $s$ with $\pi_{\gamma_d}(s)<1$:
\begin{align}\label{eq:L1_bound_tilde_rev}
E_O[|\Tilde{Y}_d|| \varsigma_d=s]\leq |c_d(s)|+C_{\mathcal{Y}},
\end{align}
because if $c_d(s)\leq a$ or $c_d(s)\geq b$ then $y_1(s),y_2(s)$ have the same sign and
$E_O[|\Tilde{Y}_d|| \varsigma_d=s]=|E_O[\Tilde{Y}_d| \varsigma_d=s]|=|c_d(s)|$; if $a<c_d(s)<b$ then
$\Tilde{Y}_d\in\{a,b\}$ so $E_O[|\Tilde{Y}_d|| \varsigma_d=s]\leq |a|+|b|=C_{\mathcal{Y}}$. Then:
\begin{align}\label{eq:L1_tilde_part_rev}
\begin{split}
E_O\big[|\Tilde{Y}_d|\mathbbm{1}[D\neq d]\big]
&=E_O\Big[(1-\pi_{\gamma_d}(\varsigma_d))E_O[|\Tilde{Y}_d|| \varsigma_d]\Big]\\
&\leq \int (1-\pi_{\gamma_d})|c_d|d\gamma_d + C_{\mathcal{Y}}P_O(D\neq d)\\
&=\int |m_d-\pi_{\gamma_d}\mu_d|d\gamma_d+ C_{\mathcal{Y}}P_O(D\neq d)\\
&\leq \int |m_d|d\gamma_d+\int|\pi_{\gamma_d}\mu_d|d\gamma_d+ C_{\mathcal{Y}}P_O(D\neq d)\\
&= \int |m_d|d\gamma_d+\int\big|E_O(Y\mathbbm{1}[D=d]| \varsigma_d)\big|d\gamma_d+ C_{\mathcal{Y}}P_O(D\neq d)\\
&\leq \int |m_d|d\gamma_d+E_O\big(|Y|\mathbbm{1}[D=d]\big)+ C_{\mathcal{Y}}P_O(D\neq d)<\infty
\end{split}
\end{align}
where the first line is by LIE, the second uses $(1-\pi_{\gamma_d}(\varsigma_d))=0$ on $\{\pi_{\gamma_d}(\varsigma_d)=1\}$ and \eqref{eq:L1_bound_tilde_rev} on $\{\pi_{\gamma_d}(\varsigma_d)<1\}$, the third by \eqref{eq:key_id_rev} and $\pi_{\gamma_d}\in [0,1]$, the fourth by the triangle inequality, the fifth by \eqref{eq:pi_mu_identity_rev}, the sixth by Jensen's inequality and LIE, and the final since $m_d$ is $\gamma_d$ integrable, $E_O[|Y\mathbbm{1}[D=d]|]\leq E_O[|Y(d)|]<\infty$, and $C_{\mathcal{Y}}<\infty$. 
Then by \eqref{eq:upsilon_def_rev}, \eqref{eq:L1_tilde_part_rev} and $E_O[|Y\mathbbm{1}[D=d]|]<\infty$, $E_O[|\upsilon_d|]<\infty$. Moreover, $E_E[|\upsilon_d|] = |y_0|<\infty$ and hence $E[|\upsilon_d|]<\infty$ so $\upsilon_d\in Sel^1(\mathbf{Y}_d)$.
\end{proof}

\end{appendices}

\newpage

\setcounter{section}{0}
\renewcommand{\thesection}{S.\arabic{section}}
\renewcommand{\thesubsection}{S.\arabic{section}.\arabic{subsection}}

\renewcommand{\thelemma}{S.\arabic{lemma}}
\renewcommand{\theproposition}{S.\arabic{proposition}}
\renewcommand{\thetheorem}{S.\arabic{theorem}}

\begin{center}
\vspace*{-1.5cm}

{\LARGE
Identification of Long-Term Treatment Effects via\\
Temporal Links, Observational, and Experimental Data\par}

\vspace{0.6em}
{\large Supplemental Appendix\par}

\vspace{0.8em}
{\large Filip Obradovi\'c\textsuperscript{*}\par}
\end{center}

\begingroup
\makeatletter
\renewcommand{\@makefntext}[1]{%
  \noindent \makebox[1.8em][r]{\textsuperscript{*}}#1
}
\makeatother

\footnotetext[1]{UCLA, Department of Economics. Email: \href{mailto:obradovicfilip@ucla.edu}{obradovicfilip@ucla.edu}}
\endgroup

\vspace{-0.5in}
\renewcommand{\abstractname}{}
\begin{abstract}
This supplement: i) develops a consistent estimation procedure; ii) discusses computational simplifications; and iii) collects auxiliary technical results used in the main text.
\end{abstract}

\setcounter{page}{1}

\section{Estimation}\label{sect:estimation}

This section develops a consistent estimator building on \Cref{sect:identification}. Suppose that the researcher observes experimental and observational samples $\{(S_j,D_j,Z_j)\}_{j=1}^{n_E}$ and $\{(Y_i,S_i,D_i)\}_{i=1}^{n_O}$, respectively. Define $n:=\min\{n_O,n_E\}$. Let $\theta\in\Theta\subseteq\mathbb{R}^{d_\theta}$ be the subvector of $(m,\gamma)$ that may be partially identified under maintained assumptions, and $\mathcal{H}(\theta)$ its identified set. Let $\vartheta$ be a (possibly empty) subvector collecting components of $(m,\gamma)$ point identified by assumption.

The identified set $\mathcal{H}(\theta)$ is characterized by a collection of finitely many inequality and equality constraints, under appropriate structure on $\mathcal{M}^A$. Recall that $k=|\mathcal{S}|$ and let $\mu_d$ be a $k$-dimensional vector with components $\mu_d(s) = E_O[Y|S=s,D=d]$. Let $\eta_d$ be a $k\times(|\mathcal{Z}|+1)$ matrix with the element $(s,z)$ being $\eta_d(s,z)=P_E(S=s,D=d|Z=z)$ for $z\leq |\mathcal{Z}|$ and $\eta_d(s,z)=P_O(S=s,D=d)$ for $z=|\mathcal{Z}|+1$. 
Collect $\beta = (\mu_0,\mu_1,\eta_0,\eta_1,\vartheta,\Tilde{\beta})\in\mathfrak{B}$, where $\Tilde{\beta}\in\Tilde{\mathfrak{B}}$ is a (possibly empty) vector of other identified features of the population distribution defining $\mathcal{H} (\theta)$. For some known functions $\Tilde{h}$ and $\Tilde g$, one can then write:
\begin{align}
\begin{split}
         \mathcal{H}(\theta)=&\left\{
       \theta\in\Theta: \Tilde{h}(\theta,\beta)\geq 0, \enskip \Tilde g(\theta,\beta)=0
    \right\}.
\end{split}
\end{align}

\begin{example}
    Suppose that the researcher imposes Assumption \ref{ass:LMIV}, which allows partial identification of $m$, and that $\gamma$ may be partially identified due to imperfect compliance. Then $\theta=(m,\gamma)$, $\vartheta$ is empty and:
    \begin{align}\label{eq:ex_inequalities}
\begin{split}
         \mathcal{H}(\theta)=\mathcal{H}(m,\gamma)        =&\left\{\begin{array}{l}
        (m,\gamma)\in\mathcal{Y}^{2k}\times (\Delta(k))^2: \text{$\forall d\in\{0,1\}$, $ \forall s\in\mathcal{S}$,}\\
        m_d(s')- m_d(s)\geq0 \text{ for  $s'\geq s$ in the product order,}\\
         \text{$\gamma_d(s)- ess\sup_z\eta_d(s,z) \geq 0,$}\\
    \text{$\left(m_d(s)-\inf\mathcal{Y}\right)\gamma_d(s)-  \left(\mu_d(s)-\inf\mathcal{Y}\right)\eta_d(s,|\mathcal{Z}|+1)\geq 0$ },\\
    \text{$\left(\sup\mathcal{Y}-m_d(s)\right)\gamma_d(s)-  \left(\sup\mathcal{Y}-\mu_d(s)\right)\eta_d(s,|\mathcal{Z}|+1)\geq 0$ }
    \end{array}\right\}.
\end{split}
\end{align}
If one additionally assumes that datasets jointly point-identify $\gamma$, such as under perfect compliance, then $\theta=m$ and $\vartheta=(\gamma_0,\gamma_1)$ where $\gamma_d(s)= ess\sup_z\eta_d(s,z)$ for any $d\in\{0,1\}$ and $s\in\mathcal{S}$. Additionally, $\mathcal{H} (m,\gamma) = \mathcal{H}(\theta)\times\{\gamma\} $ with $\mathcal{H}(\theta)$ defined by restrictions as in \eqref{eq:ex_inequalities} with $\gamma_d(s)-ess\sup_z\eta_d(s,z)\geq 0$ omitted. If one imposes \ref{ass:TI} instead of \ref{ass:LMIV}, constraints $m_0(s)-m_1(s)=0$ replace $ m_d(s')- m_d(s)\geq0$ in the definitions of $\mathcal{H}(\theta)$.

\end{example}

\begin{example}
    Suppose that the researcher imposes Assumption LUC from \citet{athey2025combining}, which point identifies $m$, and uses the framework to account for imperfect compliance, so that $\gamma$ may be partially identified. Then $\theta=\gamma$, and $\vartheta=(m_0,m_1)$ with $m_d(s) = E_O[Y|S=s,D=d]$ for any $d\in\{0,1\}$ and $s\in\mathcal{S}$. Additionally:
    \begin{align}
\begin{split}
         \mathcal{H}(\theta)=\mathcal{H}(\gamma)        =&\left\{\begin{array}{l}
        \gamma\in(\Delta(k))^2: \text{$\forall d\in\{0,1\}$, $ \forall s\in\mathcal{S}$,}\\
         \text{$\gamma_d(s)- ess\sup_z\eta_d(s,z) \geq 0,$}
    \end{array}\right\}.
\end{split}
\end{align}
\end{example}

\noindent The identified set $\mathcal{H}(\theta)$ can be equivalently represented via a criterion function. Define:
\begin{align}
    Q(\theta,\beta)
    := \|\tilde{g}(\theta,\beta)\|_1
     + \sum_{t=1}^{q}\max\left(-\tilde{h}_t(\theta,\beta), 0\right),
\end{align}
where $q$ denotes the number of inequality constraints in $\tilde{h}(\theta,\beta)$. If the assumptions hold, $\mathcal{H}(\theta)\neq\emptyset$, so $\bar{Q}:=\min_{\theta\in\Theta}Q(\theta,\beta)=0$ and $    \mathcal{H}(\theta)
    = \argmin_{\theta\in\Theta} Q(\theta,\beta)
    = \left\{\theta\in\Theta: Q(\theta,\beta) = 0\right\}.$ Writing $T(\theta,\vartheta)$ for $T(m,\gamma)$ evaluated at the $(m,\gamma)$ implied by $(\theta,\vartheta)$:
\begin{align}
\begin{split}
    \min_{(m,\gamma)\in\mathcal{H}(m,\gamma)} T(m,\gamma) = \min_{\theta\in\mathcal{H}(\theta)}  T(\theta,\vartheta)= \min_{\theta\in\Theta}  T(\theta,\vartheta)
    &\quad\text{s.t.}\quad Q(\theta,\beta) \leq \bar{Q},\\
    \max_{(m,\gamma)\in\mathcal{H}(m,\gamma)} T(m,\gamma)= \max_{\theta\in\mathcal{H}(\theta)}  T(\theta,\vartheta)
    = \max_{\theta\in\Theta}  T(\theta,\vartheta)
    &\quad\text{s.t.}\quad Q(\theta,\beta) \leq \bar{Q}.
\end{split}
\end{align}

Let $\beta_n$ be a consistent estimator of $\beta$ obtained from the empirical analogs, and let $\vartheta_n$ denote the subvector of $\beta_n$ corresponding to $\vartheta$. Denote by $\bar{Q}_n:=\min_{\theta\in\Theta}Q(\theta,\beta_n)$. The
criterion estimator of $\mathcal{H}(\tau)$ is:
\begin{align}
\begin{split}
    &\mathcal{H}_n(\tau) := \left[\tau_n^{LB}, \tau_n^{UB}\right],\\
    &\tau_n^{LB}
    = \min_{\theta\in\Theta}  T(\theta,\vartheta_n)
    \quad\text{s.t.}\quad Q(\theta,\beta_n)\leq\bar{Q}_n,\\
    &\tau_n^{UB}
    = \max_{\theta\in\Theta}  T(\theta,\vartheta_n)
    \quad\text{s.t.}\quad Q(\theta,\beta_n)\leq\bar{Q}_n.
\end{split}
\end{align}

To establish Hausdorff consistency without requiring a tuning parameter, I introduce
additional notation. Let
$\mathcal{H}^{ie}(\theta):=\{\theta\in\Theta:
\tilde{h}(\theta,\beta)\geq 0\}$ denote the inequality-constrained parameter space. Modeling assumptions involving only equalities, such as Assumptions \ref{ass:TI} and LUC
do not affect $\mathcal{H}^{ie}(\theta)$ since they are collected by $\tilde{g}(\theta,\beta)=0$. Maintain the following assumption.

\begin{myassump}{E}{(Estimation)}\label{ass:estimation}
\begin{enumerate}[label=\roman*)]
    \item $\{(S_j,D_j,Z_j)\}_{j=1}^{n_E}$ and
    $\{(Y_i,S_i,D_i)\}_{i=1}^{n_O}$ are i.i.d.\ samples;
    \item $\mathcal{Y}$ is bounded and
    $|\mathcal{S}|,|\mathcal{Z}|<\infty$;
    \item $\mathcal{M}^A$ is defined through finitely many linear
    equality and weak inequality constraints which may depend on a
    consistently estimable vector of population parameters
    $\beta\in\mathfrak{B}$ where $\mathfrak{B}$ is compact. The Jacobian of all linear equality constraints, including any simplex equalities adjoined by reformulation, has full row rank.
    \item $\operatorname{cl}\left(
    \operatorname{int}(\mathcal{H}^{ie}(\theta))\cap
    \mathcal{H}(\theta)\right)=\mathcal{H}(\theta)$ or
    $\mathcal{H}(\theta)$ is a singleton.
\end{enumerate}
\end{myassump}

Assumption \ref{ass:estimation} \textit{i)} is standard under random sampling. Condition \textit{ii)} maintains that long-term outcomes are bounded and that short-term potential outcomes and the instrument $Z$ are finitely supported. Assumption \ref{ass:estimation} \textit{iii)} defines the class of modeling assumptions compatible with the estimation
procedure. It is sufficiently general to encompass all previously stated modeling assumptions, but may be further weakened to allow for continuously differentiable restrictions on $m$ if necessary. Condition \textit{iv)} is a mild condition introduced by \citet{shi2015simple} that enables consistent estimation without requiring a tuning parameter. For example, it holds when $\operatorname{int}(\mathcal{H}(\theta))\neq\emptyset$, i.e. when components of $\theta$ are partially identified, or when $\mathcal{H}(\theta)$ is in the interior of $\mathcal{H}^{ie}(\theta)$. It may be further relaxed, as explained by \citet[Section 2]{shi2015simple}. Applying the arguments of \citet[Theorem 2.1]{shi2015simple} yields the following.

\begin{theorem}\label{th:consistency}
Let Assumptions \ref{ass:rand_assign}, \ref{ass:ex_validity}, \ref{ass:modeling_generic}, and \ref{ass:estimation} hold. Then as $n\conv\infty$:
\begin{align*}
    d_H(\mathcal{H}_n(\tau),\mathcal{H}(\tau)):=\max\left\{\sup_{\tau_0\in \mathcal{H}(\tau)}\inf_{\hat{\tau}\in \mathcal{H}_n(\tau)}||\tau_0-\hat{\tau}||,\sup_{\hat{\tau}\in\mathcal{H}_n(\tau)}\inf_{\tau_0\in \mathcal{H}(\tau)}||\tau_0-\hat{\tau}||\right\}\convp 0.
\end{align*}
\end{theorem}

\begin{remark}\label{rem:inference}
Since $\tau$ is a subvector of $(m,\gamma)$, inference on $\tau$ may be recast as a subvector inference problem based on an appropriate criterion function. This would enable the use of general methods for subvector inference, such as those in \citet{bugni2017inference}, to construct confidence sets for $\tau$ and test hypotheses about it.
\end{remark}

\subsection{Reducing Computational Complexity of Bilinear Programming}\label{sect:computational}

Exploiting the structure of the identified set $\mathcal{H}(m,\gamma)$ and the objective $T$ may further reduce the computational burden of non-convex optimization.

First, the optimization problems may be separable into lower-dimensional subproblems. Solving the individual subproblems may be less resource-intensive than solving the original problem jointly \citep{nocedal1999numerical}. This is feasible when the modeling assumption yields a rectangular set $\mathcal{M}^A = \mathcal{M}^A_0\times\mathcal{M}^A_1$ for some $\mathcal{M}^A_0$ and $\mathcal{M}^A_1$. Since the remaining constraints on $(m,\gamma)$ are also separable in $d$, it follows immediately that the identified set $\mathcal{H}(m,\gamma)$ is rectangular as well. Letting $\mathcal{T}(m_d,\gamma_d) := \int_\mathcal{S}m_d(s)d\gamma_d(s)$, we have:
\begin{align}\label{eq:separable_problems}
\begin{split}
    \min_{(\Tilde{m},\Tilde{\gamma})\in\mathcal{H}(m,\gamma)} T(\Tilde{m},\Tilde{\gamma}) = \min_{(\Tilde{m}_1,\Tilde{\gamma}_1)\in\mathcal{H}(m_1,\gamma_1)} \mathcal{T}(\Tilde{m}_1,\Tilde{\gamma}_1)-\max_{(\Tilde{m}_0,\Tilde{\gamma}_0)\in\mathcal{H}(m_0,\gamma_0)} \mathcal{T}(\Tilde{m}_0,\Tilde{\gamma}_0)\\
    \max_{(\Tilde{m},\Tilde{\gamma})\in\mathcal{H}(m,\gamma)} T(\Tilde{m},\Tilde{\gamma}) = \max_{(\Tilde{m}_1,\Tilde{\gamma}_1)\in\mathcal{H}(m_1,\gamma_1)} \mathcal{T}(\Tilde{m}_1,\Tilde{\gamma}_1)-\min_{(\Tilde{m}_0,\Tilde{\gamma}_0)\in\mathcal{H}(m_0,\gamma_0)} \mathcal{T}(\Tilde{m}_0,\Tilde{\gamma}_0)
\end{split}
\end{align}
where $\mathcal{H}(m_d,\gamma_d)$ collects all constraints on $(m_d,\gamma_d)$ with $m_d\in\mathcal{M}^A_d$. For example, $\mathcal{M}^A$ is rectangular whenever the modeling assumption does not relate values of $m_1$ and $m_0$, such as with Assumptions \ref{ass:LMIV} and LUC.

Second, the optimization problems become linear in certain cases. They are therefore convex, which substantially simplifies computation. As suggested by examples in the previous section, this occurs in two cases. First, if $\mathcal{H}(m,\gamma)=\{m\}\times\mathcal{H}(\gamma)$, the problems reduce to linear programs over $\gamma$. For example, assumptions that point-identify $m$ independently of $\gamma$, such as LUC, yield $\mathcal{H}(m,\gamma)$ of this form. Second, if $\mathcal{H}(m,\gamma)=\mathcal{H}(m)\times\{\gamma\}$ and $\mathcal{M}^A$ is representable by linear constraints, the problems reduce to linear programs over $m$. This arises under Assumptions \ref{ass:LMIV} and \ref{ass:TI} when the right-hand sides of the constraints for each $\gamma_d$ sum to one. \Cref{sect:app} exploits both separability and linearization to simplify computation in the application.

Third, for each feasible $\gamma$, the inner optimization over $m$ may admit a closed-form solution. In such cases, fixing $m$ appropriately can reduce the dimension of the parameter space explored by branch-and-bound algorithms. To formalize this idea, rewrite the optimization problems as bilevel programs. Decompose $\mathcal{H}(m,\gamma)$ into its projection $\mathcal{H}(\gamma) := \{\gamma': \exists m'\text{ s.t. } (m',\gamma')\in\mathcal{H}(m,\gamma)\}$ and corresponding fibers $\mathcal{H}(m|\gamma') := \{m': (m',\gamma')\in \mathcal{H}(m,\gamma)\}$ at each $\gamma'\in\mathcal{H}(\gamma)$. These fibers form a correspondence $\mathcal{H}(m|\cdot):\mathcal{H}(\gamma)\rightrightarrows \mathcal{M}^A$. The identified set can then be written as:
\begin{align}
        \mathcal{H}(\tau) = \left[\min_{\Tilde{\gamma}\in\mathcal{H}(\gamma)}\min_{\Tilde{m}\in\mathcal{H}(m|\Tilde{\gamma})}T(\Tilde{m},\Tilde{\gamma}),\max_{\Tilde{\gamma}\in\mathcal{H}(\gamma)}\max_{\Tilde{m}\in\mathcal{H}(m|\Tilde{\gamma})}T(\Tilde{m},\Tilde{\gamma})\right].
\end{align}

The inner optimization problems may have closed-form solutions given by selectors of the correspondence $\mathcal{H}(m|\cdot)$. This is formalized by the following definition.

\begin{definition}[Minimal and maximal selectors]
Let $\mathcal{H}(m|\cdot):\mathcal{H}(\gamma)\rightrightarrows \mathcal{M}^A$ be the correspondence defined by the fibers of $\mathcal{H}(m,\gamma)$ over its projection $\mathcal{H}(\gamma)$. A selector $L:\mathcal{H}(\gamma)\to\mathcal{M}^A$ is \textit{minimal with respect to $T$} if, for every $\gamma\in\mathcal{H}(\gamma)$, $T(L(\gamma),\gamma)\leq T(m,\gamma)$ for all $m\in\mathcal{H}(m|\gamma)$. A selector $U:\mathcal{H}(\gamma)\to\mathcal{M}^A$ is \textit{maximal with respect to $T$} if, for every $\gamma\in\mathcal{H}(\gamma)$, $T(U(\gamma),\gamma)\geq T(m,\gamma)$ for all $m\in\mathcal{H}(m|\gamma)$.
\end{definition}

If minimal and maximal selectors exist, using them will yield the identified set $\mathcal{H}(\tau)$ since:
    \begin{align*}
       \left[\min_{(\Tilde{m},\Tilde{\gamma})\in\mathcal{H}(m,\gamma)}T(\Tilde{m},\Tilde{\gamma}),\max_{(\Tilde{m},\Tilde{\gamma})\in\mathcal{H}(m,\gamma)}T(\Tilde{m},\Tilde{\gamma})\right] =  \left[\min_{\Tilde{\gamma}\in\mathcal{H}(\gamma)}T(L_{\Tilde{\gamma}},\Tilde{\gamma}),\max_{\Tilde{\gamma}\in\mathcal{H}(\gamma)}T(U_{\Tilde{\gamma}},\Tilde{\gamma})\right].
    \end{align*}

To operationalize the result, consider the lower bound $\min_{(\Tilde{m},\Tilde{\gamma})\in\mathcal{H}(m,\gamma)} T(\Tilde{m},\Tilde{\gamma}) $. In the minimization problem, one can replace the constraints $m\in\mathcal{M}^A$ in \Cref{cor:interval_set} with $m(s)=L_\gamma(s)$ for each $s\in\mathcal{S}$, thereby obtaining the lower bound without separately optimizing over $m$ for each $\gamma$. Similarly, in the maximization problem, one can replace $m\in\mathcal{M}^A$ with $m(s)=U_\gamma(s)$ to obtain the upper bound. Unlike the previous two cases, this simplification may be useful for estimation only when $\Bar{Q}_n=0$, or equivalently, when the plug-in estimator is non-empty.\footnote{The plug-in and criterion estimators are numerically equivalent when $\Bar{Q}_n=0$. The plug-in is also Hausdorff consistent for $\mathcal{H}(\theta)$ under Assumption \ref{ass:estimation} since $\Bar{Q}_n=0$ w.p.a. 1.} 
\Cref{lem:min_max_selectors} derives minimal and maximal selectors under Assumptions \ref{ass:LMIV} and \ref{ass:TI}.

\section{Proofs and Auxiliary Lemmas}\label{sect:appendix_lemmas}

\begin{proof}[Proof of \Cref{th:consistency}]

The proof proceeds in three steps. Step 1 reduces the problem to Hausdorff consistency of $\mathcal{H}_n(\theta)$. Step 2 reformulates the criterion using slack variables and establishes Hausdorff consistency of the resulting argmin set. Step 3 shows that this is sufficient via Step 1.

\medskip

\underline{\textbf{Step 1.} \textit{Sufficiency of $d_H(\mathcal{H}_n(\theta),\mathcal{H}(\theta))\convp 0$.}}

By Assumption \ref{ass:estimation} \textit{ii)}, $|\mathcal{S}|<\infty$. Thus, $(m,\gamma)$ is a finite-dimensional vector and $T$ is jointly continuous in $(m,\gamma)$. Since $\mathcal{H}(m,\gamma)$ is defined by finitely many weak inequalities and equalities, it is closed. Moreover, $\mathcal{H}(m,\gamma)\subseteq \mathcal{Y}^{2k}\times(\Delta(k))^2$, and the latter set is bounded because $\mathcal{Y}$ is bounded by Assumption \ref{ass:estimation} \textit{ii)}. Therefore, $\mathcal{H}(m,\gamma)$ is compact. Hence $\mathcal{H}(\tau)
=\left[
\min_{(\Tilde m,\Tilde \gamma)\in\mathcal{H}(m,\gamma)}T(\Tilde m,\Tilde \gamma),\max_{(\Tilde m,\Tilde \gamma)\in\mathcal{H}(m,\gamma)}T(\Tilde m,\Tilde \gamma)\right]$ 
is a closed interval. The same holds for $\mathcal{H}_n(\tau)$ by construction. Since the Hausdorff distance between closed intervals $[a,b]$ and $[c,d]$ equals $\max\{|a-c|,|b-d|\}$, it suffices to show that the endpoints converge in probability. I treat the upper bound; the lower bound is symmetric.

By boundedness of $\mathcal{Y}$, each component of $\mu_d$ lies in a compact interval and each component of $\eta_d$ lies in $[0,1]$, by definition. By the definition of $\Theta$ and boundedness of $\mathcal{Y}$, $\Theta$ is a product of finitely many compact sets, and hence compact. By Assumption \ref{ass:estimation} \textit{iii)}, $\mathfrak{B}$ is compact. Therefore $\Theta\times\mathfrak{B}$, and hence $\Theta\times\operatorname{proj}_{\vartheta}(\mathfrak{B})$, are compact. Since $T(\theta,\vartheta)$ is jointly continuous on $\Theta\times\operatorname{proj}_{\vartheta}(\mathfrak{B})$, the Heine-Cantor theorem implies that $T$ is uniformly continuous on this domain.

Fix any $\varepsilon>0$. $\tau_n^{UB}
= \max_{\theta\in\mathcal{H}_n(\theta)}T(\theta,\vartheta_n)$ and
$\tau^{UB}=\max_{\theta\in\mathcal{H}(\theta)}T(\theta,\vartheta)$. By the triangle inequality:
\begin{align}\label{eq:difference_tau}
|\tau_n^{UB}-\tau^{UB}|
\leq
\underbrace{\max_{\theta\in\Theta}|T(\theta,\vartheta_n)-T(\theta,\vartheta)|}_{=: \rho_n}
+
\underbrace{\left|
\max_{\theta\in\mathcal{H}_n(\theta)}T(\theta,\vartheta)
-
\max_{\theta\in\mathcal{H}(\theta)}T(\theta,\vartheta)
\right|}_{=: \Delta_n}.
\end{align}

Since $T$ is uniformly continuous on $\Theta\times\operatorname{proj}_{\vartheta}(\mathfrak{B})$, for any $\varepsilon>0$ there exists $\delta>0$ such that for any $\vartheta'\in\operatorname{proj}_{\vartheta}(\mathfrak{B})$: $\|\vartheta'-\vartheta\|<\delta
\Rightarrow
\max_{\theta\in\Theta}|T(\theta,\vartheta')-T(\theta,\vartheta)|<\varepsilon.$
On the event $\{\|\vartheta_n-\vartheta\|<\delta\}$ this implies $\rho_n<\varepsilon$. Therefore:
$
P(\rho_n>\varepsilon)\leq P(\|\vartheta_n-\vartheta\|\geq\delta)\to 0,$ 
since $\vartheta_n\convp\vartheta$. Hence $\rho_n\convp 0$.

For $\Delta_n$, fix any $\varepsilon'>0$ and choose $\delta>0$ such that $\|\theta-\theta'\|<\delta
\Rightarrow
|T(\theta,\vartheta)-T(\theta',\vartheta)|<\varepsilon'$ for all $\theta,\theta'\in\Theta$. Since $\mathcal{H}_n(\theta)$ and $\mathcal{H}(\theta)$ are compact and $T(\cdot,\vartheta)$ is continuous, the maxima are attained. If $d_H(\mathcal{H}_n(\theta),\mathcal{H}(\theta))<\delta$, then for any $\theta\in\mathcal{H}_n(\theta)$ there exists $\theta'\in\mathcal{H}(\theta)$ with $\|\theta-\theta'\|<\delta$, hence $T(\theta,\vartheta)
\leq
T(\theta',\vartheta)+\varepsilon'
\leq
\max_{\theta^\ast\in\mathcal{H}(\theta)}T(\theta^\ast,\vartheta)+\varepsilon'.$ Taking the maximum over $\theta\in\mathcal{H}_n(\theta)$:
\begin{align*}
\max_{\theta\in\mathcal{H}_n(\theta)}T(\theta,\vartheta)
\leq
\max_{\theta\in\mathcal{H}(\theta)}T(\theta,\vartheta)+\varepsilon'.
\end{align*}
The reverse inequality follows by the same argument interchanging $\mathcal{H}_n(\theta)$ and $\mathcal{H}(\theta)$. Therefore:
\begin{align*}
\Delta_n =\left|
\max_{\theta\in\mathcal{H}_n(\theta)}T(\theta,\vartheta)
- \max_{\theta\in\mathcal{H}(\theta)}T(\theta,\vartheta)
\right|\leq \varepsilon'
\end{align*}
whenever $d_H(\mathcal{H}_n(\theta),\mathcal{H}(\theta))<\delta$. Taking $\varepsilon'=\varepsilon/2$ implies:
\begin{align}\label{eq:delta_hausdorff}
P(\Delta_n>\tfrac{\varepsilon}{2})
\leq
P(d_H(\mathcal{H}_n(\theta),\mathcal{H}(\theta))\geq\delta).
\end{align}

Combining \eqref{eq:difference_tau} and \eqref{eq:delta_hausdorff}:
\begin{align*}
P(|\tau_n^{UB}-\tau^{UB}|>\varepsilon)
&\leq P(\rho_n>\tfrac{\varepsilon}{2})
+ P(d_H(\mathcal{H}_n(\theta),\mathcal{H}(\theta))\geq\delta).
\end{align*}
Thus, for $d_H(\mathcal{H}_n(\tau),\mathcal{H}(\tau))\convp 0$, it is sufficient to show $d_H(\mathcal{H}_n(\theta),\mathcal{H}(\theta))\convp 0$.

\medskip

\underline{\textbf{Step 2.} \textit{Reformulation and Hausdorff consistency.}}

To apply Lemma A.1 of \citet{shi2015simple}, I reformulate the problem using slack variables. Convert all weak inequality constraints $\Tilde{h}(\theta,\beta)\geq 0$ to equalities by introducing $\lambda\in\mathbb{R}^q_+$ so that $\Tilde{h}(\theta,\beta)-\lambda=0$. Denote by $\Tilde{\theta}:=(\theta,\lambda)\in\mathfrak{T}$ the augmented parameter vector and define $g(\Tilde{\theta},\beta) :=\begin{pmatrix}
\Tilde{h}(\theta,\beta)-\lambda \\
\Tilde{g}(\theta,\beta)
\end{pmatrix}.$ When $\Theta$ is defined using simplices, as when $\theta$ includes $\gamma$, replace each $\Delta(k)$ with $[0,1]^k$ in $\Theta$ and adjoin the simplex equalities $\sum_s\gamma_d(s)-1=0$ for $d\in\{0,1\}$ to $\Tilde g$ as additional rows. This does not alter the feasible set in $\theta$-space, since $\gamma\in(\Delta(k))^2$ if and only if $\gamma\in[0,1]^{2k}$ and the simplex equalities hold. When $\Theta$ is already a product of compact intervals, no reformulation is needed.

Next, absorb the constraints defining $\Theta$ into $\Tilde h$ by appending $\theta_i-a_i\geq 0$ and $b_i-\theta_i\geq 0$ for each coordinate $i$ of $\theta$, where $[a_i,b_i]$ denotes the corresponding factor of $\Theta$. The slack variables for these additional constraints are handled identically. Then enlarge each factor of $\Theta$ to $[a_i-1,b_i+1]$. Since the original box constraints are now enforced through $\Tilde{h}$, this reformulation does not alter the feasible set in $\theta$-space and ensures $\mathcal{H}(\theta)\subset\operatorname{int}(\Theta)$.

By Assumption \ref{ass:estimation} \textit{iii)}, any restrictions implied by $\mathcal{M}^A$ can be written as finitely many linear equalities and weak inequalities in $\theta$ whose coefficients may depend on $\beta$. Together with $|\mathcal{S}|<\infty$ and $|\mathcal{Z}|<\infty$ from Assumption \ref{ass:estimation} \textit{ii)}, this implies that each component of $\Tilde h(\theta,\beta)$ is obtained from finitely many linear or bilinear expressions in $\theta$ and finitely many coordinates of $\beta$. Therefore $(\theta,\beta)\mapsto \Tilde h(\theta,\beta)$ is continuous on $\Theta\times\mathfrak{B}$. Since $\Theta\times\mathfrak{B}$ is compact, each component of $\Tilde h$ is bounded above on $\Theta\times\mathfrak{B}$.  Henceforth, let $q$ denote the total number of inequality 
constraints, including the appended box constraints. Define $\bar h\in\mathbb{R}^q$ with components $\bar h_t:=\sup_{(\theta,\beta)\in\Theta\times\mathfrak{B}}\Tilde h_t(\theta,\beta)<\infty$ and set $\Lambda:=\prod_{t=1}^q[0,\bar h_t+1].$ Then any feasible $(\theta,\lambda)$ satisfying $\tilde{h}(\theta,\beta)-\lambda=0$ for some $\beta\in\mathfrak{B}$ necessarily has 
$\lambda_t=\tilde{h}_t(\theta,\beta)\geq 0$ for each $t$, since feasibility 
requires $\tilde{h}_t(\theta,\beta)\geq 0$. Moreover, $\lambda_t=\tilde{h}_t(\theta,\beta)\leq \bar{h}_t<\bar{h}_t+1$ for each $t$. 
Therefore $\lambda\in\Lambda$, and imposing $\lambda\in\Lambda$ is 
without loss of generality.

The augmented parameter space $\mathfrak{T}:=\Theta\times\Lambda$ is a product of compact intervals, and thus compact, with $\operatorname{cl}(\operatorname{int}(\mathfrak{T}))=\mathfrak{T}$. For each $\theta\in\Theta$ and $\beta'\in\mathfrak{B}$, minimizing over $\lambda\in\Lambda$ yields:
\begin{align*}
\min_{\lambda\in\Lambda}\|g((\theta,\lambda),\beta')\|_1
=
\|\Tilde g(\theta,\beta')\|_1+\sum_{t=1}^q\max(-\Tilde h_t(\theta,\beta'),0)
=
Q(\theta,\beta'),
\end{align*}
because $\min_{\lambda_t\in[0,\bar h_t+1]}|\Tilde h_t(\theta,\beta')-\lambda_t|
=
\max(-\Tilde h_t(\theta,\beta'),0)$ for each $t$. Let $\Bar Q_n:=
\min_{\Tilde{\theta}\in\mathfrak{T}}\|g(\Tilde{\theta},\beta_n)\|_1$. The minimum $\Bar Q_n$ exists by continuity of $\|g(\cdot,\beta_n)\|_1$ on the compact set $\mathfrak{T}$, so $\Tilde{\Theta}_n$ is nonempty and compact. Define:
\begin{align*}
\Tilde{\Theta}
:=\{\Tilde{\theta}\in\mathfrak{T}: g(\Tilde{\theta},\beta)=0\}, \quad
\Tilde{\Theta}_n:=\argmin_{\Tilde{\theta}\in\mathfrak{T}}\|g(\Tilde{\theta},\beta_n)\|_1=\{\Tilde{\theta}\in\mathfrak{T}: \|g(\Tilde{\theta},\beta_n)\|_1=\Bar Q_n\}.
\end{align*}

It follows that $\mathcal{H}_n(\theta)$ and $\mathcal{H}(\theta)$ are the projections of the augmented sets onto $\theta$-coordinates $\mathcal{H}(\theta)
=\left\{\theta\in\Theta:
\exists\lambda\in\Lambda \text{ s.t. } (\theta,\lambda)\in\Tilde{\Theta} \right\}$, and $\mathcal{H}_n(\theta)
=\left\{\theta\in\Theta:\exists\lambda\in\Lambda \text{ s.t. } (\theta,\lambda)\in\Tilde{\Theta}_n\right\}$. I then establish $d_H(\Tilde{\Theta}_n,\Tilde{\Theta})\convp 0$ in two parts.

\medskip

\underline{\textbf{Step 2a.} \textit{$\sup_{\Tilde{\theta}\in\Tilde{\Theta}_n}\inf_{\Tilde{\theta}_0\in\Tilde{\Theta}}\|\Tilde{\theta}-\Tilde{\theta}_0\|\convp 0$.}}

By the reverse triangle inequality:
\begin{align}\label{eq:unif_conv_v4_compact}
\sup_{\Tilde{\theta}\in\mathfrak{T}}
\left|
\|g(\Tilde{\theta},\beta_n)\|_1-\|g(\Tilde{\theta},\beta)\|_1
\right|
\leq
\sup_{\Tilde{\theta}\in\mathfrak{T}}
\|g(\Tilde{\theta},\beta_n)-g(\Tilde{\theta},\beta)\|_1
=: r_n.
\end{align}
The function $g$ is polynomial in $(\Tilde{\theta},\beta)$ and hence jointly continuous. Since it is jointly continuous on the compact set $\mathfrak{T}\times\mathfrak{B}$, it is uniformly continuous. By Assumption \ref{ass:estimation} \textit{i)} to \textit{iii)}, $\beta_n\convp\beta$. Hence $r_n\convp 0$. Fix any $\Tilde{\theta}_0\in\Tilde{\Theta}$. Then:
\begin{align*}
0\leq\Bar Q_n\leq\|g(\Tilde{\theta}_0,\beta_n)\|_1
\leq\|g(\Tilde{\theta}_0,\beta)\|_1+r_n=r_n\convp 0,
\end{align*}
where the first inequality is by non-negativity of $Q$, the second is by definition of $\Bar Q_n$, and the last equality uses \eqref{eq:unif_conv_v4_compact} and $g(\Tilde{\theta}_0,\beta)=0$ since $\Tilde{\theta}_0\in\Tilde{\Theta}$. Thus $\Bar Q_n\convp 0$.

Next, fix any $\varepsilon>0$ and define:
\begin{align*}
A_\varepsilon
:=\left\{\Tilde{\theta}\in\mathfrak{T}:
\inf_{\Tilde{\theta}_0\in\Tilde{\Theta}}\|\Tilde{\theta}-\Tilde{\theta}_0\|\geq \varepsilon \right\}.
\end{align*}
Note that $\delta_\varepsilon:=\inf_{\Tilde{\theta}\in A_\varepsilon} \|g(\Tilde{\theta},\beta)\|_1>0$. Since $A_\varepsilon$ is by definition a closed
subset of the compact set $\mathfrak{T}$, it is compact, and the
continuous function $\|g(\cdot,\beta)\|_1$ attains its infimum on
$A_\varepsilon$. If $\delta_\varepsilon=0$, the minimizer
$\Bar{\Tilde{\theta}}\in A_\varepsilon$ satisfies
$g(\Bar{\Tilde{\theta}},\beta)=0$, so
$\Bar{\Tilde{\theta}}\in\Tilde{\Theta}$, contradicting
$\Bar{\Tilde{\theta}}\in A_\varepsilon$. Hence $\delta_\varepsilon>0$.

Now consider the event $\{\Tilde{\Theta}_n\cap A_\varepsilon\neq\emptyset\}$. On this event, pick any $\Tilde{\theta}\in\Tilde{\Theta}_n\cap A_\varepsilon$. By definition of $\delta_\varepsilon$, $\|g(\Tilde{\theta},\beta)\|_1\geq \delta_\varepsilon$. Since $\Tilde{\theta}\in\Tilde{\Theta}_n$, $\|g(\Tilde{\theta},\beta_n)\|_1=\Bar Q_n.$ By \eqref{eq:unif_conv_v4_compact}, $\|g(\Tilde{\theta},\beta)\|_1\leq \Bar Q_n+r_n.$ Therefore, $\delta_\varepsilon\leq \Bar Q_n+r_n$ on the event, and hence $\{\Tilde{\Theta}_n\cap A_\varepsilon\neq\emptyset\}\subseteq \{\Bar Q_n+r_n\geq \delta_\varepsilon\}.$
Since $\Bar Q_n\convp 0$ and $r_n\convp 0$, it follows that $\Bar Q_n+r_n\convp 0$. Since $\delta_\varepsilon>0$ is fixed $P(\Tilde{\Theta}_n\cap A_\varepsilon\neq\emptyset)
\leq P(\Bar Q_n+r_n\geq \delta_\varepsilon)\to 0.$ Therefore:
\begin{align*}
P\left(\sup_{\Tilde{\theta}\in\Tilde{\Theta}_n}
\inf_{\Tilde{\theta}_0\in\Tilde{\Theta}}
\|\Tilde{\theta}-\Tilde{\theta}_0\|
\geq \varepsilon
\right)&=P(\Tilde{\Theta}_n\cap A_\varepsilon\neq\emptyset) \\
&\leq P(\Bar Q_n+r_n\geq \delta_\varepsilon)
\to 0.
\end{align*}

\medskip
\underline{\textbf{Step 2b.} \textit{$\sup_{\Tilde{\theta}_0\in\Tilde{\Theta}}\inf_{\Tilde{\theta}\in\Tilde{\Theta}_n}\|\Tilde{\theta}-\Tilde{\theta}_0\|\convp 0$.}}

If $\Tilde{\Theta}$ is a singleton, the claim follows immediately from Step 2a. Suppose that $\Tilde{\Theta}$ is non-singleton. Define the zero-set correspondence $\Tilde{\Theta}^0(\beta')
:=
\{\Tilde{\theta}\in\mathfrak{T}: g(\Tilde{\theta},\beta')=0\}.$ I verify the conditions of Lemma A.1 of \citet{shi2015simple} for this correspondence at $\beta'=\beta$. By construction, $\mathfrak{T}$ is a product of compact intervals, so $\mathfrak{T}$ is compact and $\operatorname{cl}(\operatorname{int}(\mathfrak{T}))=\mathfrak{T}$. Next, the function $g$ is polynomial in $(\Tilde{\theta},\beta)$ and hence continuously differentiable in $\Tilde{\theta}$ for each fixed $\beta'$ in an open set containing $\mathfrak{B}$. Now note that the Jacobian of $g$ with respect to $\Tilde{\theta}=(\theta,\lambda)$ has the block structure:
\begin{align*}
\frac{\partial g(\Tilde{\theta},\beta)}{\partial \Tilde{\theta}'}
=
\begin{pmatrix}
\partial \Tilde{h}(\theta,\beta)/\partial \theta' & -I_q \\
\partial \Tilde{g}(\theta,\beta)/\partial \theta' & 0
\end{pmatrix},
\end{align*}
where $I_q$ is the $q\times q$ identity matrix. The $-I_q$ block arises because $\partial(\Tilde{h}_t-\lambda_t)/\partial\lambda_t=-1$ and $\partial(\Tilde{h}_t-\lambda_t)/\partial\lambda_s=0$ for $s\neq t$, while the remaining equalities do not involve $\lambda$. Since the slackness rows have a $-I_q$ block in the $\lambda$-columns while all other rows have zeros there, the slackness rows are linearly independent of each other and of all other rows. Therefore the full Jacobian has full row rank if and only if $\partial \Tilde{g}(\theta,\beta)/\partial\theta'$ has full row rank, which holds by Assumption \ref{ass:estimation} \textit{iii)}. Finally, take any $(\theta,\lambda)\in\Tilde{\Theta}$. By Assumption \ref{ass:estimation} \textit{iv)}, there exists a sequence $\{\theta^{(j)}\}\subseteq \operatorname{int}(\mathcal H^{ie}(\theta))\cap \mathcal H(\theta)$ with $\theta^{(j)}\to\theta$. Since $\theta^{(j)}\in \operatorname{int}(\mathcal H^{ie}(\theta))$, all inequalities hold strictly, so $\Tilde h(\theta^{(j)},\beta)>0$ componentwise. Setting $\lambda^{(j)}=\Tilde h(\theta^{(j)},\beta)$, we have $\lambda^{(j)}\in\operatorname{int}(\Lambda)$ because $0<\lambda_t^{(j)}=\Tilde h_t(\theta^{(j)},\beta)\leq \bar h_t<\bar h_t+1$ for each $t$. By the enlargement of $\Theta$, $\mathcal H(\theta)\subset\operatorname{int}(\Theta)$, so $\theta^{(j)}\in\operatorname{int}(\Theta)$. Therefore $(\theta^{(j)},\lambda^{(j)})\in \Tilde{\Theta}\cap\operatorname{int}(\mathfrak{T}),$ and $(\theta^{(j)},\lambda^{(j)})\to (\theta,\lambda)$ by continuity of $\Tilde h$. Hence, $\operatorname{cl}(\Tilde{\Theta}\cap\operatorname{int}(\mathfrak{T}))=\Tilde{\Theta}$.

Lemma A.1 of \citet{shi2015simple} therefore implies that
$\Tilde{\Theta}^0(\cdot)$ is continuous at $\beta$. Fix any
$\varepsilon>0$. By continuity, there exists $\delta>0$ such that
$\|\beta'-\beta\|<\delta$ implies
$d_H(\Tilde{\Theta}^0(\beta'),\Tilde{\Theta})<\varepsilon$ and
$\Tilde{\Theta}^0(\beta')\neq\emptyset$. Then: 
\begin{align}\label{eq:conv_step}
P\left(d_H(\Tilde{\Theta}^0(\beta_n),\Tilde{\Theta})\geq\varepsilon\right)
+
P\left(\Tilde{\Theta}^0(\beta_n)=\emptyset\right)
\leq
2\,P\left(\|\beta_n-\beta\|\geq\delta\right)
\to 0,
\end{align}
where the inequality holds because both events on the left are subsets of
$\{\|\beta_n-\beta\|\geq\delta\}$ by continuity, and the convergence follows from
$\beta_n\convp\beta$. On the event
$\{\Tilde{\Theta}^0(\beta_n)\neq\emptyset\}$,
$\Bar{Q}_n\leq\|g(\Tilde{\theta},\beta_n)\|_1=0$ for any
$\Tilde{\theta}\in\Tilde{\Theta}^0(\beta_n)$, so $\Bar{Q}_n=0$ and
$\Tilde{\Theta}_n=\Tilde{\Theta}^0(\beta_n)$. Therefore,
\begin{align*}
P\left(d_H(\Tilde{\Theta}_n,\Tilde{\Theta})\geq\varepsilon\right)
\leq
P\left(d_H(\Tilde{\Theta}^0(\beta_n),\Tilde{\Theta})\geq\varepsilon\right)
+
P\left(\Tilde{\Theta}^0(\beta_n)=\emptyset\right)
\to 0,
\end{align*}
where the inequality follows because on the event
$\{\Tilde{\Theta}^0(\beta_n)\neq\emptyset\}$,
$\Tilde{\Theta}_n=\Tilde{\Theta}^0(\beta_n)$ and hence
$d_H(\Tilde{\Theta}_n,\Tilde{\Theta})
=d_H(\Tilde{\Theta}^0(\beta_n),\Tilde{\Theta})$, so the event
$\{d_H(\Tilde{\Theta}_n,\Tilde{\Theta})\geq\varepsilon\}$ is contained
in $\{d_H(\Tilde{\Theta}^0(\beta_n),\Tilde{\Theta})\geq\varepsilon\}
\cup\{\Tilde{\Theta}^0(\beta_n)=\emptyset\}$, and the convergence
from \eqref{eq:conv_step}. Hence
$d_H(\Tilde{\Theta}_n,\Tilde{\Theta})\convp 0$ when $\Tilde{\Theta}$
is non-singleton. Steps 2a and 2b together yield
$d_H(\Tilde{\Theta}_n,\Tilde{\Theta})\convp 0$.

\medskip

\underline{\textbf{Step 3.} \textit{Hausdorff consistency of $\Tilde{\Theta}_n$ implies Hausdorff consistency of $\mathcal{H}_n(\theta)$.}}

Let $\pi:\mathfrak{T}\to\Theta$ denote the coordinate projection $\pi(\Tilde{\theta})=\theta$ for $\Tilde{\theta}=(\theta,\lambda)$. For any $\Tilde{\theta}_1=(\theta_1,\lambda_1)$ and $\Tilde{\theta}_2=(\theta_2,\lambda_2)$ in $\mathfrak{T}$, $\|\pi(\Tilde{\theta}_1)-\pi(\Tilde{\theta}_2)\|
=
\|\theta_1-\theta_2\|
\leq
\|(\theta_1,\lambda_1)-(\theta_2,\lambda_2)\|
=
\|\Tilde{\theta}_1-\Tilde{\theta}_2\|.$ Thus Hausdorff distance cannot increase under $\pi$. Therefore, by Step 2:
\begin{align*}
d_H(\mathcal{H}_n(\theta),\mathcal{H}(\theta))
=
d_H(\pi(\Tilde{\Theta}_n),\pi(\Tilde{\Theta}))
\leq
d_H(\Tilde{\Theta}_n,\Tilde{\Theta})
\convp 0.
\end{align*}
where the equality follows since, by arguments in Step 2, $\pi(\Tilde{\Theta})=\mathcal{H}(\theta)$ and $
\pi(\Tilde{\Theta}_n)=\mathcal{H}_n(\theta).$
Step 1 then yields $d_H(\mathcal{H}_n(\tau),\mathcal{H}(\tau))\convp 0$.

\end{proof}

\begin{lemma}\label{lem:rectangular_rset}
Let $\mathbf{R}_0$ and $\mathbf{R}_1$ be two random closed sets. Then $Sel(\mathbf{R}_0 \times \mathbf{R}_1) = Sel(\mathbf{R}_0) \times Sel(\mathbf{R}_1)$.
\end{lemma}

\begin{proof}
Fix an arbitrary selection $(\rho_0, \rho_1) \in Sel(\mathbf{R}_0 \times \mathbf{R}_1)$. Then:
\begin{align}
    \begin{split}
        1 &= P\Big((\rho_0, \rho_1) \in \mathbf{R}_0 \times \mathbf{R}_1\Big)= P\Big(\rho_0 \in \mathbf{R}_0, \rho_1 \in \mathbf{R}_1\Big)\leq P\Big(\rho_0 \in \mathbf{R}_0\Big)
    \end{split}
\end{align}
where the first equality follows by $(\rho_0, \rho_1) \in Sel(\mathbf{R}_0 \times \mathbf{R}_1)$, the second equality and the inequality are by observation. Hence $P(\rho_0 \in \mathbf{R}_0) = 1$. By a similar argument, $P(\rho_1 \in \mathbf{R}_1) = 1$. Therefore $(\rho_0, \rho_1) \in Sel(\mathbf{R}_0) \times Sel(\mathbf{R}_1)$.

Next, fix an arbitrary $(\rho_0, \rho_1) \in Sel(\mathbf{R}_0) \times Sel(\mathbf{R}_1)$. Then:
\begin{align}
    \begin{split}
        1 &= P\Big(\rho_0 \in \mathbf{R}_0\Big)= P\Big(\rho_0 \in \mathbf{R}_0, \rho_1 \in \mathbf{R}_1\Big) + P\Big(\rho_0 \in \mathbf{R}_0, \rho_1 \notin \mathbf{R}_1\Big)\\
        &= P\Big(\rho_0 \in \mathbf{R}_0, \rho_1 \in \mathbf{R}_1\Big)= P\Big((\rho_0, \rho_1) \in \mathbf{R}_0 \times \mathbf{R}_1\Big)
    \end{split}
\end{align}
where the first equality is by $\rho_0 \in Sel(\mathbf{R}_0)$, the second is by observation, the third is since $P\Big(\rho_0 \in \mathbf{R}_0, \rho_1 \notin \mathbf{R}_1\Big) \leq P\Big(\rho_1 \notin \mathbf{R}_1\Big) = 0$ given that $\rho_1 \in Sel(\mathbf{R}_1)$, and the fourth is by observation. Thus $(\rho_0, \rho_1) \in Sel(\mathbf{R}_0 \times \mathbf{R}_1)$.
\end{proof}

\begin{lemma}\label{lem:prob_space}
    Suppose the probability space $(\Omega, \mathcal{F},P)$ is non-atomic and that $\mathcal{F}_0\subseteq \mathcal{F}$ is a sub-$\sigma$-algebra. After possibly enlarging the probability space, $P$ is atomless over $(\Omega,\mathcal{F}_0)$. That is, for all $A\in\mathcal{F}$ with $P(A)>0$ there exists $B\in\mathcal{F}$ with $B\subseteq A$ such that $0<P(B|\mathcal{F}_0)<P(A|\mathcal{F}_0)$ with positive probability. 
\end{lemma}

\begin{proof}
Fix $A\in\mathcal F$ with $P(A)>0$. There exists a 
random variable $U:\Omega\to[0,1]$ which is uniformly distributed on $[0,1]$ and independent of the 
$\sigma$–algebra $\sigma(\mathcal F_0,\mathbbm{1}_A)$.\footnote{If necessary, replace
$(\Omega,\mathcal F,P)$ by $(\Omega\times[0,1],\mathcal F\otimes\mathcal B([0,1]),P\otimes\lambda)$
and let $U(\omega,u)=u$.}

Define:
\begin{align}
    B := A\cap\{U\leq\tfrac12\}.
\end{align}

It is immediate that $B\in\mathcal F$ and $B\subset A$. By independence of $U$ from 
$\sigma(\mathcal F_0,\mathbbm{1}_A)$:
\begin{align*}
    P(B|\mathcal F_0)
= E\left[\mathbbm{1}_A \mathbbm{1}_{\{U\leq 1/2\}}|\mathcal F_0\right]
= E\left[\mathbbm{1}_A|\mathcal F_0\right]
    E\left[\mathbbm{1}_{\{U\leq 1/2\}}\right]
= \frac{1}{2} P(A|\mathcal F_0)\quad\text{$P-$a.s.}
\end{align*}
Because $E[P(A|\mathcal F_0)] = P(A)>0$:
\begin{align*}
    P\big(P(A|\mathcal F_0)>0\big) > 0,
\end{align*}
On the event $\{P(A|\mathcal F_0)>0\}$ it follows $0 < P(B|\mathcal F_0) = \tfrac12 P(A|\mathcal F_0) < P(A|\mathcal F_0),$ and thus:
\begin{align}
    \big\{0 < P(B|\mathcal F_0) < P(A|\mathcal F_0)\big\}
\supseteq \big\{P(A|\mathcal F_0)>0\big\}.
\end{align}

Therefore $\big\{0 < P(B|\mathcal F_0) < P(A|\mathcal F_0)\big\}$ has strictly positive probability.
\end{proof}

\begin{lemma}\label{lem:selections_combined}
Suppose that Assumptions \ref{ass:rand_assign}, \ref{ass:ex_validity} hold. Then the identified set $\mathcal{H}(P_{Y(0),S(0)},P_{Y(1),S(1)})$ for marginal distribution functions $P_{Y(d),S(d)}$ is:
\begin{equation}\label{eq:identified_cartesian}
    \mathcal{H}(P_{Y(0),S(0)},P_{Y(1),S(1)}) = \mathcal{H}(P_{Y(0),S(0)})\times\mathcal{H}(P_{Y(1),S(1)})\neq \emptyset.
\end{equation}
Let $C(B)=\{s:\mathcal{Y}\times\{s\}\subseteq B\}$. $\mathcal{H}(P_{Y(d),S(d)})$ when combined data are used is:
\begin{align}\label{eq:marginals_joint_definition}
\left\{\delta\in\mathcal{P}^{\mathcal{Y}\times \mathcal{S}}:\begin{array}{ll}
     & \forall B\in\mathcal{C}(\mathcal{Y}\times \mathcal{S}), \\
     &\delta(B)\geq  \max\left\{
       P_O\big((Y,S)\in B,D=d\big),
       ess\sup_{Z} P_E\big(S\in C(B),D=d|Z\big)
   \right\}
\end{array}\right\}.
\end{align}

If the experimental data are not observed, then $\mathcal{H}(P_{Y(d),S(d)})$ is equivalent to:
\begin{align}\label{eq:marginals_observational_definition}
 \{\delta\in\mathcal{P}^{\mathcal{Y}\times\mathcal{S}}: \forall B\in\mathcal{C}(\mathcal{Y}\times\mathcal{S}), \enskip \delta(B) \geq P_O((Y,S)\in B,D=d) \}
\end{align}
\end{lemma}

\begin{proof}
The proof proceeds by extending arguments of \citet[Proposition 2.3]{beresteanu2012partial}. Recall that $ \Tilde{Z} = \mathbbm{1}[G=E] Z+\mathbbm{1}[G=O](sup\mathcal{Z}+1)\in \Tilde{\mathcal{Z}}$. Note that $\Tilde{Z} =Z$ when $G=E$ and $\Tilde{Z}$ equals a distinct constant when $G=O$. Therefore, Assumptions \ref{ass:rand_assign} and \ref{ass:ex_validity} hold if and only if $\Tilde{Z}\independent (Y(d),S(d))$ for all $d\in\{0,1\}$. Let $\Tilde{I}$ be the set of random elements $(E_1,E_2,E_3)$ such that $(E_1,E_2,E_3)\in \mathcal{Y}\times\mathcal{S}\times\Tilde{\mathcal{Z}}$ and $(E_1,E_2)\independent E_3$, i.e. that satisfy the two assumptions. Define the random set:
    \begin{align}
        (\mathbf{Y}_0,\mathbf{S}_0,\mathbf{Y}_1,\mathbf{S}_1) :=\begin{cases}
            \{(Y,S)\}\times\mathcal{Y}\times\mathcal{S}, \text{ if $(D,G)=(0,O)$}\\
            \mathcal{Y}\times\mathcal{S}\times\{(Y,S)\}, \text{ if $(D,G)=(1,O)$}\\
            \mathcal{Y}\times\{S\}\times\mathcal{Y}\times\mathcal{S}, \text{ if $(D,G)=(0,E)$}\\
            \mathcal{Y}\times\mathcal{S}\times\mathcal{Y}\times\{S\}, \text{ if $(D,G)=(1,E)$}\\
        \mathcal{Y}\times\mathcal{S}\times\mathcal{Y}\times\mathcal{S}, \text{ otherwise}
        \end{cases}.
    \end{align}

By definition of $(\mathbf{Y}_0,\mathbf{S}_0,\mathbf{Y}_1,\mathbf{S}_1)$, all information on $(Y(0),S(0),Y(1),S(1))$ in the combined data can be summarized by $(Y(0),S(0),Y(1),S(1))\in Sel((\mathbf{Y}_0,\mathbf{S}_0,\mathbf{Y}_1,\mathbf{S}_1))$. Also, by definition of the random set, $(\mathbf{Y}_0,\mathbf{S}_0,\mathbf{Y}_1,\mathbf{S}_1) =(\mathbf{Y}_0,\mathbf{S}_0)\times(\mathbf{Y}_1,\mathbf{S}_1)$.

By \Cref{lem:rectangular_rset}, $Sel((\mathbf{Y}_0,\mathbf{S}_0,\mathbf{Y}_1,\mathbf{S}_1)) =Sel((\mathbf{Y}_0,\mathbf{S}_0))\times Sel((\mathbf{Y}_1,\mathbf{S}_1))$. All information in the data and assumptions about $(Y(0),S(0),Y(1),S(1))$ can thus equivalently be expressed as $(Y(d),S(d),\Tilde{Z})\in Sel((\mathbf{Y}_d,\mathbf{S}_d),\Tilde{Z})\cap \Tilde{I}$ for $d\in\{0,1\}$. If Assumptions \ref{ass:rand_assign} and \ref{ass:ex_validity} hold, $(Y(d),S(d),\Tilde{Z})\in Sel((\mathbf{Y}_d,\mathbf{S}_d),\Tilde{Z})\cap \Tilde{I}\neq \emptyset$ for $d\in\{0,1\}$. Then:
\begin{align}\label{eq:marginals_set}
\begin{split}
     &\mathcal{H}(P_{Y(0),S(0)},P_{Y(1),S(1)}) \\
     &= \left\{(\delta_0,\delta_1)\in\left(\mathcal{P}^{\mathcal{Y}\times \mathcal{S}}\right)^2:\forall d\in\{0,1\}, \enskip \exists  (\upsilon_d,\varsigma_d,\Tilde{Z})\in Sel((\mathbf{Y}_d,\mathbf{S}_d),\Tilde{Z})\cap \Tilde{I} \text{ s.t. $\delta_{d}\eqd(\upsilon_d,\varsigma_d)$} \right\}\\
    &=\bigtimes_{d\in\{0,1\}} \left\{\delta_{d} \in\mathcal{P}^{\mathcal{Y}\times \mathcal{S}}:\exists  (\upsilon_d,\varsigma_d,\Tilde{Z})\in Sel((\mathbf{Y}_d,\mathbf{S}_d),\Tilde{Z})\cap \Tilde{I} \text{ s.t. $\delta_{d}\eqd(\upsilon_d,\varsigma_d)$}\right\}\\
    &=\mathcal{H}(P_{Y(0),S(0)})\times \mathcal{H}(P_{Y(1),S(1)})\neq\emptyset
\end{split}
\end{align}
where the first equality follows by definition, the second by observation, the third is by definition of $\mathcal{H}(P_{Y(d),S(d)})$, and the final since $(Y(d),S(d),\Tilde{Z})\in Sel((\mathbf{Y}_d,\mathbf{S}_d),\Tilde{Z})\cap \Tilde{I}\neq \emptyset$ for $d\in\{0,1\}$, showing \eqref{eq:identified_cartesian}.

By \citet[Theorem 2.1]{artstein1983distributions}, the distribution function $P((Y(d),S(d),\Tilde{Z}))\in\mathcal{P}^{\mathcal{Y\times\mathcal{S}\times\Tilde{\mathcal{Z}}}}$ characterizes a selection in $Sel((\mathbf{Y}_d,\mathbf{S}_d),\Tilde{Z})$ if and only if:
\begin{align}\label{eq:artstein_worst_case_lemma_combined}
   \forall B\in\mathcal{C}(\mathcal{Y}\times \mathcal{S}\times \Tilde{\mathcal{Z}}):\quad P((Y(d),S(d),\Tilde{Z})\in B) \geq P(((\mathbf{Y}_d,\mathbf{S}_d),\Tilde{Z})\subseteq B) 
\end{align}

By \citet[Theorem 2.33]{molchanov2018random}, \eqref{eq:artstein_worst_case_lemma_combined} is equivalent to:
\begin{align}\label{eq:artstein_worst_case_as}
       \forall B\in\mathcal{C}(\mathcal{Y}\times \mathcal{S}):\quad P((Y(d),S(d))\in B|\Tilde{Z}) \geq P((\mathbf{Y}_d,\mathbf{S}_d)\subseteq B|\Tilde{Z}) \text{ $P-$a.s.} 
\end{align}

Possible forms that $B$ can take are: 1) $B = \mathcal{Y}\times\mathcal{S}$;\footnote{The support of a random vector $X$ is the smallest closed set $\mathcal{X}$ such that $P(X\in\mathcal{X}) = 1$. Hence $\mathcal{Y}\times\mathcal{S}\in \mathcal{C}(\mathcal{Y}\times\mathcal{S})$.} 2) $B \subsetneq \mathcal{Y}\times\mathcal{S}$. For $B = \mathcal{Y}\times\mathcal{S}$, $P((\mathbf{Y}_d,\mathbf{S}_d)\subseteq B|\Tilde{Z}) = 1$ $P-$a.s. Note that the constraint imposed by the containment functional in this case is redundant when $P((Y(d),S(d),\Tilde{Z}))\in\mathcal{P}^{\mathcal{Y\times\mathcal{S}\times\Tilde{\mathcal{Z}}}}$, given that $P((Y(d),S(d))\in\mathcal{Y}\times\mathcal{S}|\Tilde{Z}) = 1$ $P-$a.s. by definition.

Consider the containment functional for $B\in\mathcal C(\mathcal{Y}\times\mathcal{S})$ when $B \subsetneq \mathcal{Y}\times\mathcal{S}$. If $D\neq d$, $(\mathbf{Y}_d,\mathbf{S}_d) = \mathcal Y\times\mathcal S$, so the random set can be a subset of $B$ only if $D=d$. Events $\{G=O,D=d\}$ and $\{G=E,D=d\}$ are mutually exclusive. Hence, for any closed $B\in\mathcal C(\mathcal{Y}\times\mathcal{S})$ with $B \subsetneq \mathcal{Y}\times\mathcal{S}$, $P-$a.s:
\begin{align*}
P\big((\mathbf{Y}_d,\mathbf{S}_d)\subseteq B|\Tilde{Z}\big) &=  P\big(G=O,D=d,(\mathbf{Y}_d,\mathbf{S}_d)\subseteq B|\Tilde{Z}\big)+P\big(G=E, D=d,(\mathbf{Y}_d,\mathbf{S}_d)\subseteq B|\Tilde{Z}\big)  \\
&=P\big(G=O,D=d,(Y,S)\in B|\Tilde{Z}\big)+P\big(G=E,D=d, \mathcal{Y}\times\{S\}\subseteq B|\Tilde{Z}\big)\\
&=P\big(G=O,D=d,(Y,S)\in B|\Tilde{Z}\big)+P\big(G=E,D=d,S\in C(B)|\Tilde{Z}\big).
\end{align*}
where the first line follows by the definition of the random set and the fact that events $\{G=O,D=d\}$ and $\{G=E,D=d\}$ are mutually exclusive, the second follows by definition of the random set, and the third by definition of $C(B)$.

Now, use the non-redundant case to characterize the constraints imposed by \eqref{eq:artstein_worst_case_as}. The distribution function $P((Y(d),S(d),\Tilde{Z}))\in\mathcal{P}^{\mathcal{Y\times\mathcal{S}\times\Tilde{\mathcal{Z}}}}$ characterizes a selection in  $Sel((\mathbf{Y}_d,\mathbf{S}_d),\Tilde{Z})$ if and only if $\forall B\in\mathcal{C}(\mathcal{Y}\times \mathcal{S})$ such that $B \subsetneq \mathcal{Y}\times\mathcal{S}$ $P-$a.s.:
\begin{align}
\begin{split}
   P\big((Y(d),S(d))\in B\big|\Tilde{Z})
   &\geq 
   P\big(G=O,D=d,(Y,S)\in B|\Tilde{Z}\big)
   +
   P\big(G=E,D=d,S\in C(B)|\Tilde{Z}\big).
\end{split}
\end{align}
To incorporate the fact that $\Tilde{Z}\independent (Y(d),S(d))$, intersect $Sel((\mathbf{Y}_d,\mathbf{S}_d),\Tilde{Z})\cap \Tilde{I}$ which yields $\forall B\in\mathcal{C}(\mathcal{Y}\times \mathcal{S})$ such that $B \subsetneq \mathcal{Y}\times\mathcal{S}$:

\begin{align}\label{eq:intersect_artstein}
\begin{split}
   P\big((Y(d),S(d))\in B\big)
   &\geq   ess\sup_{\Tilde{Z}} \left[\begin{array}{ll}
       &P\big(G=O,D=d,(Y,S)\in B|\Tilde{Z}\big)
   +\\
   &P\big(G=E,D=d,S\in C(B)|\Tilde{Z}\big)
   \end{array}\right]\\
   &= 
   \max\left\{
       P_O\big((Y,S)\in B,D=d\big),
       ess\sup_{Z} P_E\big(S\in C(B),D=d|Z\big)
   \right\},
\end{split}
\end{align}
where the equality follows by the definition of $\Tilde Z$. If $\Tilde Z\notin\mathcal{Z}$ (i.e. $G=O$), $P\big(G=E,D=d,S\in C(B)|\Tilde{Z}\big)=0$ and the expression reduces to $P_O\big((Y,S)\in B,D=d\big)$. If $\Tilde Z\in\mathcal{Z}$ (i.e. $G=E$), $P\big(G=O,D=d,(Y,S)\in B|\Tilde{Z}\big)=0$ and the expression reduces to $ess\sup_{Z} P_E\big(S\in C(B),D=d|Z\big)$.

Finally, note that when experimental data are not observed, $P(G=O)=1$ and thus $P\big(G=E,D=d,S\in C(B)|\Tilde{Z}\big)=0$. Thus, \eqref{eq:intersect_artstein} simplifies to:
\begin{align}\label{eq:intersect_artstein_observational}
\begin{split}
   P((Y(d),S(d))\in B)&\geq P_O\big((Y,S)\in B,D=d\big)
\end{split}
\end{align}
Thus, \eqref{eq:marginals_observational_definition} follows by \eqref{eq:artstein_worst_case_as} and \eqref{eq:intersect_artstein_observational}.

Sharpness follows by construction. For any $(P_{Y(0),S(0)},P_{Y(1),S(1)})\in\mathcal{H}(P_{Y(0),S(0)},P_{Y(1),S(1)})$ there exist $(Y(0),S(0),Y(1),S(1))$that are consistent with the data and assumptions such that $(Y(d),S(d))\eqd P_{Y(d),S(d)}$ for $d\in\{0,1\}$.
\end{proof}

\begin{lemma}\label{lem:marginals_y}
Let $\mathcal{H}(P_{Y(0),S(0)},P_{Y(1),S(1)})$, and $\mathcal{H}^O(P_{Y(0),S(0)},P_{Y(1),S(1)})$ be the identified sets for pairs of distributions $P_{Y(d),S(d)}$ for $d\in\{0,1\}$ when the experimental data are observed and unobserved, respectively. Let $\mathcal{H}(P_{Y(d)})$, $\mathcal{H}(P_{S(d)})$, $\mathcal{H}^O(P_{Y(d)})$, and $\mathcal{H}^O(P_{S(d)})$ be the identified sets for the corresponding marginals. Suppose only that Assumptions \ref{ass:rand_assign}, \ref{ass:ex_validity} hold. Then:
\begin{enumerate}[label = \roman*)] 
    \item  $\mathcal{H}^O(P_{Y(0)},P_{Y(1)})=\mathcal{H}(P_{Y(0)},P_{Y(1)})$.
    \item $\mathcal{H}(P_{Y(0)},P_{S(0)},P_{Y(1)},P_{S(1)})=\bigtimes_{d\in\{0,1\}}\Big(\mathcal{H}(P_{Y(d)})\times\mathcal H(P_{S(d)})\Big)$;
    \item $\mathcal{H}^O(P_{Y(0)},P_{S(0)},P_{Y(1)},P_{S(1)})=\bigtimes_{d\in\{0,1\}}\Big(\mathcal{H}^O(P_{Y(d)})\times\mathcal H^O(P_{S(d)})\Big)$;
\end{enumerate}
\end{lemma}
\begin{proof}

\textit{i)}

Let $\pi_{k}\delta$ denote the projections of a distribution function $\delta$ onto the $k$-th marginal. Observe that:
\begin{align}\label{eq:joint_set_rectangular}
\begin{split}
    \mathcal{H}(P_{Y(0)},P_{Y(1)}) = \left\{(\pi_1\delta_0,\pi_1\delta_1) \in(\mathcal P^\mathcal Y)^2:  (\delta_0,\delta_1)\in\mathcal{H}(P_{Y(0),S(0)},P_{Y(1),S(1)})\right\}\\
= \bigtimes_{d\in\{0,1\}}\left\{\pi_1\delta_d \in\mathcal P^\mathcal Y:  \delta_d\in\mathcal{H}(P_{Y(d),S(d)})\right\} = \bigtimes_{d\in\{0,1\}} \mathcal{H}(P_{Y(d)})
\end{split}
\end{align}
where the first equality is by definition of $\mathcal{H}(P_{Y(0)},P_{Y(1)})$, the second follows since by \Cref{lem:selections_combined} $\mathcal{H}(P_{Y(0),S(0)},P_{Y(1),S(1)}) = \mathcal{H}(P_{Y(0),S(0)})\times \mathcal{H}(P_{Y(1),S(1)})$, and the third is by definition of $\mathcal{H}(P_{Y(d)})$. By an analogous argument:
\begin{align}\label{eq:joint_set_rectangular_observational}
\begin{split}
    \mathcal{H}^O(P_{Y(0)},P_{Y(1)}) = \bigtimes_{d\in\{0,1\}} \mathcal{H}^O(P_{Y(d)})
\end{split}
\end{align}

Hence, it suffices to show that $\mathcal{H}(P_{Y(d)}) = \mathcal{H}^O(P_{Y(d)})$ for $d\in\{0,1\}$. 

To that end, fix any $d\in\{0,1\}$. I show that $\mathcal{H}(P_{Y(d)})\subseteq\mathcal{H}^O(P_{Y(d)})$ and $\mathcal{H}^O(P_{Y(d)})\subseteq\Tilde{\mathcal{H}} \subseteq \mathcal{H}(P_{Y(d)})$, where $\Tilde{\mathcal{H}} := \left\{\delta\in\mathcal P^\mathcal Y: \forall B\in\mathcal C(\mathcal Y), \enskip \delta(B)\geq P_O(Y\in B,D=d) \right\}$.

\underline{$\mathcal{H}(P_{Y(d)})\subseteq\mathcal{H}^O(P_{Y(d)})$}

Observe that:
\begin{align*}
       &\max\left\{
       P_O\big((Y,S)\in B,D=d\big),
       ess\sup_{Z} P_E\big(S\in C(B),D=d|Z\big)
   \right\}\geq P_O((Y,S)\in B,D=d).
\end{align*}
By \Cref{lem:selections_combined} then $\mathcal{H}(P_{Y(d),S(d)})\subseteq \mathcal{H}^O(P_{Y(d),S(d)})$, and therefore $\mathcal{H}(P_{Y(d)})\subseteq\mathcal{H}^O(P_{Y(d)})$.   

\underline{$\mathcal{H}^O(P_{Y(d)})\subseteq\Tilde{\mathcal{H}}$}

Fix an arbitrary $\delta\in\mathcal{H}^O(P_{Y(d)})$. Let $\bar \delta\in\mathcal{H}^O(P_{Y(d),S(d)})$ be such that $\delta = \pi_1\bar\delta$. By \Cref{lem:selections_combined}, for any $B\in\mathcal C(\mathcal Y)$:
\begin{align*}
\delta(B)
= \bar\delta(B\times\mathcal S) \geq P_O((Y,S)\in B\times\mathcal S,D=d) = P_O(Y\in B,D=d).
\end{align*}

Hence, $\mathcal{H}^O(P_{Y(d)})\subseteq\Tilde{\mathcal{H}}$.

\underline{$\Tilde{\mathcal{H}} \subseteq \mathcal{H}(P_{Y(d)})$.}

Fix an arbitrary $\delta_{Y(d)}\in \Tilde{\mathcal{H}}$. I construct a distribution $\delta_{Y(d),S(d)} \in \mathcal{H}(P_{Y(d),S(d)})$ whose $Y(d)$–marginal is $\delta_{Y(d)}$. For $d \in \{0,1\}$ and $B_S \in \mathcal{C}(\mathcal{S})$, define:
\begin{align}
\begin{split}
       & L_d(B_S):= \max\left\{P_O(S \in B_S, D=d),ess\sup_{Z} P_E(S \in B_S, D=d | Z) \right\}.
\end{split}
\end{align}

If Assumptions \ref{ass:rand_assign} and \ref{ass:ex_validity} hold, by \Cref{lem:selections_combined} there exists at least one joint distribution $\bar\delta \in \mathcal H\big(P_{Y(d),S(d)}\big) \subseteq \mathcal P^{\mathcal Y\times\mathcal S}$. By the lemma, for any $B\in\mathcal C(\mathcal Y\times\mathcal S)$:
\begin{equation}
\label{eq:H-ineq-bar-delta}
\bar\delta(B)
\geq
\max\left\{
       P_O\big((Y,S)\in B,D=d\big),
       ess\sup_{Z} P_E\big(S\in C(B),D=d|Z\big)
   \right\}
\end{equation}

Observe that $ \pi_{2}\bar\delta$ denotes the $S(d)$–marginal of $\bar\delta$. Fix any $B_S\in\mathcal C(\mathcal S)$. By \eqref{eq:H-ineq-bar-delta}:
\begin{align*}
 \pi_{2}\bar\delta(B_S)=\bar\delta(\mathcal{Y}\times B_S) \geq \max\left\{P_O(S \in B_S, D=d),ess\sup_{Z} P_E(S \in B_S, D=d | Z) \right\}= L_d(B_S).
\end{align*}

Therefore, there always exists a distribution $\delta_{S(d)}$ such that $\delta_{S(d)}(B_S)\geq L_d(B_S)$ for any $B_S\in\mathcal C(\mathcal S)$. Fix any such $\delta_{S(d)} = \pi_{2}\bar\delta$. For $B_Y\subseteq\mathcal Y$, and $B_S\subseteq\mathcal S$, define finite signed measures:
\begin{align*}
    \nu_{d,Y}(B_Y):=\delta_{Y(d)}(B_Y)-P_O(Y\in B_Y,D=d),\qquad
\nu_{d,S}(B_S):=\delta_{S(d)}(B_S)-P_O(S\in B_S,D=d).
\end{align*}

Since $\delta_{Y(d)}\in \Tilde{\mathcal H}$ and $\delta_{S(d)}(B_S) \geq L_d(B_S)\geq P_O(S\in B_S,D=d)$, $\nu_{d,Y}$ and $\nu_{d,S}$ are nonnegative finite measures. Moreover, $\nu_{d,Y}(\mathcal Y) = \nu_{d,S}(\mathcal S) =  1 - P_O(D=d)$, so they have the same total mass. Since $\nu_{d,Y}$ and $\nu_{d,S}$ are non-negative finite measures on Polish spaces with the same total mass, there exists a coupling $\nu_d$ on $\mathcal Y\times\mathcal S$ with marginals $\nu_{d,Y}$ and $\nu_{d,S}$. Define for any $B\in\mathcal C(\mathcal Y\times\mathcal S)$:
\begin{align}
\delta_{Y(d),S(d)}(B):=P_O((Y,S)\in B,D=d) + \nu_d(B).
\end{align}

By construction, $\delta_{Y(d),S(d)}$ has marginals $\delta_{Y(d)}$ and $\delta_{S(d)}$:
\begin{align*}
\pi_1\delta_{Y(d),S(d)} &= P_O(Y\in\cdot,D=d) + \nu_{d,Y} = \delta_{Y(d)},\\
\pi_2\delta_{Y(d),S(d)} &= P_O(S\in\cdot,D=d) + \nu_{d,S} = \delta_{S(d)}.
\end{align*}

Again, by construction, for any $B\in\mathcal{C(\mathcal Y\times\mathcal S)}$:
\begin{align}\label{eq:couping_ineq_1}
    \delta_{Y(d),S(d)}(B) \geq P_O((Y,S) \in B, D=d)
\end{align}
Moreover, by definition $\mathcal{Y}\times C(B)\subseteq B$ for any $B\in\mathcal{C(\mathcal Y\times\mathcal S)}$. Therefore:
\begin{align}\label{eq:couping_ineq_2}
    \begin{split}
        \delta_{Y(d),S(d)}(B)&\geq\delta_{Y(d),S(d)}\big(\mathcal Y\times C(B)\big) = 
\delta_{S(d)}\big(C(B)\big)\\
&\geq \max\left\{P_O(S \in C(B), D=d),ess\sup_{Z} P_E(S \in C(B), D=d | Z) \right\}\\
&\geq ess\sup_{Z} P_E(S \in C(B), D=d | Z)
    \end{split}
\end{align}
where the first inequality is by $\mathcal{Y}\times C(B)\subseteq B$, the equality is by $\pi_2\delta_{Y(d),S(d)}  = \delta_{S(d)}$, the second inequality by $\delta_{S(d)}(B_S)\geq L_d(B_S)$ for any $B_S\in\mathcal C(\mathcal S)$, and the final inequality is by observation. Combining \eqref{eq:couping_ineq_1} and \eqref{eq:couping_ineq_2}, for any $B\in\mathcal{C(\mathcal Y\times\mathcal S)}$:
\begin{align}\label{eq:coupling_inequality}
    \delta_{Y(d),S(d)}(B) \geq
\max\left\{
       P_O\big((Y,S)\in B,D=d\big),
       ess\sup_{Z} P_E\big(S\in C(B),D=d|Z\big)
   \right\}
\end{align}

By \Cref{lem:selections_combined}, $\delta_{Y(d),S(d)}\in\mathcal{H}(P_{Y(d),S(d)})$. Therefore, $\delta_{Y(d)}\in\mathcal{H}(P_{Y(d)})$ so $\Tilde{\mathcal{H}} \subseteq \mathcal{H}(P_{Y(d)})$.

Combining all three statements, $\mathcal{H}(P_{Y(d)})\subseteq\mathcal{H}^O(P_{Y(d)})$ and $\mathcal{H}^O(P_{Y(d)})\subseteq\Tilde{\mathcal{H}} \subseteq \mathcal{H}(P_{Y(d)})$. This implies $\mathcal{H}(P_{Y(d)})=\Tilde{\mathcal{H}}=\mathcal{H}^O(P_{Y(d)})$. By \eqref{eq:joint_set_rectangular} and \eqref{eq:joint_set_rectangular_observational} then:

\begin{align}
\begin{split}
    \mathcal{H}(P_{Y(0)},P_{Y(1)}) = \bigtimes_{d\in\{0,1\}} \mathcal{H}(P_{Y(d)})= \bigtimes_{d\in\{0,1\}} \mathcal{H}^O(P_{Y(d)})=\mathcal{H}^O(P_{Y(0)},P_{Y(1)}).
\end{split}
\end{align}

\textit{ii)}

By definition:
\begin{align}\label{eq:cartesian_proof}
    \begin{split}
&\mathcal{H}(P_{Y(0)},P_{S(0)},P_{Y(1)},P_{S(1)})\\
&=\left\{(\pi_1\delta_0, \pi_2\delta_0,\pi_1\delta_1,\pi_2\delta_1) \in(\mathcal P^\mathcal Y)^2\times (\mathcal P^\mathcal S)^2:  (\delta_0,\delta_1)\in\mathcal{H}(P_{Y(0),S(0)},P_{Y(1),S(1)})\right\}\\
&= \bigtimes_{d\in\{0,1\}}\left\{(\pi_1\delta_d,\pi_2\delta_d) \in\mathcal P^\mathcal Y\times \mathcal P^\mathcal S:  \delta_d\in\mathcal{H}(P_{Y(d),S(d)})\right\}  = \bigtimes_{d\in\{0,1\}} \mathcal{H}(P_{Y(d)},P_{S(d)})
    \end{split}
\end{align}
where the second equality follows since by \Cref{lem:selections_combined} $\mathcal{H}(P_{Y(0),S(0)},P_{Y(1),S(1)}) = \mathcal{H}(P_{Y(0),S(0)})\times \mathcal{H}(P_{Y(1),S(1)})$ and the third is by definition.

It is thus sufficient to show that $\mathcal{H}(P_{Y(d)},P_{S(d)})=\mathcal{H}(P_{Y(d)})\times \mathcal{H}(P_{S(d)})$ for $d\in\{0,1\}$. Fix an arbitrary $d\in\{0,1\}$. Since a Cartesian product of projections of a set is always a superset of the set itself, it is immediate that $\mathcal{H}(P_{Y(d)},P_{S(d)})\subseteq \mathcal{H}(P_{Y(d)})\times \mathcal{H}(P_{S(d)})$. It remains to show that  $\mathcal{H}(P_{Y(d)})\times \mathcal{H}(P_{S(d)})\subseteq\mathcal{H}(P_{Y(d)},P_{S(d)})$. Let $\dtilde{H}:=\left\{\delta\in\mathcal P^\mathcal S: \forall B\in\mathcal C(\mathcal S), \enskip  
\delta(B)\geq L_d(B) \right\}$. Recall from the proof of \textit{i)} that $\Tilde{\mathcal{H}} = \left\{\delta\in\mathcal P^\mathcal Y: \forall B\in\mathcal C(\mathcal Y), \enskip  
\delta(B)\geq P_O(Y\in B,D=d) \right\}$. The proof also shows: 1) $\dtilde{\mathcal{H}}\neq\emptyset$ when Assumptions \ref{ass:rand_assign} and \ref{ass:ex_validity} hold; and 2) for any $\delta_{Y(d)}\in\tilde{\mathcal{H}}$ and any $\delta_{S(d)}\in{\dtilde{\mathcal{H}}}$ there exists a joint distribution $\delta_{Y(d),S(d)}\in\mathcal{H}
(P_{Y(d),S(d)})$ such that $\pi_1\delta_{Y(d),S(d)}=\delta_{Y(d)}$ 
and $\pi_2\delta_{Y(d),S(d)}=\delta_{S(d)}$ on $\mathcal C(\mathcal Y\times \mathcal S)$. Hence:
\begin{align}\label{eq:htilde_cartesian}
\Tilde{\mathcal{H}}\times\dtilde{\mathcal{H}}\subseteq \left\{(\pi_1\delta,\pi_2\delta)\in\mathcal P^{\mathcal Y}\times\mathcal P^{\mathcal S}:\delta \in \mathcal H(P_{Y(d),S(d)})\right\} = \mathcal{H}(P_{Y(d)},P_{S(d)})
    \end{align}
where the equality is by definition. Recall also that $\mathcal{H}(P_{Y(d)})=\Tilde{\mathcal{H}}$, by the proof of \textit{i)}. It thus remains to show that $\dtilde{\mathcal{H}} = \mathcal{H}(P_{S(d)})$. Since for any $\delta_{S(d)}\in \dtilde{\mathcal{H}}$ there exists  $\delta_{Y(d),S(d)}\in\mathcal H(P_{Y(d),S(d)})$ such that $\pi_2 \delta_{Y(d),S(d)}=\delta_{S(d)}$, then $\dtilde{\mathcal{H}} \subseteq \mathcal{H}(P_{S(d)})$. For the converse, note that for any $\delta_{S(d)}\in  \mathcal{H}(P_{S(d)})$ there exists $\delta_{Y(d),S(d)}\in \mathcal H(P_{Y(d),S(d)})$ such that $\pi_2\delta_{Y(d),S(d)} = \delta_{S(d)}$. By \Cref{lem:selections_combined}, for that $\delta_{Y(d),S(d)}$ then $\forall B\in\mathcal C(\mathcal Y\times \mathcal S)$:
\begin{align}\label{}
    \delta_{Y(d),S(d)}(B) \geq
\max\left\{
       P_O\big((Y,S)\in B,D=d\big),
       ess\sup_{Z} P_E\big(S\in C(B),D=d|Z\big)
   \right\}.
\end{align}

Fix any $B_S\in\mathcal C(\mathcal S)$. Then by definition of $C(B_S)$:
\begin{align}
\begin{split}
       \delta_{Y(d),S(d)}(\mathcal Y\times B_S) &= \delta_{S(d)}(B_S)\\
       &\geq  \max\left\{ess\sup_{Z}P_E(S\in B_S,D=d|Z),P_O(S\in B_S,D=d) \right\}=L_d(B_S).
\end{split}
\end{align}
Therefore, $\delta_{S(d)}\in\dtilde{\mathcal H}$ and thus $\mathcal{H}(P_{S(d)}) \subseteq\dtilde{\mathcal{H}} $.

It then follows that $\mathcal{H}(P_{S(d)}) =\dtilde{\mathcal{H}} $. Since $\mathcal{H}(P_{Y(d)})=\Tilde{\mathcal{H}}$, by \eqref{eq:htilde_cartesian} then $\mathcal{H}(P_{Y(d)})\times\mathcal{H}(P_{S(d)}) \subseteq \mathcal{H}(P_{Y(d)},P_{S(d)})$ and therefore $\mathcal{H}(P_{Y(d)})\times\mathcal{H}(P_{S(d)}) = \mathcal{H}(P_{Y(d)},P_{S(d)})$. The result is then by \eqref{eq:cartesian_proof}.

\textit{iii)}

The result follows from the steps in the proof of \textit{ii)} by redefining:
$L_d(B_S):=  P_O(S \in B_S, D=d).$
\end{proof}

\begin{lemma}\label{lem:removing_joint_independence}
    Let $\Tilde{Z} = \mathbbm{1}[G=E] Z+\mathbbm{1}[G=O](sup\mathcal{Z}+1)$ and $\Bar{I}$ be the set of random elements $(E_1,E_2)$ such that $(E_1,E_2)\in \mathcal{S}\times\Tilde{\mathcal{Z}}$ and $E_1\independent E_2$. Define random sets $\mathbf{Y}_d$ and $\mathbf{S}_d$ for $d\in\{0,1\}$:
\begin{align}
           \mathbf{Y}_d := \begin{cases}
            \{Y\}, \text{ if $(D,G)=(d,O)$}\\
            \mathcal{Y}, \text{ otherwise}
        \end{cases}, \enskip
        \mathbf{S}_d := \begin{cases}
            \{S\}, \text{ if $(D,G)\in\{(d,E),(d,O)\}$}\\
            \mathcal{S}, \text{ otherwise}
        \end{cases}.
\end{align}
Then $\mathcal{H}^{EV/RA}(m,\gamma) =  \Tilde{\mathcal{H}}^{EV/RA}(m,\gamma)$ for:
\begin{align*}
         \mathcal{H}^{EV/RA}(m,\gamma)&:= \left\{\begin{array}{ll} (m,\gamma)\in\mathcal{M}\times \left(\mathcal{P}^{\mathcal{S}}\right)^2: \forall d\in\{0,1\}, \enskip \exists(\varsigma_d,\Tilde{Z})\in Sel((\mathbf{S}_d,              \Tilde{Z}))\cap \Bar{I},\\
    \exists\upsilon_d\in Sel(\mathbf{Y}_d),\enskip (\upsilon_d,\varsigma_d)\independent \Tilde{Z}, \enskip \gamma_d \eqd \varsigma_d, \enskip m_d(\varsigma_d) = E_O[\upsilon_d|\varsigma_d] \text{ a.s.}
\end{array}\right\}\\
         \Tilde{\mathcal{H}}^{EV/RA}(m,\gamma)&:= \left\{\begin{array}{ll} (m,\gamma)\in\mathcal{M}\times \left(\mathcal{P}^{\mathcal{S}}\right)^2: \forall d\in\{0,1\}, \enskip \exists(\varsigma_d,\Tilde{Z})\in Sel((\mathbf{S}_d,              \Tilde{Z}))\cap \Bar{I},\\
    \exists\upsilon_d\in Sel(\mathbf{Y}_d), \enskip \gamma_d \eqd \varsigma_d, \enskip m_d(\varsigma_d) = E_O[\upsilon_d|\varsigma_d] \text{ a.s.}
\end{array}\right\}
\end{align*}
\end{lemma}

\begin{proof}
$\mathcal{H}^{EV/RA}(m,\gamma)\subseteq \Tilde{\mathcal{H}}^{EV/RA}(m,\gamma)$ since the former imposes a strict superset of conditions on $(m,\gamma)$. For the converse, fix an arbitrary $(m,\gamma)\in\Tilde{\mathcal{H}}^{EV/RA}(m,\gamma)$ and $d\in\{0,1\}$, and let $(\upsilon_d,\varsigma_d)$ be the corresponding selections in $Sel(\mathbf{Y}_d)\times \left[Sel((\mathbf{S}_d,\Tilde{Z}))\cap \Bar{I}\right]$ that generate $m_d$ and $\gamma_d$. To prove that $(m,\gamma)\in\mathcal{H}^{EV/RA}(m,\gamma)$, I show that there exist $(\upsilon'_d,\varsigma'_d)$ such that: 1) $(\varsigma'_d,\Tilde{Z})\in Sel((\mathbf{S}_d,\Tilde{Z}))\cap \Bar{I}$ and $\upsilon'_d\in Sel(\mathbf{Y}_d)$; 2) $ m_d(\varsigma_d') = E_O[\upsilon_d'|\varsigma_d']$ and  $\varsigma_d'\eqd \gamma_d$; and 3) $(\upsilon_d',\varsigma_d')\independent \Tilde{Z}$.

Let $P_{\upsilon_d',\varsigma_d'}$ be a distribution such that $\forall z\in\Tilde{\mathcal{Z}}$, $P_{\upsilon_d',\varsigma_d'}(\cdot|\varsigma_d'=s, \Tilde{Z}=z) = P_{\upsilon_d,\varsigma_d}(\cdot|\varsigma_d = s, G=O)$ $\forall s\in\mathcal{S}$, and $P_{\varsigma_d'}(\cdot|\Tilde{Z}=z) = P_{\varsigma_d}(\cdot)$. Note that these conditions fully specify $P_{\upsilon_d',\varsigma_d'}$. I first show that there exist $(\varsigma_d',\Tilde{Z})\in Sel((\mathbf{S}_d,\Tilde{Z}))\cap \Bar{I}$ and $\upsilon_d'\in Sel(\mathbf{Y}_d)$ such that $(\upsilon_d',\varsigma_d')\eqd P_{\upsilon_d',\varsigma_d'}$. I then show that $(\upsilon_d',\varsigma_d')$ fulfill conditions  2) $ m_d(\varsigma_d') = E_O[\upsilon_d'|\varsigma_d']$ and  $\varsigma_d'\eqd \gamma_d$; and 3) $(\upsilon_d',\varsigma_d')\independent \Tilde{Z}$. Recall that, as in the proof of \Cref{lem:selections_combined}, by \citet[Theorem 2.1]{artstein1983distributions} and \citet[Theorem 2.33]{molchanov2018random}, $(\upsilon_d,\varsigma_d,\Tilde{Z})\in Sel(\mathbf{Y}_d)\times Sel(\mathbf{S}_d)\times\{\Tilde{Z}\}$ if and only if $\forall B\in\mathcal{C}(\mathcal{Y}\times\mathcal{S})$ $P-$a.s.:
\begin{align}\label{eq:joint_artstein}
  P_{\upsilon_d,\varsigma_d}(B|\Tilde{Z}) \geq P\big(G=O,D=d,(Y,S)\in B|\Tilde{Z}\big) +  P\big(G=E,D=d,S\in C(B)|\Tilde{Z}\big)
\end{align}

where $C(B)=\{s:\mathcal{Y}\times\{s\}\subseteq B\}$. Since $(\upsilon_d,\varsigma_d,\Tilde{Z})\in  Sel(\mathbf{Y}_d)\times \left[Sel(\mathbf{S}_d)\times\{\Tilde{Z}\}\cap \Bar{I}\right]$, it must be that $P_{\varsigma_d}(\cdot|\Tilde{Z}) = P_{\varsigma_d}(\cdot)$ $P-$a.s. Then for any $ B\in\mathcal{C}(\mathcal{Y}\times\mathcal{S})$ $P-$a.s.:
\begin{align}\label{eq:generated_distribution_1}
    \begin{split}
        P_{\upsilon_d',\varsigma_d'}(B)=P_{\upsilon_d',\varsigma_d'}(B| \Tilde{Z})  = P_{\upsilon_d,\varsigma_d}(B|G=O)\geq  P_O((Y,S)\in B,D=d)
    \end{split}
\end{align}

where the first equality is by $P_{\upsilon_d',\varsigma_d'}(\cdot| \Tilde{Z}) = P_{\upsilon_d',\varsigma_d'}(\cdot) $ which follows by $P_{\varsigma_d'}(\cdot|\Tilde{Z}=z) = P_{\varsigma_d}(\cdot)$ and $P_{\upsilon_d',\varsigma_d'}(\cdot|\varsigma_d'=s, \Tilde{Z}=z) = P_{\upsilon_d,\varsigma_d}(\cdot|\varsigma_d = s, G=O)$ $\forall s\in\mathcal{S}$ and $\forall z\in\Tilde{\mathcal{Z}}$. The second is by $P_{\upsilon_d',\varsigma_d'}(\cdot|\varsigma_d'=s, \Tilde{Z}=z) = P_{\upsilon_d,\varsigma_d}(\cdot|\varsigma_d = s, G=O)$, $P_{\varsigma_d'}(\cdot|\Tilde{Z}=z) = P_{\varsigma_d}(\cdot)$ and $P_{\varsigma_d}(\cdot)=P_{\varsigma_d}(\cdot|\Tilde{Z})$ $P-$a.s. The inequality is by \eqref{eq:joint_artstein} fixing $\Tilde{Z}$ to the value corresponding to $G=O$. For any $B\in\mathcal{C}(\mathcal{Y} \times\mathcal{S})$ also:

\begin{align}\label{eq:generated_distribution_intermediate}
\begin{split}
    P_{\upsilon_d',\varsigma_d'}(B) &\geq  P_{\upsilon_d',\varsigma_d'}(\mathcal Y\times C(B))= P_{\varsigma_d'}(C(B))\\
    &= P_{\varsigma_d}(C(B))= P_{\varsigma_d}(C(B)|\Tilde{Z}) =   P_{\upsilon_d,\varsigma_d}(\mathcal Y\times C(B)|\Tilde{Z})\\
    &\geq  P\big(G=O,D=d,(Y,S)\in \mathcal Y\times C(B)|\Tilde{Z}\big) +  P\big(G=E,D=d,S\in C(B)|\Tilde{Z}\big)
\end{split}
\end{align}

where the first inequality is because $\mathcal Y\times C(B)\subseteq B$, the first equality is by definition of a marginal distribution, the third is by $P_{\varsigma_d'}(\cdot|\Tilde{Z})= P_{\varsigma_d}(\cdot)$, the fourth is by $P_{\varsigma_d}(\cdot|\Tilde{Z}) = P_{\varsigma_d}(\cdot)$ $P-$a.s, the fifth is by definition of a marginal distribution, and the inequality is by \eqref{eq:joint_artstein}. Since \eqref{eq:generated_distribution_intermediate} holds for almost any $\Tilde{Z}=\Tilde{z}$ and hence for almost any $Z=z$ when $G=E$:
\begin{align}\label{eq:generated_distribution_2}
\begin{split}
    P_{\upsilon_d',\varsigma_d'}(B)
    &\geq  ess\sup_Z  P_E(S\in C(B),D=d|Z).
\end{split}
\end{align}

By \eqref{eq:generated_distribution_1} and \eqref{eq:generated_distribution_2} then $\forall B\in\mathcal{C}(\mathcal{Y}\times\mathcal{S})$:
\begin{align}\label{eq:restrictions_varsigma_upsilon_prime}
P_{\upsilon_d',\varsigma_d'}(B)&\geq  \max\left\{
       P_O\big((Y,S)\in B,D=d\big),
       ess\sup_{Z} P_E\big(S\in C(B),D=d|Z\big)
   \right\}
\end{align}

Then recall that $\Tilde{I}$ is the set of random elements $(E_1,E_2,E_3)$ such that $(E_1,E_2,E_3)\in \mathcal{Y}\times\mathcal{S}\times\Tilde{\mathcal{Z}}$ and $(E_1,E_2)\independent E_3$. By \Cref{lem:selections_combined}, $\exists(\upsilon_d',\varsigma_d',\Tilde{Z})\in Sel(\mathbf{Y}_d)\times Sel(\mathbf{S}_d)\times\{\Tilde{Z}\}\cap\Tilde{I}$ such that $(\upsilon_d',\varsigma_d')\eqd P_{\upsilon_d',\varsigma_d'}$ if and only if $\forall B\in\mathcal{C}(\mathcal{Y}\times\mathcal{S})$ \eqref{eq:restrictions_varsigma_upsilon_prime} holds. Since $Sel(\mathbf{Y}_d)\times Sel(\mathbf{S}_d)\times\{\Tilde{Z}\}\cap\Tilde{I}\subseteq Sel(\mathbf{Y}_d)\times \left[Sel(\mathbf{S}_d)\times\{\Tilde{Z}\}\cap \Bar{I}\right]$, there also exist $(\varsigma_d',\Tilde{Z})\in Sel((\mathbf{S}_d,\Tilde{Z}))\cap \Bar{I}$ and $\upsilon_d'\in Sel(\mathbf{Y}_d)$ such that $(\upsilon_d',\varsigma_d')\eqd P_{\upsilon_d',\varsigma_d'}$.

Since $P_{\upsilon_d',\varsigma_d'}(\cdot|\varsigma_d'=s) = P_{\upsilon_d,\varsigma_d}(\cdot|\varsigma_d = s, G=O)$, then $ m_d(\varsigma_d') = E_O[\upsilon_d'|\varsigma_d']$ a.s. Because $P_{\varsigma_d'}(\cdot|\Tilde{Z}) = P_{\varsigma_d}(\cdot) = P_{\varsigma_d'}(\cdot) $, $\varsigma_d'\eqd \varsigma_d\eqd \gamma_d$. Finally, because $(\upsilon_d',\varsigma_d',\Tilde{Z})\in Sel(\mathbf{Y}_d)\times Sel(\mathbf{S}_d)\times\{\Tilde{Z}\}\cap\Tilde{I}$, $(\upsilon_d',\varsigma'_d)\independent \Tilde{Z}$. Therefore, if $(m,\gamma)\in \Tilde{\mathcal{H}}^{EV/RA}(m,\gamma)$, then  $(m,\gamma)\in \mathcal{H}^{EV/RA}(m,\gamma)$.
\end{proof}

\begin{lemma}\label{lem:prop_score}
 Let $\Tilde{Z}$, $\Bar{I}$, and $\mathbf{S}_d$ be defined as in \Cref{lem:removing_joint_independence}. For any $\gamma_d$ that is a distribution of a selection in $Sel((\mathbf{S}_d,\Tilde{Z}))\cap \Bar{I}$, there exists a $\gamma_d-$integrable function $\pi_{\gamma_d}$ such that for any measurable set $B\in\mathcal{B}(\mathcal{S})$:
\begin{align}
    P_O(S\in B,D=d) = \int_B \pi_{\gamma_d} d\gamma_d.
\end{align}

Then for the propensity score functional $\pi_{\gamma_d} := \frac{d P_O(S,D=d)}{d\gamma_d}$ and any $\varsigma_d'\in Sel((\mathbf{S}_d,\Tilde{Z}))\cap \Bar{I}$ with $\varsigma_d'\eqd \gamma_d'$: 
\begin{align}
P_O(D=d|\varsigma_d') = \pi_{\gamma_d'}(\varsigma_d')\text{ a.s.}
\end{align}
\end{lemma}
\begin{proof}
Fix any $\gamma_d$ such that $\exists \varsigma_d \in Sel((\mathbf{S}_d,\Tilde{Z}))\cap \Bar{I}$ and $\gamma_d\eqd \varsigma_d$. Then for any $B\in\mathcal{B}(\mathcal{S})$:

$$P_O(\varsigma_d\in B,D=d)\leq P_O(\varsigma_d\in B)= P(\varsigma_d\in B)= \gamma_d(B)$$

where the inequality is by observation. For the first equality, recall that $I$ is a set of random elements $(E_1,E_2)\in\mathcal{S}\times \Tilde{\mathcal{Z}}$, and observe that  $(\varsigma_d,\Tilde{Z})\in I$. Therefore, $\varsigma_d\independent G$, by definition of $\Tilde{Z}$. For the second equality note that $\varsigma_d\eqd \gamma_d$.

Next, note that $P_O(\varsigma_d\in B,D=d) = P_O(S\in B,D=d)$ for any measurable set $B\in\mathcal{B}(\mathcal{S})$ because $\varsigma_d\in Sel(\mathbf{S}_d)$ and $P_O(\mathbf{S}_d = \{S\},D=d)=1$.  
Therefore, $P_O(S\in B,D=d)\leq \gamma_d(B)$ for any $B\in\mathcal{B}(\mathcal{S})$. Hence, $P_O(S,D=d)$ is absolutely continuous with respect to $\gamma_d$. Then, by the Radon-Nikodym theorem there exists a measurable function $ \pi_{\gamma_d}$ such that for any measurable set $B\in\mathcal{B}(\mathcal{S})$:
\begin{align*}
    P_O(S\in B,D=d) = \int_{B} \pi_{\gamma_d}d\gamma_d
\end{align*}

and $\pi_{\gamma_d} = d P_O(S,D=d)/d\gamma_d$. 

Therefore, for any $\gamma_d'$ that is a distribution of a selection in  $Sel((\mathbf{S}_d,\Tilde{Z}))\cap \Bar{I}$, there exists $\pi_{\gamma_d'} = d P_O(S,D=d)/d\gamma_d'$  $ P_O(S\in B,D=d)  = \int_{B} \pi_{\gamma_d'}d\gamma_d'$ for any measurable set $B\in\mathcal{B}(\mathcal{S})$. Hence, also $\pi_{\gamma_d'}(s) = P_O(D=d|\varsigma'_d=s)$, $\gamma_d'-$a.e. $s\in\mathcal{S}$, which concludes the proof.
\end{proof}

\begin{lemma}\label{lem:almost_sure_implication}
Let $\mathcal{Y}$ be a compact set. If there exists $d\in\{0,1\}$ such that $V_O[Y|S,D=d]>0$ $P-$a.s., then $E_O[Y|S,D=d]\in (\inf\mathcal{Y},\sup\mathcal{Y})$ $P-$a.s.
\end{lemma}

\begin{proof}

I prove that $E_O[Y|S,D=d]<\sup\mathcal{Y}$ $P-$a.s. and $E_O[Y|S,D=d]>\inf\mathcal{Y}$ $P-$a.s follows by a symmetric argument. Since $\mathcal{Y}$ is a compact set, both $\sup\mathcal{Y}$ and $\inf\mathcal{Y}$ are finite.

By contraposition suppose that $P_O(E_O[Y|S,D=d]\geq\sup\mathcal{Y})>0$. Then by definition of $\mathcal{Y}$, $P_O(E_O[Y|S,D=d]=\sup\mathcal{Y})>0$, so there exists a Borel set $B\in \mathcal{B}(\mathcal{S})$ with $P_O(S\in B|D=d)>0$ such that $E_O[Y|S\in B,D=d] = \sup\mathcal{Y}$. Now, I show that this implies $P_O(Y=\sup\mathcal{Y}|S\in B,D=d) = 1$. Suppose not, so that $P_O(Y=\sup\mathcal{Y}|S\in B,D=d) < 1$, then:
    \begin{align}\label{eq:aS(d)erivation}
    \begin{split}
         E_O[Y|S\in B,D=d] =&E_O[Y|Y = \sup\mathcal{Y},S\in B,D=d]P_O(Y=\sup\mathcal{Y}|S\in B,D=d) +  \\
        &E_O[Y|Y \neq \sup\mathcal{Y},S\in B,D=d]P_O( Y\neq\sup\mathcal{Y}|S\in B,D=d)\\
        =&E_O[Y|Y = \sup\mathcal{Y},S\in B,D=d]P_O(Y=\sup\mathcal{Y}|S\in B,D=d) +  \\
        &E_O[Y|Y < \sup\mathcal{Y},S\in B,D=d]P_O( Y<\sup\mathcal{Y}|S\in B,D=d)\\
        = &\sup\mathcal{Y} P_O(Y=\sup\mathcal{Y}|S\in B,D=d)+\\
        &E_O[Y|Y < \sup\mathcal{Y},S\in B,D=d]P_O( Y<\sup\mathcal{Y}|S\in B,D=d)\\
        < &\sup\mathcal{Y}
        \end{split}
    \end{align}
    where the first equality is by LIE, second is by definition of $\mathcal{Y}$, third by observation, and the fourth by $E_O[Y|Y < \sup\mathcal{Y},S\in B,D=d]<\sup\mathcal{Y}$ and $P_O(Y=\sup\mathcal{Y}|S\in B,D=d) < 1$. By assumption, $E_O[Y|S\in B,D=d] = \sup \mathcal{Y}$. Then \eqref{eq:aS(d)erivation} yields a contradiction, showing that $P_O(Y=\sup\mathcal{Y}|S\in B,D=d) = 1$. But then $V_O[Y|S\in B,D=d]=0$ and $P_O(S\in B|D=d)>0$, so $P_O(V_O[Y|S,D=d]=0)>0$ which contradicts $V_O[Y|S,D=d]>0$ $P-$a.s. Thus $V_O[Y|S,D=d]>0$ $P-$a.s. implies $ E_O[Y|S,D=d]<\sup \mathcal{Y}$ $P-$a.s.

\end{proof}

\begin{proposition}\label{prop:selection_assumptions_without_experiment}
       Let Assumptions \ref{ass:rand_assign} and \ref{ass:ex_validity} hold, and assume latent unconfoundedness (LUC): $Y(d)\independent D|S(d),G=O$ for $d\in\{0,1\}$. Suppose that $P_O(S(d)\in\cdot|D\neq d)\ll P_O(S(d)\in\cdot|D= d)$. Let $\mathcal{H}^{WC}(\tau)$ be the identified set for $\tau$ when latent unconfoundedness is not maintained. 
      \begin{enumerate}[label = \roman*)]
        \item Suppose $\mathcal{Y}$ is a bounded set and that the observed data distribution $P_O(Y,S,D)$ is such that $S$ is not a perfect predictor so $V_O[Y|S,D=d]>0$ $P-$a.s. for some $d\in\{0,1\}$. Then $\mathcal{H}^{O}(\tau)\subsetneq \mathcal{H}^{WC}(\tau).$
       \item If the observed data distribution $P_O(Y,S,D)$ is such that $E_O[Y|S,D=d]$ is a trivial measurable function for all $d\in\{0,1\}$, then $\tau$ is point-identified by the observational data, and $\mathcal{H}(\tau) = \mathcal{H}^{O}(\tau)$.
       \end{enumerate}
\end{proposition}
\begin{proof}

\textit{i)} $\mathcal{Y}$ is closed by definition. Since it is bounded, it is a compact set. Then $\sup\mathcal{Y}<\infty$ and $\inf\mathcal{Y}>-\infty$. Using arguments of \citet{manski1990nonparametric}, the sharp upper bound of $\mathcal{H}^{WC}(\tau)$ is:
    \begin{align}
        \tau\leq E_O[Y(2D-1)] + \sup \mathcal{Y}P_O(D=0) - \inf \mathcal{Y}P_O(D=1) = \sup\mathcal{H}^{WC}(\tau).
    \end{align}

   Suppose that $V_O[Y|S,D=1]>0$ $P-$a.s. Fix $d=1$ and the case for $d=0$ follows by symmetric arguments. By \Cref{lem:almost_sure_implication}, $V_O[Y|S,D=d]>0$ $P-$a.s. implies $ E_O[Y|S,D=d]<\sup \mathcal{Y}$ $P-$a.s. If there exists $d\in\{0,1\}$ s.t. $V_O[Y|S,D=d]>0$ $P-$a.s., then it must be that for every Borel set $B\in \mathcal{B}(\mathcal{S})$ with $P_O(S\in B|D=d)>0$ we have $ E_O[Y|S\in B,D=d]<\sup \mathcal{Y}$. Under LUC then:
        \begin{align}\label{eq:mean_Y(d)ecomposition}
    \begin{split}
        E_O[Y(d)|D\neq d]  &= E_O[E_O[Y(d)|S(d),D\neq d]|D\neq d]\\
        &= E_O[E_O[Y(d)|S(d),D= d]|D\neq d] \\
        &= E_O[E_O[Y|S,D= d]|D\neq d]\\
        &< \sup \mathcal{Y}
    \end{split}
    \end{align}

  where the first line is by LIE, the second by LUC, the third by definition, 
and the fourth because $E_O[Y\mid S=s,D=d]$ satisfies 
$E_O[Y\mid S=s,D=d]<\sup\mathcal Y$ $P_O(S\in\cdot\mid D=d)$-a.s. by 
\Cref{lem:almost_sure_implication}, and absolute continuity of 
$P_O(S(d)\in\cdot\mid D\neq d)$ with respect to $P_O(S\in\cdot\mid D=d)$ 
implies $E_O[Y\mid S(d),D=d]<\sup\mathcal Y$ $P_O(\cdot\mid D\neq d)$-a.s. Then under LUC:
    \begin{align}\label{eq:luc_inequality}
    \begin{split}
        E[Y(d)] = E_O[Y\mathbbm{1}[D=d]]+E[Y(d)|D\neq d]P_O(D\neq d)\\
        <E_O[Y\mathbbm{1}[D=d]]+\sup\mathcal{Y}P_O(D\neq d).
    \end{split}
    \end{align}

Therefore, under LUC:
   \begin{align*}
        \tau &= E[Y(1)-Y(0)]\\
        & = E_O[YD]+E_O[Y(1)|D= 0]P_O(D=0) - E_O[Y(1-D)] - E_O[Y(0)|D= 1]P_O(D=1)\\
        &< E_O[Y(2D-1)] +\sup \mathcal{Y}P_O(D=0) - \inf \mathcal{Y}P_O(D=1) = \sup \mathcal{H}^{WC}(\tau)
    \end{align*}
where the inequality follows by \eqref{eq:luc_inequality}. Thus $\sup\mathcal{H}^{O}(\tau)<\sup\mathcal{H}^{WC}(\tau)$. So there must exist a point in $\mathcal{H}^{WC}(\tau)$ which is not contained in $\mathcal{H}^{O}(\tau)$. Conclude that $  \mathcal{H}^{O}(\tau)\subsetneq \mathcal{H}^{WC}(\tau)$.

\textit{ii)} Suppose that for every $d\in\{0,1\}$ $E_O[Y|S,D=d]$ is a trivial measurable function. Hence there exists a $y_d\in\mathcal{Y}$ such that $E_O[Y|S,D=d] = y_d $ $P-$a.s. Then:
    \begin{align}
    \begin{split}
        E_O[Y(d)|D\neq d] &= E_O[E_O[Y(d)|S(d),D\neq d]|D\neq d] \\
                 &= E_O[E_O[Y(d)|S(d),D= d]|D\neq d]     \\
                 &=E_O[E_O[Y|S,D= d]|D\neq d] \\
        &=y_d 
    \end{split}
    \end{align}
    where the final line follows since $E_O[Y|S,D=d] = y_d $ $P-$a.s. Given that $y_d $ is identified by the data, then $E_O[Y(d)]$ is identified for every $d\in\{0,1\}$, so $\tau$ is too. It is also immediate that $\mathcal{H}(\tau) = \mathcal{H}^{O}(\tau)$ since for every $d\in\{0,1\}$ and any $\gamma_d\in\mathcal{P}^{\mathcal{S}}$, we have that $E[Y(d)] = \int_{\mathcal{S}} y_d d\gamma_d(s) = y_d $. Since experimental data only affect the feasible $\gamma_d$, the result follows.
    
\end{proof}

\begin{lemma}\label{lem:rel_to_surrogacy}
    Suppose Assumption \ref{ass:rand_assign} holds. Assume that there is perfect experimental compliance so $Z=D|G=E$ $P-$a.s. and define conditions:
\begin{enumerate}[label=C.\arabic*]
    \item (Surrogacy) $Y\independent D|S,G=E$ \label{cond:surrogacy};
    \item (Comparability) $Y\independent G|S$ \label{cond:comparability}.
\end{enumerate}

Then:
\begin{enumerate}[label=\roman*)]
    \item \ref{cond:surrogacy} implies $E_E[Y(1)|S(1)=s]=E_E[Y(0)|S(0)=s]$ for all $s\in\mathcal{S}$;
    \item \ref{cond:surrogacy} and \ref{ass:ex_validity} imply $E_g[Y(1)|S(1)=s]=E_{g'}[Y(0)|S(0)=s]$ for all $s\in\mathcal{S}$ and $g,g'\in\{O,E\}$;
    \item \ref{cond:comparability} implies $E_O[Y|S=s]=E_E[Y(1)|S(1)=s]P_E(D=1|S=s)+E_E[Y(0)|S(0)=s]P_E(D=0|S=s)$ for all $s\in\mathcal{S}$;
    \item \ref{cond:comparability} and \ref{ass:ex_validity} imply $E_O[Y|S=s]=E_g[Y(1)|S(1)=s]P_E(D=1|S=s)+E_{g'}[Y(0)|S(0)=s]P_E(D=0|S=s)$ for all $s\in\mathcal{S}$ and $g,g'\in\{O,E\}$;
    \item \ref{cond:surrogacy} and \ref{cond:comparability} imply $E_O[Y|S=s] = E_E[Y(d)|S(d)=s]$ for all $s\in\mathcal{S}$;
    \item \ref{cond:surrogacy}, \ref{cond:comparability} and \ref{ass:ex_validity} imply $E_O[Y|S=s] = E_g[Y(d)|S(d)=s]$ and hence $E_g[Y(d)|S(d)=s]=\sum_{d'\in\{0,1\}}E_O[Y(d')|S(d')=s,D=d']P_O(D=d'|S=s)$, for all $s\in\mathcal{S}$ and $g\in\{O,E\}$.
 \end{enumerate}
\end{lemma}

\begin{proof}
    \textit{i)} Write for any $d\in\{0,1\}$:
            \begin{align}
        \begin{split}
       E_E[Y|S] &= E_E[Y|S,D=d] = E_E[Y(d)|S(d),D=d] = E_E[Y(d)|S(d)]
        \end{split}
    \end{align}
    where the first equality is by surrogacy, second is by definition, and third is by random assignment and perfect compliance. 

    \textit{ii)} Under Assumption \ref{ass:ex_validity}, $E_E[Y(d)|S(d)]= E[Y(d)|S(d)]$. The result then follows from \textit{i)}.

    \textit{iii)} Write:
    \begin{align}
        \begin{split}
            E_O[Y|S=s] &= E_E[Y|S=s] \\
            &= E_E[Y(1)|S(1)=s,D=1]P_E(D=1|S=s)\\
            &+E_E[Y(0)|S(0)=s,D=0]P_E(D=0|S=s)\\
        &=E_E[Y(1)|S(1)=s]P_E(D=1|S=s)+E_E[Y(0)|S(0)=s]P_E(D=0|S=s)
        \end{split}
    \end{align}
    where the first equality is by comparability, the second is by LIE and definitions of $Y$ and $S$, and the third is by random assignment and perfect compliance.
    
    \textit{iv)} Under Assumption \ref{ass:ex_validity}, $E_E[Y(d)|S(d)]= E[Y(d)|S(d)]$. The result then follows from \textit{iii)}.

    \textit{v)} Immediate from \textit{i)} and \textit{iii)}.

    \textit{vi)} Immediate from \textit{ v)} under Assumption \ref{ass:ex_validity}.

\end{proof}

\begin{lemma}\label{lem:min_max_selectors}
Let Assumptions \ref{ass:rand_assign} and \ref{ass:ex_validity} hold. Suppose that $\mathcal{S}$ is a finite set. Fix $\gamma' \in \mathcal{H}(\gamma)$. Define pointwise data bounds:
\begin{align}\label{eq:pointwise_data_bounds}
\begin{split}
    L_d(s) &:= \mu_d(s)\pi_{\gamma'_d}(s) + \inf\mathcal{Y}\left(1 - \pi_{\gamma'_d}(s)\right),\\
    U_d(s) &:= \mu_d(s)\pi_{\gamma'_d}(s) + \sup\mathcal{Y}\left(1 - \pi_{\gamma'_d}(s)\right).
\end{split}
\end{align}
\begin{enumerate}[label=\roman*)]
    \item Suppose Assumption \ref{ass:LMIV} holds. Define monotone envelopes:
    \begin{align}\label{eq:monotone_envelopes}
    \begin{split}
        m_d^{L,\gamma'}(s) &:= \sup_{s' \leq s} L_d(s'), \qquad
        m_d^{U,\gamma'}(s) := \inf_{s' \geq s} U_d(s'),
    \end{split}
    \end{align}
    where $s' \leq s$ denotes the product order. Then:
    \begin{align}
        m^{LB}_{\gamma'} := (m_0^{U,\gamma'}, m_1^{L,\gamma'}), \qquad m^{UB}_{\gamma'} := (m_0^{L,\gamma'}, m_1^{U,\gamma'})
    \end{align}
    are minimal and maximal selectors of $\mathcal{H}(m|\gamma')$ with respect to $T$.

    \item Suppose Assumption \ref{ass:TI} holds. Define:
    \begin{align}
    \begin{split}
        m^{L,\gamma'}(s) &:= \max_{d \in \{0,1\}} L_d(s), \qquad
        m^{U,\gamma'}(s) := \min_{d \in \{0,1\}} U_d(s),
    \end{split}
    \end{align}
    and:
    \begin{align}
    \begin{split}
        m^{TI,L,\gamma'}(s) &:= m^{L,\gamma'}(s)\mathbbm{1}[\gamma'_1(s) \geq \gamma'_0(s)] + m^{U,\gamma'}(s)\mathbbm{1}[\gamma'_1(s) < \gamma'_0(s)],\\
        m^{TI,U,\gamma'}(s) &:= m^{L,\gamma'}(s)\mathbbm{1}[\gamma'_1(s) < \gamma'_0(s)] + m^{U,\gamma'}(s)\mathbbm{1}[\gamma'_1(s) \geq \gamma'_0(s)].
    \end{split}
    \end{align}
    Then:
    \begin{align}
        m^{LB}_{\gamma'} := (m^{TI,L,\gamma'}, m^{TI,L,\gamma'}), \qquad m^{UB}_{\gamma'} := (m^{TI,U,\gamma'}, m^{TI,U,\gamma'})
    \end{align}
    are minimal and maximal selectors of $\mathcal{H}(m|\gamma')$ with respect to $T$.
\end{enumerate}
\end{lemma}

\begin{proof}
\textit{i)} Fix $\gamma' \in \mathcal{H}(\gamma)$ such that $\mathcal{H}(m|\gamma') \neq \emptyset$. By \Cref{thm:tractable_set}, for each $d \in \{0,1\}$, any $m \in \mathcal{H}(m|\gamma')$ satisfies $m_d(s) \in [L_d(s), U_d(s)]$ for all $s \in \mathcal{S}$. Under Assumption \ref{ass:LMIV}, $m_d$ is nondecreasing in the product order. Hence for any $s' \leq s$, $m_d(s) \geq m_d(s') \geq L_d(s')$, so $m_d(s) \geq \sup_{s' \leq s} L_d(s') = m_d^{L,\gamma'}(s)$. Similarly, for any $s' \geq s$, monotonicity implies $m_d(s) \leq m_d(s') \leq U_d(s')$, so $m_d(s) \leq \inf_{s' \geq s} U_d(s') = m_d^{U,\gamma'}(s)$.

To show $m_d^{L,\gamma'}$ and $m_d^{U,\gamma'}$ are feasible, note that both are nondecreasing in the product order by construction: if $s \leq s''$, then $\{s': s' \leq s\} \subseteq \{s': s' \leq s''\}$, so $m_d^{L,\gamma'}(s) \leq m_d^{L,\gamma'}(s'')$, and similarly $\{s': s' \geq s\} \supseteq \{s': s' \geq s''\}$, so $m_d^{U,\gamma'}(s) \leq m_d^{U,\gamma'}(s'')$. Since $\mathcal{H}(m|\gamma') \neq \emptyset$, pick any feasible $\tilde{m}_d$. Then $m_d^{L,\gamma'}(s) \leq \tilde{m}_d(s) \leq m_d^{U,\gamma'}(s)$ for all $s \in \mathcal{S}$, which implies $m_d^{L,\gamma'}(s) \leq m_d^{U,\gamma'}(s)$. Moreover, $L_d(s) \leq m_d^{L,\gamma'}(s)$ and $m_d^{U,\gamma'}(s) \leq U_d(s)$ for all $s$, so both envelopes satisfy the pointwise data restrictions. Hence $m_d^{L,\gamma'}, m_d^{U,\gamma'} \in \mathcal{M}^A$ and satisfy the constraints, so $(m_0^{U,\gamma'}, m_1^{L,\gamma'}), (m_0^{L,\gamma'}, m_1^{U,\gamma'}) \in \mathcal{H}(m|\gamma')$.

Since $T(m,\gamma') = \sum_{s \in \mathcal{S}} m_1(s)\gamma'_1(s) - \sum_{s \in \mathcal{S}} m_0(s)\gamma'_0(s)$ is nondecreasing in $m_1(s)$ and nonincreasing in $m_0(s)$ for each $s$, and $\gamma'_d(s) \geq 0$, $m^{LB}_{\gamma'}$ minimizes $T$ over $\mathcal{H}(m|\gamma')$ and $m^{UB}_{\gamma'}$ maximizes it.

\textit{ii)} Fix $\gamma' \in \mathcal{H}(\gamma)$ such that $\mathcal{H}(m|\gamma') \neq \emptyset$. Under Assumption \ref{ass:TI}, $m_1 = m_0 =: m$. By \Cref{thm:tractable_set}, $m(s) \in [L_d(s), U_d(s)]$ for both $d \in \{0,1\}$ and all $s \in \mathcal{S}$. Taking the intersection, $m(s) \in [m^{L,\gamma'}(s), m^{U,\gamma'}(s)]$ for all $s \in \mathcal{S}$. Since $m^{TI,L,\gamma'}(s), m^{TI,U,\gamma'}(s) \in [m^{L,\gamma'}(s), m^{U,\gamma'}(s)]$ for all $s\in\mathcal{S}$ by construction, it follows that $m^{LB}_{\gamma'}, m^{UB}_{\gamma'} \in \mathcal{H}(m\mid \gamma')$. Under Assumption \ref{ass:TI}:
\begin{align}
    T(m,\gamma') = \sum_{s \in \mathcal{S}} m(s) (\gamma'_1(s) - \gamma'_0(s)).
\end{align}
For each $s$, if $\gamma'_1(s) \geq \gamma'_0(s)$, then $T$ is nondecreasing in $m(s)$, so the minimum is attained at $m^{L,\gamma'}(s)$ and maximum at $m^{U,\gamma'}(s)$. If $\gamma'_1(s) < \gamma'_0(s)$, then $T$ is nonincreasing in $m(s)$, so the minimum is attained at $m^{U,\gamma'}(s)$ and maximum at $m^{L,\gamma'}(s)$. This yields $m^{TI,L,\gamma'}$ as the minimizer and $m^{TI,U,\gamma'}$ as the maximizer.
\end{proof}

\end{document}